\documentclass[aps,amssymb,amsmath,amsfonts,pra,showpacs]{revtex4}

\usepackage{graphicx}
\usepackage{bm,bbm}
\usepackage{epstopdf}
\usepackage{amsthm}
\usepackage{bbold}
\usepackage{amsmath}
\usepackage{hyperref}
\usepackage{graphicx}
\usepackage{graphics}
\usepackage{epsfig}
\usepackage{color}
\usepackage{multirow}

\newtheorem{proposition}{Proposition}
\newtheorem{lemma}{Lemma}

\newtheorem{theorem}{Theorem}
\newtheorem{corollary}{Corollary}
\newtheorem{conjecture}{Conjecture}

\renewcommand{\c}[1]{\mathcal{#1}}

\newcommand{\I}{{\rm i}}

\newcommand{\1}{\mathbbm{1}} 
\newcommand{\idty}{\1}

\DeclareMathOperator*{\tr}{Tr}

\newcommand{\<}{\langle}
\renewcommand{\>}{\rangle}

\newcommand{\beq}{\begin{equation}}
\newcommand{\eeq}{\end{equation}}
\newcommand{\ket}[1]{\ensuremath{ |{#1} \rangle}}
\newcommand{\bra}[1]{\ensuremath{\langle{#1} |}}
\renewcommand{\rho}{\varrho}

\newcommand{\eg}{e.g.\;}
\newcommand{\ie}{i.e.}

\newcommand{\finpro}{\hfill $\Box$}
\newcommand{\Span}{\operatorname{span}}

\newcommand{\braket}[2]{\langle #1 | #2 \rangle}
\newcommand{\ketbra}[2]{| #1 \rangle \langle #2 |}

\newcommand{\HS}{{\rm HS}}
\newcommand{\Bu}{{\rm Bu}}
\newcommand{\Hel}{{\rm He}}
\newcommand{\Aclas}{{\rm cq}}
\newcommand{\classQ}{{\cal CQ}}
\newcommand{\sep}{{\rm s}}
\newcommand{\onehalf}{\frac{1}{2}}
\newcommand{\states}{{\cal S}}
\newcommand{\mmax}{{\rm max}}
\newcommand{\EoF}{E_{\rm EoF}}
\newcommand{\ONB}{orthonormal basis\;}
\newcommand{\ONBs}{orthonormal bases\;}
\newcommand{\opt}{{\rm opt}}
\newcommand{\CCQ}{closest classical-quantum\;}
\newcommand{\Proofof}[1]{\noindent {\it Proof of #1. }}

\def\real{{\mathbb{R}}}

\def\complex{{\mathbb{C}}}

\newcommand{\Hh}{{\cal H}}

\newcommand{\Ll}{{\cal L}}

\newcommand{\Nn}{{\cal N}}

\newcommand{\Rr}{{\cal R}}

\newcommand{\Uu}{{\cal U}}

\begin{document}

\title{Geometric measures of quantum correlations: characterization, quantification, and comparison by distances and operations}
\author{W. Roga$^{1,2}$\footnote{wojciech.roga@gmail.com}, D. Spehner$^{3,4}$\footnote{Dominique.Spehner@ujf-grenoble.fr},
and F. Illuminati$^{1,5}$\footnote{fabrizio.illuminati@gmail.com}}
\affiliation{$^1$ Dipartimento di Ingegneria Industriale, Universit\`a degli Studi di Salerno, Via Giovanni Paolo II 132, I-84084 Fisciano (SA), Italy}
\affiliation{$^2$ Department of Physics, John Anderson Building, University of Strathclydeb, 107 Rottenrow, Glasgow, G4 0NG, UK}
\affiliation{$^3$ Univ. Grenoble Alpes and CNRS, Institut Fourier, F-38000 Grenoble, France}
\affiliation{$^4$ CNRS and univ. Grenoble Alpes, LPMMC, F-38000 Grenoble, France}
\affiliation{$^5$ INFN, Sezione di Napoli, Gruppo collegato di Salerno, I-84084 Fisciano (SA), Italy}

\date{February 29, 2016}

\begin{abstract}
We investigate and compare three distinguished geometric measures of
bipartite quantum correlations that have been recently introduced in
the literature: the geometric discord, the measurement-induced
geometric discord, and the discord of response, each one defined
according to three contractive distances on the set of quantum states, namely the trace,
Bures, and Hellinger distances. We establish a set of exact algebraic
relations and inequalities between the different
measures. In
particular, we show that the geometric discord and the discord of response
based on the Hellinger distance are easy to compute analytically for all
quantum states whenever the reference subsystem is a qubit. These two
measures  thus provide the first instance
of discords that are simultaneously fully computable, reliable (since
they satisfy all the basic Axioms that must be obeyed by a proper
measure  of quantum correlations), and operationally viable (in
terms of state distinguishability). We apply the general mathematical
structure to determine the closest classical-quantum state of a given
state and
the maximally quantum-correlated states at fixed global state purity
according to the different distances, as well as a necessary condition
for a channel to be quantumness breaking.
\end{abstract}

\pacs{03.67.Mn, 03.65.Ud, 03.65.Ta}

\maketitle

\section{Introduction} \label{sec-intro}

The characterization and quantification of quantum correlations in
composite quantum systems is of primary importance in quantum
information theory. In particular, it is a prerequisite for
understanding the origin of the quantum advantages in tasks of quantum
technology and quantum information processing. It has been recognized
in the last decade that quantum correlations may be present even in
separable mixed states. The quantum-correlated states are singled out
by a non-vanishing value of the entropic quantum discord (whose
definition and properties will be recalled in the following)~\cite{Zurek2000, Ollivier2001, Henderson2001}. States of a
bipartite system $AB$ with vanishing quantum discord with respect to
subsystem $A$ possess only classical correlations
between $A$ and $B$ and are  called classical-quantum (or
$A$-classical) states. They are of the form
\begin{equation}
\sigma_\Aclas = \sum_{i=1}^{n_A}  q_{i}\ket{\alpha_i}\bra{\alpha_i} \otimes \rho_{B | i} \; ,
\label{eq:cq}
\end{equation}
where $\{ \ket{\alpha_i }\}$ is an orthonormal basis of the Hilbert
space $\Hh_A$ of the reference subsystem $A$, $n_A$ is the dimensionality of $\Hh_A$, the set
$\{ q_i\}$ is a probability distribution (\ie, $q_i\geq 0$, $\sum_i
q_i=1$), and $\rho_{B| i}$ are arbitrary states of subsystem $B$.
The classical-quantum states (\ref{eq:cq}) form a non convex  set
$\classQ$, the convex hull of which is the set of  all separable
states. This means that there  are separable states which are not
classical-quantum. For pure states,  however, classicality is
equivalent to separability, since a pure state  is classical-quantum
if and only if it is a product state.  Therefore quantum correlations must coincide with entanglement on pure states.

The evaluation of the entropic quantum discord is a highly
nontrivial challenge, even when one restricts to the
simplest case of two qubits (see e.g. Refs.~\cite{Girolami2011,Modi_review,Huang2013}). Geometric measures of
quantum discord provide alternative ways to the entropic discord for
quantifying quantum correlations in bipartite systems~\cite{Daki'c2010, Modi2010, Giampaolo2013,
  Aaronson2013, Spehner2013, Spehner2014, Roga2014, Ciccarello2014,
  Luo2013, Spehner_review, Piani2014}. These measures offer the advantage of
easier computability. Most of them have operational interpretations
in terms of state distinguishability. On the other hand, they depend
on a specific choice of some distance $d$ on the set of quantum states
of the bipartite system $AB$.
The most common choices are: 1) the {\em trace distance} $d_{\tr}$ and
{\em Hilbert-Schmidt distance} $d_\HS$, defined respectively as
\begin{equation} \label{eq-trace_and_HS_dist}
d_{\tr}(\rho ,\sigma )\equiv \left\|\rho -\sigma \right\|_{\tr} \equiv \tr\left|\rho -\sigma \right|
\quad , \quad
d_{\HS}(\rho ,\sigma )\equiv\left\|\rho -\sigma \right\|_{\HS }\equiv
\sqrt{\tr   |\rho -\sigma |^2 }\; ,
\end{equation}
where  $\rho$ and $\sigma$ are two arbitrary states of $AB$ and
$|X|\equiv\sqrt{X^{\dagger} X}$ is the modulus of the operator
$X$; 2) the {\it Bures distance}~\cite{Bures1969,Uhlmann1976}
\begin{equation} \label{eq-Burea_dist}
d_{\Bu}(\rho ,\sigma )\equiv\big( 2 -2 \sqrt{F(\rho ,\sigma )} \big)^\onehalf
\quad , \quad F(\rho ,\sigma )\equiv \left\|
\sqrt{\sigma} \sqrt{\rho} \right\|_{\tr}^2 = \big(\tr \big(
\sqrt{\rho}\,\sigma\sqrt{\rho} \big)^\onehalf  \big)^2\;,
\end{equation}
where $F(\rho ,\sigma )$ is the Uhlmann fidelity between $\rho$
and $\sigma $; 3) the {\em quantum Hellinger distance},
called ``Hellinger distance'' for brevity in the sequel, which is defined as
\begin{equation} \label{eq-Q_Hellinger_distance}
d_{\Hel}(\rho ,\sigma )\equiv \|\sqrt{\rho }-\sqrt{\sigma }\|_\HS =
\big( 2 -2 \tr\sqrt{\rho}\sqrt{\sigma} \big)^\onehalf  \; .
\end{equation}
For each of these distances,
three major classes of geometric measures have been introduced in recent years:

\begin{itemize}
\item[I)]
The requirement that quantum correlations must vanish on the
classical-quantum states has been  exploited in Refs.~\cite{Daki'c2010,
  Luo2010, Nakano2013, Ciccarello2014, Abad2012, Streltsov2011c,
  Spehner2013, Spehner2014}
to define the {\it geometric discord}, equal to the square distance
from a given state $\rho$ of $AB$ to the set $\classQ$ of classical-quantum states:
\begin{equation}
D_{A}^G (\rho )  \equiv \min_{\sigma_\Aclas \in \classQ}  d (\rho,\sigma_\Aclas )^2 \, .
\label{def:geomdisc1}
\end{equation}

\item[II)]
The {\it measurement-induced geometric discord} is defined by
minimizing over all local projective  measurements on $A$ the square
distance between $\rho$ and the corresponding post-measurement state
in the absence of readout~\cite{Luo2010}:
\begin{equation}
D_{A}^M  (\rho )  \equiv \min_{\{ \Pi_{i}^A \} } d \bigl(\rho,
\rho_{\rm p.m.}^{ \{ \Pi_{i}^A \} } \bigr)^2 \quad , \quad
\rho_{\rm p.m.}^{\{ \Pi_{i}^A \} } = \sum_{i=1}^{n_A} \Pi_{i}^{A} \otimes \idty \,\rho\,\Pi_{i}^{A} \otimes \idty\;.
\label{def:geomdisc2}
\end{equation}
The minimum is taken over all families
$\{ \Pi_{i}^{A} \}$  of rank-one orthogonal projectors on $\Hh_A$.
The quantity $D_{A}^M  (\rho ) $ characterizes the distinguishability
between $\rho$ and the corresponding state after an arbitrary
local von Neumann measurement on $A$. Since the output of such a measurement is always a
classical-quantum state, one has $D^G_{A} ( \rho ) \leq D^M_{A}
(\rho)$ for any $\rho$. This inequality is an equality
if $d$ is the Hilbert-Schmidt distance~\cite{Luo2010}.
For the trace distance, the geometric discord and 
measurement-induced geometric discord  coincide only  if $A$ is a qubit (see Ref.~\cite{Nakano2013}, where explicit
counter-examples for higher dimensional subsystems $A$ is
also reported).  For the Bures and Hellinger distances, $D_{A}^G $ and
$D_{A}^M$ are in general different, irrespective of the space dimension  $n_A$ (see Sec.~\ref{sec-meas_ind_geo_disc} below).

\item[III)]
Imposing the fundamental requirement that quantum correlations should
be invariant under  local changes of basis, one can introduce the {\it
  discord of response},  defined as~\cite{Gharibian2012,Giampaolo2013,Roga2014}
\begin{equation}
D_{A}^R  (\rho )  \equiv \frac{1}{\Nn} \min_{U_A \in \Uu_\Lambda} d \bigl( \rho,U_A \otimes \idty \,\rho \,U_A^{\dagger}\otimes \idty \bigr)^2 \, ,
\label{def:quantumn}
\end{equation}
where the minimum is taken over all local unitary operators 
$U_A$ on $\Hh_A$ separated from the identity by the condition of having a fixed
non-degenerated spectrum $\Lambda = \{ e^{2 \I \pi j/n_A} ; j=1 ,
\ldots , n_A \} $ given by the roots  of unity (see Ref.~\cite{Roga2014} for a
thorough discussion on the choice of the spectrum). Hereafter we
denote by  $\Uu_\Lambda$ the family of such
unitaries with spectrum $\Lambda$. The normalization factor $\Nn$ in
Eq.~(\ref{def:quantumn}) is equal to $\Nn=4$ for the trace distance
and $\Nn=2$ for the Bures, Hilbert-Schmidt, and Hellinger
distances. As we shall see below, the normalization is such that
$D_{A}^R (\rho)$ has maximal value equal to unity.
The discord of response characterizes how distinguishable is the
locally unitarily perturbed state from the original
one. Alternatively, it can be seen as a measure of the sensitivity of
the state $\rho$ to local unitary perturbations. For the Hilbert-Schmidt
distance it holds $D^R_A (\rho) = 2 D^G_A (\rho)$ when $A$ is a qubit~\cite{Gharibian2012,Roga2014}.

\end{itemize}

Hereafter, we omit for simplicity the lower subscript $A$  on the discord $D$,  as we will always take $A$ as the reference
subsystem. The chosen distance is
indicated explicitly. For
instance, $D^G_{\tr}$,  $D^G_{\HS}$, $D^G_{\Bu}$, and $D^G_{\Hel}$
denote the geometric discords defined via the trace,
Hilbert-Schmidt, Bures, and Hellinger distances, respectively.
Let us note that another class of geometric measures, called the
surprisal of measurement recoverability, has been studied recently
in Ref.~\cite{Wilde15}. 
It is given in terms of the minimal distance between $\rho$ and the subset
formed by all transformations of $\rho$ under entanglement breaking
channels acting on $A$.  

Besides the fact that the discords defined in
Eqs.~(\ref{def:geomdisc1})-(\ref{def:quantumn})
are typically easier to compute than the entropic quantum discord of
Refs.~\cite{Zurek2000, Ollivier2001, Henderson2001},
one of the main appealing features of the geometric approach is that
it contains  additional information. More specifically,
let us consider a  state $\rho$ in the set $\states_{AB}$ of all quantum states of
$AB$ and a distinguished subset of $\states_{AB}$, which might
coincide with I)~the subset  $\classQ$ of classical-quantum states, or
II)~the subset formed by all  post-measurement states obtained from
$\rho$ through local von Neumann measurements on $A$, or III)~the subset formed
by all local  unitary transformations of $\rho$ with unitaries in
$\Uu_\Lambda$.  Then the  state belonging to this subset that is
closest to $\rho$ for the distance $d$  provides  some useful geometrical information about $\rho$,
which goes beyond the sole knowledge of the value of the distance
between $\rho$ and this closest state. For instance, it has been proposed in Ref.~\cite{Modi2010} to measure 
classical correlations in $\rho$ by determining the minimal distance
between a product  state and a closest classical-quantum state to $\rho$.
The geometrical information can also  be useful when considering dissipative dynamical
evolutions. For instance, one can get some insight
on the efficiency of the dynamical process in changing the amount of quantum
correlations by comparing the physical trajectory $t \mapsto \rho_t$
in  $\states_{AB}$  with the geodesics connecting
$\rho_t$ to its closest classical-quantum state(s).

Another important feature of the geometric measures, which is related
to the distinguishability of quantum states, concerns their
operational interpretations. Indeed, various instances of these
measures turn out to be valuable figures of merit in the context of
quantum technology protocols~\cite{Pirandola2011}, including  
quantum illumination~\cite{Tan2008,Farace2014}, quantum metrology and phase
estimation~\cite{Girolami2014}, quantum
refrigeration~\cite{Correa2013}, and quantum local uncertainty~\cite{Girolami2013}.
In particular, the discord of response enjoys a beautiful
operational interpretation in terms of the probability of error in
protocols of quantum reading and quantum illumination~\cite{Roga2015}.

The aim of this paper is to develop a systematic theory and
exact mathematical characterization, quantification, and comparison
between the geometric measures according to
the different definitions in Eqs.~(\ref{def:geomdisc1})-(\ref{def:quantumn})
and to the different distances introduced above.
We establish algebraic relations and inequalities holding
between them. In particular, we provide
a lower bound on the geometric discord $D^G_\Bu$ for the Bures
distance in terms of the corresponding discord for the Hellinger
distance, and show that the latter is simply related to the
Hilbert-Schmidt geometric discord for the square root of the state.
Thanks to this relation, the Hellinger geometric discord $D^G_\Hel$ is fully
computable; we illustrate this point by
giving a closed expression  for arbitrary qubit-qudit
states. For a fixed distance, we also bound
$D^M$ and $D^R$ in terms of $D^G$ and $D^M$ in terms of $D^R$. Bounds
between $D_{\tr}^R$, $D_\Bu^R$, and $D_\Hel^R$ are obtained as well. We show
that some of these bounds are tight.
In the particular case where the reference subsystem $A$ is a qubit
(the other subsystem being arbitrary), we show that
$D^G_{\tr}=D^M_{\tr}= D^R_{\tr}$ and derive an exact algebraic relation between $D^R$ and
$D^G$ holding for both the Bures and the Hellinger
distances. Remarkably, this relation has the same form for the two metrics.
We also describe the closest classical-quantum state(s) and closest
post-measurement state(s) of a given bipartite state $\rho$ and obtain the values taken by $D^G$, $D^M$, and $D^R$ on
pure states for these two metrics.

Collecting the above results, we establish that the
Hellinger  geometric discord $D^G_\Hel$ and the Hellinger discord of
response $D^R_\Hel$ provide the first two instances of measures of quantum correlations that are
fully computable, reliable -- since they satisfy all the basic Axioms
that must be obeyed by a {\em bona fide} measure of  quantum discord (which are detailed below)
--, and operationally viable in terms of distinguishability of quantum states.

A further way to compare the measures defined in Eqs.~(\ref{def:geomdisc1})-(\ref{def:quantumn}) is to study the
maximally quantum-correlated states at fixed global state purity for
the different measures. We
obtain analytical expressions for the maximal discord of response of
two-qubit states as a function of their purity for the trace and Hellinger distances and compare our results
with those found previously in the literature for the Hilbert-Schmidt and Bures metrics.

Finally, we discuss applications of the geometric measures to the problem
of quantumness breaking channels. We determine a necessary
condition for a local channel to destroy completely the quantum
correlations of any bipartite state.

Before going into the detailed presentation and discussion of our results, it is
worth recalling what we exactly mean by a {\em bona fide} measure of
quantum correlations.
Following previous
works~\cite{Modi_review,Girolami2013,Ciccarello2014,Girolami2014,Roga2014,Spehner_review},
we stipulate that such a measure  must be a
non-negative function $D$ on the set of quantum states of the
bipartite system $AB$ fulfilling the following four basic Axioms:
\begin{itemize}
\item[(i)] $D$ vanishes on classical-quantum states and only on such states;
\item[(ii)] $D$ is invariant under local unitary transformations $\rho \mapsto  U_A \otimes U_B \rho\, U_A^\dagger \otimes U_B^\dagger$
(here $U_A$ and $U_B$ are unitaries acting on subsystems $A$ and $B$, respectively);
\item[(iii)] $D$ is monotonically non-increasing under local
  Completely Positive Trace Preserving (CPTP) maps acting on subsystem $B$;
\item[(iv)] $D$ reduces to an entanglement monotone on pure states.
\end{itemize}
Axioms (i-iv) are satisfied in particular by the entropic quantum discord. Here, we point out that proper measures of quantum correlations should
also obey the following additional requirement (which is also fulfilled by the entropic discord)~\cite{Spehner_review}:
\begin{itemize}
\item[(v)] if the dimension $n_A$ of $\Hh_A$ is
  smaller or equal to the  dimension $n_B$ of the space $\Hh_B$ of $B$,
then $D (\rho)$ is maximum if and only if $\rho$ is maximally entangled, that is,
$\rho$ has maximal entanglement of formation $E_{\rm EoF} ( \rho) = \ln n_A$.
\end{itemize}
Let us note that Axioms (iii) and (iv) imply that,  when $n_A\leq n_B$,
$D$ is maximum on
maximally entangled pure states, \ie, if  $\rho$ is a maximally
entangled pure
state then  $D(\rho)=D_\mmax$~\cite{Piani2014}. This follows from the facts that
 a function $D$ on $\states_{AB}$ satisfying (iii) 
is maximal on pure states if $n_A\leq n_B$~\cite{Streltsov12} and that any pure state can be obtained from a maximally
entangled pure state via a LOCC~\cite{Spehner_review}.
Thus, if Axioms (i-iv) are satisfied, the additional requirement in
Axiom (v) is essentially that   $D(\rho)=D_\mmax$ holds {\it only}
for the maximally entangled states $\rho$. 
     
It has been shown in previous works~\cite{Roga2014,Spehner_review} that the geometric discord $D^G_\Bu$ and discord of response
$D^R_\Bu$ satisfy Axioms (i)-(iv) for the Bures
distance, and hence
are {\em bona fide} measures of quantum correlations. In
this paper, we will prove that this is also the case for the three
measures $D^G_\Hel$, $D_\Hel^M$, and $D_\Hel^R$ based on
the Hellinger distance, as well as for the Bures measurement-induced
geometric discord
$D^M_\Bu$ and trace discord of response $D_{\tr}^R$. In contrast, it is known that
$D^G_\HS= D^M_\HS$ and $D^R_\HS$ do not fulfill Axiom (iii)
because of the lack of monotonicity of the Hilbert-Schmidt distance under CPTP
maps (an explicit counter-example is given in Ref.~\cite{Piani2012} for
$D^G_\HS$ and applies to $D^R_\HS$ as well, see below).
Therefore, the use of the Hilbert-Schmidt distance in the definitions
of Eqs.~(\ref{def:geomdisc1})-(\ref{def:quantumn}) can and does lead to unphysical predictions.
Considering the distances $d_p$ associated to the $p$-norms
$\|X\|_p\equiv (\tr |X|^p)^{1/p}$, 
one has that for $p>1$, $d_p$ is not contractive under CPTP
maps~\cite{Perez-Garcia2006} (see  also Ref.~\cite{Ozawa2000} for a counter-example for $p=2$, which also
holds for any $p>1$). This is why the distances $d_p$ cannot be used
to define  measures of quantumness apart from the case
$p=1$, corresponding to the contractive trace distance, while for
$p=2$ the
non-contractive Hilbert-Schmidt  distance is well
tractable and thus used to establish
bounds on the {\em bona fide} geometric measures.

Regarding our last Axiom (v), the only result established so far
in the literature concerns the Bures geometric discord~\cite{Spehner2013,Spehner_review}.
We will demonstrate below that all the
other measures  based on the trace, Bures, and
Hellinger distances also satisfy this axiom. Our proofs are valid for arbitrary (finite) space
dimensions $n_A$ and $n_B$ of
subsystems $A$ and $B$, excepted for $D^G_\Hel$, for which they are restricted
to the special cases $n_A=2, 3$,
and for $D^M_\Hel$, $D^G_{\tr}$, and $D^M_{\tr}$, for which they are
restricted to $n_A=2$.

The paper is organized as follows.
Given its length and the wealth of mathematical
relations and bounds that we have determined, we begin by summarizing the main
results in Section~\ref{sec-main_result}. We first give general
expressions of the geometric measures for the Bures and Hellinger
distances, which are convenient starting points to compare them
(see Sec.~\ref{sec-main_results_general_expressions_Bures_Hellinger}).
We then report in some synoptic Tables the various relations and
bounds satisfied by $D^G$, $D^M$, and $D^R$ for the 
trace, Hilbert-Schmidt, Bures, and Hellinger distances (see
Sec.~\ref{sec-main_results_general_bounds}). Closed expressions for the Hellinger
geometric discord and Hellinger discord of response for
arbitrary qubit-qudit states are obtained in Sec.~\ref{sec-computability}, thereby
illustrating the computability of these two measures.
A detailed comparison of all the geometric measures in the specific
case of qubit-qudit systems (\ie, for $n_A=2$ and $n_B \geq 2$) is provided in Section~\ref{eq-results_bounds_qubit}, where
we derive from the synoptic Tables a large set of relations and
bounds. We discuss there which bounds are tight.
For the sake of completeness, we recall in
Section~\ref{sec-preliminaries} the definitions of the entropic
quantum discord and local quantum  uncertainty
(see Sec.~\ref{sec-def_discords}), some known bounds between the four aforementioned
distances (see Sec.~\ref{wdtc}), and the main arguments and results from
the literature enabling to
show that $D^G$, $D^M$, and $D^R$ are {\em bona fide} measures
of quantum correlations for the trace, Bures, and Hellinger
distances (see Sec.~\ref{sec-def_geo_discords_proper}). We also recall
in this section the link  between the Bures geometric discord and
a quantum state discrimination task~\cite{Spehner2013} (see Sec.~\ref{sec-Bures_geometric_disc}).
In Section~\ref{sec-geometric_discord} we study the geometric
discords, prove the identities and bounds reported in
Table~\ref{tab1}, and present further results for the Hellinger
geometric discord. Section~\ref{sec-meas_ind_geo_disc} is devoted to
the study of the measurement-induced geometric discord. The
results summarized in Table~\ref{tab2} are proven in this section for
the Bures and Hellinger distances. In
Section~\ref{sec-disc_of_response} we study the discord of response
and prove the nontrivial relations and bounds reported in Table~\ref{tab3}.
In  Section~\ref{mqcs} we study the maximal quantum correlations
at fixed purity
according to the different discords of response and discuss the different
orderings that they induce on quantum states. The problem of
quantumness breaking channels is addressed in
Section~\ref{qbc}. Finally, we present a
short discussion and our
conclusions in Section~\ref{som}. The four appendices report the technical
proofs of some results stated in Sections~\ref{sec-disc_of_response}-\ref{qbc}.

\section{Synopsis: Summary of main results} \label{sec-main_result}
\subsection{General expressions for the geometric measures: Hellinger
  and Bures distances}
\label{sec-main_results_general_expressions_Bures_Hellinger}

Let us first restrict our attention to the Hellinger and Bures
distances. We will show in the subsequent sections that the three
geometric measures $D^G$ (geometric discord), $D^M$
(measurement-induced geometric discord), and $D^R$ (discord of
response) are obtained by maximizing or minimizing a given trace over
all \ONBs $\{ \ket{ \alpha_i} \}$ of the reference subsystem space $\Hh_A$.
In the case of the Hellinger distance, we have
\begin{eqnarray}
\label{eq-formula_Hellinger_geo_discord}
D_\Hel^G ( \rho)
& = &
2  - 2  \max_{ \{ \ket{\alpha_i} \} } \biggl\{ \sum_{i=1}^{n_A}  {\tr}_B  \bra{\alpha_i} \sqrt{\rho} \ket{\alpha_i}^2 \biggr\}^\onehalf \; ,
\\
\label{eq-formula_Hellinger_meas_ind_discord}
D_\Hel^M ( \rho)
& =  &
2 - 2 \max_{\{ \ket{\alpha_i}\} } \sum_{i=1}^{n_A} {\tr}_B \bra{\alpha_i} \sqrt{\rho} \ket{\alpha_i} \sqrt{\bra{\alpha_i} \rho \ket{\alpha_i}} \; ,
\\
\label{eq-formula_Hellinger_discord_of_response}
D_\Hel^R ( \rho)
& = &
2 \min_{ \{ \ket{\alpha_i} \} }   \sum_{i,j=1}^{n_A} \sin^2 \Big( \frac{\pi (i-j)}{n_A} \Bigr)
 {\tr}_B  \big| \bra{\alpha_i} \sqrt{\rho} \ket{\alpha_j} \big|^2   \; .
\end{eqnarray}
The derivation of Eq.~(\ref{eq-formula_Hellinger_geo_discord}) is the
content of Theorem~\ref{eq-theo_geo_disc_Hell_mixed_states}, proved in
Section~\ref{sec-geometric_discord} below. Equation~(\ref{eq-formula_Hellinger_meas_ind_discord}) is a rather direct
consequence of the definitions and
Eq.~(\ref{eq-formula_Hellinger_discord_of_response}) is derived in
Appendix~\ref{app-proof_theo5-7}. Before providing the corresponding
expressions for the Bures distance, let us introduce the
probabilities $\eta_i$ and states $\rho_i$ depending on the  \ONB $\{ \ket{\alpha_i} \}$ defined by
\begin{equation} \label{eq-state_Q_discrimination}
\eta_i = \bra{\alpha_i} \rho_A \ket{\alpha_i}  \quad , \quad
\rho_i = \eta_i^{-1} \sqrt{\rho }\, \ketbra{\alpha_i}{\alpha_i}  \otimes \idty \; \sqrt{\rho }
\quad , \quad  i=1,\ldots, n_A \; ,
\end{equation}
where $\rho_A = \tr_B (\rho)$ is the reduced state of $A$.
It has been shown in Refs.~\cite{Spehner2013,Spehner_review} that the Bures geometric
discord is  obtained by maximizing over all $\{ \ket{\alpha_i} \}$'s the
maximal success probability  $P_{\rm S}^{\,\rm{opt\,v.N.}} ( \{ \rho_i,\eta_i \})$ to discriminate the states $\rho_i$ with prior probabilities
$\eta_i$ by means of von Neumann measurements with projectors of rank
$n_B$. The reader unfamiliar with quantum state discrimination theory
can find the definition of this success probability in
Section~\ref{sec-Bures_geometric_disc} below.
It turns out that $D^M_\Bu$ and $D^R_\Bu$ can also be expressed in
terms of $\rho_i$ and $\eta_i$. More precisely,  one has (see
Sections~\ref{sec-Bures_geometric_disc}, \ref{sec-Bures_meas_ind_geo_disc}  and~\ref{sec-Bures_Hellinger2}):
\begin{eqnarray}
\label{eq-formula_Bures_geo_discord}
D_\Bu^G ( \rho)
& = &
2  - 2  \max_{\{ \ket{\alpha_i} \} } \sqrt{P_{\rm S}^{\,\rm{opt\,v.N.}} ( \{ \rho_i,\eta_i \})} \; ,
\\
\label{eq-formula_Bures_meas_ind_discord}
D_\Bu^M ( \rho)
& =  &
2 - 2 \max_{ \{ \ket{\alpha_i} \}} \tr \sqrt{\sum_{j=1}^{n_A} \eta_j^2 \rho_j^2}  \; ,
\\
\label{eq-formula_Bures_discord_of_response}
D_\Bu^R ( \rho)
& = &
1 -  \max_{ \{ \ket{\alpha_i} \} } \tr \bigg| \sum_{j=1}^{n_A} \eta_j e^{-\I \frac{2 \pi j}{n_A}} \rho_j \bigg| \; .
\end{eqnarray}
By using these expressions, the values of $D^G$, $D^M$, and $D^R$
for a bipartite pure state can be determined explicitly in terms of the
Schmidt coefficients of this state. These values are given in
Tables~\ref{tab1}-\ref{tab3} and are all entanglement monotones. This enable us to
show that for the Bures and
Hellinger distances, $D^G$, $D^M$, and $D^R$ are {\em bona fide} measures of quantum
correlations satisfying the Axioms (i-iv) above, as detailed in Section~\ref{sec-def_geo_discords_proper}. In
Appendix~\ref{app-maximal_discod_resp} we 
show that these measures
obey Axiom (v) as well, although a proof for arbitrary space dimensions $n_A$
is still lacking in a few cases (see Tables~\ref{tab1}-\ref{tab3} for more detail). In contrast,
for  the Hilbert-Schmidt distance, $D^G$, $D^M$, and $D^R$ do not fulfill
Axiom (iii) and hence are not {\em bona fide} measures of quantum correlations.

\begin{table}[t]
\scriptsize
\begin{tabular}{|c||c|c|c|c|}

\hline
                             &   \multicolumn{4}{|c|}{Geometric discord $D^G$}
\\[1mm]
\hline
Distance                     &        Bures   &     Hellinger     &          Trace   &          Hilbert Schmidt \\
\hline
\hline
\begin{tabular}{c}  Proper measure of\\  quantum correlations \end{tabular}
                    &    \checkmark          &  \checkmark        &    proved for $n_A=2$   &    no
\\[1mm]
\hline
Satisfies  Axiom (v)
                    &  \checkmark          &  proved for $n_A=2,3$ &  proved for $n_A=2$ &
\\[1mm]
\hline
\begin{tabular}{c} Maximal value  \\ if $n_A \leq n_B$ \end{tabular}
&  $2 - 2/\sqrt{n_A}$ \hspace*{4mm} & $2 -2/\sqrt{n_A}$ \hspace*{4mm} 
& $1$ for $n_A=2$  &
\\[1mm]
\hline
Value for pure states &   $2 - 2 \sqrt{\mu_\mmax}$  &  $2 - 2 K^{-\onehalf}$  &      ?         &      $1 - K^{-1}$
\\[1mm]
\hline
\begin{tabular}{c} Relations and \\ cross inequalities \end{tabular}
                  &   \multicolumn{4}{|c|}{$2 -2 \sqrt{ 1 - D_\Hel^G(\rho) /2} \;\leq\; D_\Bu^G (\rho) \; \leq\;  D_\Hel^G(\rho)
\;= \;  2 -2 \sqrt{ 1 - D_\HS^G(\sqrt{\rho})}$}
\\[1mm]
\hline
\begin{tabular}{c} Computability  \\ for two qubits \end{tabular}
                 &   Bell-diagonal states    &    all states
                 & $\left\{ \begin{array}{l} \text{X-states} \\ \text{quantum-classical states} \end{array} \right.$
                 &   all states
\\[1mm]
\hline
\begin{tabular}{c} Closest classical-\\ quantum state \end{tabular}
                & \begin{tabular}{c} given by Eq.~(\ref{eq-again_I_was_stupid}) \\ ((\ref{eq-closest_A_clas_state_Bures}) for pure states) \end{tabular}
                &  \begin{tabular}{c} given by Eq.~(\ref{eq-CCL_Hel}) \\ ((\ref{eq-Hellinger_CCQ_state}) for pure states) \end{tabular}
 &    ?  & \begin{tabular}{c} given by
                    Eq.~(\ref{eq-closest_states_meas_ind_geo_disc_HS})
                    \\    ((\ref{eq-closest_states_meas_ind_geo_disc})
                    for pure states) \end{tabular}
\\
\hline

\end{tabular}
\caption{\label{tab1} Summary of the original results from
  Section~\ref{sec-geometric_discord}, as well as of previous
  results obtained in
  Refs.~\cite{Daki'c2010,Aaronson2013,Spehner2013,Ciccarello2014}, for
  the geometric discord with the Bures, Hellinger, trace, and
  Hilbert-Schmidt distances. Here $n_A$ denotes the Hilbert space
  dimension of the reference subsystem $A$. The quantities
  $\mu_\mmax=\max \{ \mu_i \}$ and $K= (\sum_i \mu_i^2 )^{-1}$ are,
  respectively, the maximal Schmidt coefficient and the Schmidt number
  of a pure state. The question marks ``?'' indicate unsolved  problems.}
\end{table}

\begin{table}
\scriptsize
\begin{tabular}{|c||c|c|c|c|}

\hline
                             &   \multicolumn{4}{|c|}{Measurement-induced geometric discord $D^M$}
\\[1mm]
\hline
Distance                     &          Bures     &     Hellinger &          Trace   &          Hilbert Schmidt \\
\hline
\hline
\begin{tabular}{c}  Proper measure of\\  quantum correlations \end{tabular}
                    &    \checkmark          &  \checkmark        &     \checkmark   &    no
\\[1mm]
\hline
Satisfies  Axiom (v)
                    &  \checkmark          &   for $n_A=2$ (conjecture) &  proved for $n_A=2$ &
\\[1mm]
\hline
\begin{tabular}{c} Maximal value  \\ if $n_A \leq n_B$ \end{tabular}      &  \begin{tabular}{r} $2 - 2/\sqrt{n_A}$ \hspace*{4mm} \\  \end{tabular}
  &  \begin{tabular}{r} $2 - 2/\sqrt{n_A}$ \hspace*{4mm} \\  \end{tabular} &   $(2-2/n_A)^2$  &
\\[1mm]
\hline
Value for pure states &   $2 - 2 K^{-\onehalf}$   &   $2 - 2 \sum_i \mu_i^{\frac{3}{2}}$  &    see Theorem 3.3 in~\cite{Piani2014}         &      $1 - K^{-1}$
\\[1mm]
\hline
\begin{tabular}{c} Comparison with the \\ geometric discord \end{tabular}
                  &  $D_\Bu^G \leq D^M_\Bu  \leq 2 D^G_\Bu - \onehalf ( D^G_\Bu )^2$
                  &  $D_\Hel^G \leq D^M_\Hel  \leq 2 D^G_\Hel -\onehalf ( D^G_\Hel )^2$
                  &  $\begin{cases} D_{\tr}^M = D_{\tr}^G & \text{for } n_A=2 \\  D_{\tr}^M \geq  D_{\tr}^G & \text{for } n_A>2 \end{cases}$
                  &  $D_\HS^M = D_\HS^G$
\\[1mm]
\hline
\begin{tabular}{c} Computability \\ for two qubits \end{tabular}
                 &    ?    &    ?
                 & $\left\{ \begin{array}{l} \text{X-states} \\ \text{quantum-classical states} \end{array} \right.$
                 &   all states
\\[1mm]
\hline
\begin{tabular}{c} Closest post-measu-\\ rement state \end{tabular}
                & \begin{tabular}{c} for pure states, given \\ by Eq.~(\ref{eq-closest_states_meas_ind_geo_disc})  \end{tabular}
                & \begin{tabular}{c} for pure states, given \\  by Eq.~(\ref{eq-closest_states_meas_ind_geo_disc}) \end{tabular}
                &  \begin{tabular}{c}   for pure states, given \\  by Eq.~(\ref{eq-closest_states_meas_ind_geo_disc}) \end{tabular}
                & \begin{tabular}{c} given by
                    Eq.~(\ref{eq-closest_states_meas_ind_geo_disc_HS}) 
                    \\    ((\ref{eq-closest_states_meas_ind_geo_disc})
                    for pure states) \end{tabular}
\\
\hline

\end{tabular}
\caption{\label{tab2} Summary of the original results from
  Section~\ref{sec-meas_ind_geo_disc}, as well as of previous
  results obtained in
  Refs.~\cite{Luo2010,Daki'c2010,Piani2014,Nakano2013,Ciccarello2014}, for the
  measurement-induced geometric discord with the Bures, Hellinger,
  trace, and Hilbert-Schmidt distances. The  notations are the same as the ones introduced and explained in the caption of Table~\ref{tab1}.}
\end{table}

\begin{table}
\scriptsize
\begin{tabular}{|c|l||c|c|c|c|}

\hline
\multicolumn{2}{|c||}{}                 &   \multicolumn{4}{|c|}{Discord of response $D^R$}
\\[1mm]
\hline
\multicolumn{2}{|c||}{Distance}    & Bures    &    Hellinger    &          Trace   &          Hilbert Schmidt \\
\hline
\hline
\multicolumn{2}{|c||}{\begin{tabular}{c}  Proper measure of\\  quantum correlations \end{tabular}}
                              &    \checkmark          &  \checkmark        &   \checkmark   &    no
\\[1mm]
\hline
\multicolumn{2}{|c||}{Satisfies  Axiom (v)}          &  \checkmark
     &  \checkmark   &  \checkmark  & no if $n_B \geq 2 n_A$
\\[1mm]
\hline
\multicolumn{2}{|c||}{\begin{tabular}{c} Maximal value  \\ if $n_A
    \leq n_B$ \end{tabular}}                 &   $1$     &   $1$    &
$1$   & $1$
\\[1mm]
\hline
\multicolumn{2}{|c||}{Value for pure states}
                              &  $1 - \sqrt{1- E^R}$ &   $E^R$ &      $E^R$        &      $E^R$
\\[1mm]
\hline
           & $n_A=2$         &  $D^R_{\Bu}   =2D^G_{\Bu}  -\frac{1}{2} ( D^{G}_{\Bu})^2$
                             &  $D^R_{\Hel}   =2D^G_{\Hel}  -\frac{1}{2} ( D^{G}_{\Hel})^2$
                             &  $D^R_{\tr}=  D^G_{\tr} $
                             &  $D^R_\HS = 2 D^G_\HS$
\\[2mm] \cline{2-6}
\begin{tabular}{c} Comparison  \\  with $D^G$ \end{tabular}             &  $n_A=3$         &
\multirow{2}[2]{*}{
$\begin{array}{c}  1-\sqrt{1- \frac{1}{n_An_B}  \sin^2(\frac{\pi}{n_A}) (D^G_{\Bu} )^2}
\\[1mm]  \leq D_\Bu^R \leq
\\[1mm] \sqrt{2 n_A n_B  \Bigl( 2 D^G_\Bu - \onehalf  (D^G_{\Bu} )^2
  \Bigr)}  \end{array}$
}
                             &  $D^R_{\Hel}   =  \frac{3}{2} D_{\Hel}^G   - \frac{3}{8} (D^{G}_{\Hel})^2$
                             &  \multirow{2}[2]{*}{$\begin{array}{c}
    \frac{1}{n_An_B} \sin^2 \bigl(\frac{\pi}{n_A} \bigr) D^G_{\tr}
    \\[1mm] \leq D^R_{\tr} \leq \\[1mm] n_A n_B D^G_{\tr}\end{array}$ }
                             &   $D^R_\HS = \frac{3}{2} D^G_\HS$
\\[2mm] \cline{2-2}\cline{4-4}\cline{6-6}
         &  $n_A >3$        &
                             &  $\begin{array}{c} \sin^2 \big(
  \frac{\pi}{n_A} \big) \big( 2 D_{\Hel}^G   - \frac{1}{2} (D^{G}_{\Hel} )^2 \big)
                                  \\[1mm] \leq D^R_{\Hel} \leq  2 D_{\Hel}^G   - \frac{1}{2} (D^{G}_{\Hel} )^2 \\[1mm] \end{array}$
                             &
                             &  $\begin{array}{c} 2 \sin^2 ( \frac{\pi}{n_A}) D^G_{\HS} \\ \leq D^R_\HS  \leq 2 D^G_{\HS} \end{array}$
\\[2mm]
\hline
       & $n_A = 2$       &   $D_\Bu^M \leq 2 - \sqrt{2} \sqrt{1  +  (1-D_\Bu^R)^2}$
                            & \multirow{3}[2]{*}{$\begin{array}{c} \\[2mm]
    \sin^2 ( \frac{\pi}{n_A} ) D^M_\Hel \leq D^R_\Hel \\[2mm] \end{array}$}
                            & $D_{\tr}^R = D_{\tr}^M$
                           & $D_\HS^R= 2 D^M_\HS$
\\[2mm] \cline{2-3} \cline{5-6}
\begin{tabular}{c} Comparison  \\  with $D^M$ \end{tabular}
     &  $n_A=3$         &  $D_\Bu^M \leq 2 - \frac{2}{\sqrt{3}} \sqrt{1  + 2 (1-D_\Bu^R)^2}$
                         &
                         &     ?
                         &     $D^R_\HS = \frac{3}{2} D^M_\HS$
\\[2mm] \cline{2-3} \cline{6-6}
         &  $n_A >3$     & $D^M_\Bu \leq 2 - \frac{2}{\sqrt{n_A}} (1 - D^R_\Bu )$
                         &
                         &
                         &  $\begin{array}{c} 2 \sin^2 ( \frac{\pi}{n_A}) D^M_{\HS} \\ \leq D^R_\HS  \leq 2 D^M_{\HS} \end{array}$
\\[2mm]
\hline
\multicolumn{2}{|c||}{\begin{tabular}{c} Cross inequalities \\ and relations \end{tabular}}
                  &   \multicolumn{4}{|c|}{
$D^R_\Bu \leq D^R_\Hel \leq  1 - ( 1 - D^R_\Bu )^2\quad $,
$\quad  ( D^R_\Hel )^2 \leq D^R_{\tr}
  \leq 1 - ( 1 - D^R_\Bu )^2 \quad $,   $\quad  D^R_\Hel ( \rho) = D_\HS^R( \sqrt{\rho})$}
\\[1mm]
\hline
\multicolumn{2}{|c||}{\begin{tabular}{c} Computability \\ for two qubits \end{tabular}}
                 &   Bell-diagonal states
                 &   all states
                 &   $\left\{ \begin{array}{l} \text{X-states} \\ \text{quantum-classical states} \end{array} \right.$
                 &   all states
\\
\hline

\end{tabular}
\caption{\label{tab3} Summary of the original results from
  Section~\ref{sec-disc_of_response} and
  Appendix~\ref{app-proof_theo5-7}, as well as of previous
  results obtained in Ref.~\cite{Roga2014}, for the discord of
  response with the Bures, Hellinger, trace, and Hilbert-Schmidt
  distances. Here  $E^R$ is the entanglement of
  response, see Eq.~(\ref{eq-entanglement_of_response}).
The remaining notations are the same as the ones introduced and explained in the caption of Table~\ref{tab1}.}
\end{table}

\subsection{Exact relations and bounds between the geometric
  measures: arbitrary bipartite systems} \label{sec-main_results_general_bounds}

The Tables~\ref{tab1}-\ref{tab3} summarize most of our results on the
properties of the geometric measures of quantum correlations, most notably
the relations and bounds between them which
are derived in
Theorems~\ref{theo-rel_geo_disc_Hel_HS},~\ref{theo-bounds_Bures_geo},~\ref{theo-comparison_meas_ind_geo_and_geo_discord_Hel},
and \ref{theo-comparison_meas_ind_geo_and_geo_discord_bures}-\ref{qtmeqdisc} below.
When not stated otherwise, all identities and bounds hold for
arbitrary finite
dimensions $n_A$ and $n_B$ of the Hilbert spaces $\Hh_A$ and
$\Hh_B$.
Many  bounds are non trivial and are established by using the
general expressions given in
Eqs.~(\ref{eq-formula_Hellinger_geo_discord})-(\ref{eq-formula_Bures_discord_of_response}).
In addition, we also report in Tables~\ref{tab1}-\ref{tab3} some
straightforward but important consequences
 of general inequalities between the trace,
Hilbert-Schmidt, Bures, and Hellinger distances, which
are recalled in Section~\ref{wdtc} below.
Tables~\ref{tab1}-\ref{tab3} do not contain all such
trivial bounds, so we write them explicitly here for the geometric discord:
\begin{eqnarray} \label{eq-trivial_bounds}
& & \frac{1}{n_A n_B} D_{\tr}^G ( \rho)  \leq  D_\HS^G ( \rho) \leq D_{\tr}^G ( \rho)
\\ \label{eq-trivial_boundsbis}
& & D_\Bu^G ( \rho)^2  \leq  D_\Hel^G (\rho)^2  \leq   D_{\tr}^G ( \rho)
\leq 2 g( D_{\Bu}^G (\rho ) )\;,
\end{eqnarray}
where we have introduced the function $g(d) \equiv 2 d - d^2/2$.
The same inequalities hold for the measurement-induced geometric
discord and the discord of response, except that the latter appears
multiplied by an extra normalization factor ${\cal{N}}=4$ for the
trace distance and ${\cal{N}}=2$ for the Hilbert-Schmidt, Bures, and
Hellinger distances. This is a  trivial consequence of the
normalization introduced in the definition of $D^R$, see Eq.~(\ref{def:quantumn}).

\subsection{Computability of the Hellinger geometric discord and Hellinger
  discord of response} \label{sec-computability}

Let us point out the simple expressions  found in
Tables~\ref{tab1} and \ref{tab3} for the
Hellinger geometric discord and discord of response in terms of the
corresponding measures for the Hilbert-Schmidt distance of the square
root of $\rho$,
\begin{equation} \label{eq-rel_disc_resp_HS_Hel}
D_\Hel^G (\rho) = 2 - 2 \sqrt{ 1 - D^G_\HS ( \sqrt{\rho} )}
\quad , \quad
D^R_\Hel ( \rho) = D^R_\HS ( \sqrt{\rho} ) \; .
\end{equation}
The first identity is the content of
Theorem~\ref{theo-rel_geo_disc_Hel_HS} below and the
second one is a trivial consequence of the definitions, see
Eqs.~(\ref{eq-Q_Hellinger_distance}) and~(\ref{def:quantumn}).
As a result, since the geometric measures with Hilbert-Schmidt distance
are known to be easy to compute~\cite{Daki'c2010,Luo2010,Modi_review}, so are the Hellinger geometric
discord and discord of response. We emphasize that $D_\Hel^G$ and
$D^R_\Hel$ are {\it bona fide} measures of quantum correlations
satisfying the basic Axioms (i-iv) of Sec.~\ref{sec-intro}, as opposed to the
Hilbert-Schmidt measures which do not obey the monotonicity Axiom (iii).
Hence  $D^G_\Hel$ and $D^R_\Hel$ have the appealing feature of being
at the same time  physically reliable and easy to compute.

In fact, we can do better and determine directly with the help of
Eq.~(\ref{eq-formula_Hellinger_geo_discord}) an explicit expression
for the Hellinger geometric discord  whenever $A$ is
a qubit  and $B$ is an arbitrary
system with a $n_B$-dimensional Hilbert space (qudit). Note that in this case
$D^G_\Hel$ and $D^R_\Hel$  are simply related to each
other, as well as when $A$ is a qutrit  (see Table~\ref{tab3}). Hence, 
if one is able to compute $D^G_\Hel$ then the
computability of $D^R_\Hel$ for $n_A=2, 3$ immediately follows.
In the case  $n_A=2$, let us introduce the vector $\vec{\sigma}$
formed by the three Pauli matrices acting on $A$. Similarly, let the
vector $\vec{\gamma}$ be formed by the $(n_B^2-1)$ self-adjoint operators
$\gamma_p$ acting on $B$ such that $\{ \idty/\sqrt{n_B} ,
\gamma_p/\sqrt{n_B} \}_{ p=1}^{n_B^2-1}$ is an \ONB of the Hilbert
space of all $n_B \times n_B$
matrices. This means that $\tr \gamma_p = 0$ and  $\tr \gamma_p
\gamma_q  = n_B \delta_{pq}$ for any  $p,q=1,\ldots , n_B^2-1$. The
 square root of $\rho$ can be decomposed as
\begin{equation} \label{eq-Boch_dec_square_root}
\sqrt{\rho} = \frac{1}{\sqrt{2 n_B}}
\Bigl( t_0 \idty \otimes \idty + \vec{x} \cdot  \vec{\sigma} \otimes \idty + \idty \otimes \vec{y} \cdot \vec{\gamma} + \sum_{m=1}^3\sum_{p=1}^{n_B^2-1}
t_{mp} \, \sigma_m \otimes \gamma_p \Bigr)
\end{equation}
with $t_0 \in [-1,1]$, $\vec{x} \in \real^3$, and $\vec{y} \in
\real^{n_B^2-1}$. We denote by $T$ the $3 \times
(n_B^2-1)$ complex matrix with coefficients $t_{mp}$. The condition
$\tr ( \sqrt{\rho})^2  = 1$ entails $t_0^2 + \| \vec{x} \|^2 + \|
\vec{y} \|^2 + \tr (T T^{\rm T} ) = 1$ (here $T^{\rm T}$ stands for
the  transpose of $T$). For any \ONB  $\{ \ket{\alpha_i} \}_{i=0,1}$ for qubit $A$, one finds
\begin{equation}
\sum_{i=0,1}  \tr [ \bra{\alpha_i} \sqrt{\rho} \ket{\alpha_i}^2 ] =
t_0^2 + \| \vec{y} \|^2 + \vec{u}^{\rm T} ( \vec{x} \vec{x}^{\rm T} + T
T^{\rm T} ) \vec{u} \; ,
\end{equation}
where we have introduced the unit vector 
$\vec{u}= \bra{\alpha_0} \vec{\sigma} \ket{\alpha_0} = - \bra{\alpha_1} \vec{\sigma} \ket{\alpha_1}$.
Maximizing over all such vectors and using Eq.~(\ref{eq-formula_Hellinger_geo_discord}), we have
\begin{equation} \label{eq-explicit_formula_geo_disc_2_qubits}
D_\Hel^G ( \rho) = 2 - 2 \sqrt{ t_0^2 + \| \vec{y}\|^2 + k_{\rm max} } \; \; ,
\end{equation}
where  $k_{\rm max}$ is the largest eigenvalue of the $3 \times 3$
matrix $K = \vec{x} \vec{x}^{\rm T} + T T^{\rm T}$.
Therefore, the calculation of $D_\Hel^G ( \rho) $ is straightforward
once one has determined the decomposition~(\ref{eq-Boch_dec_square_root}) of the square root of $\rho$. The
Hellinger geometric discord is thus easily computable on all qubit-qudit states.

The computability for qubit-qudit states was also noticed in
Ref.~\cite{Luo2013} for  a modified version of the Hellinger measurement-induced
geometric discord, defined as
\begin{equation} \label{eq-Luo_Fu_meas_induced_discord}
\Delta_\Hel^M ( \rho) = \min_{\{ \Pi_{i}^A \} } \Big\| \sqrt{\rho}
- \sum_{i=1}^{n_A} \Pi_{i}^{A} \otimes \idty \,\sqrt{\rho}\,\Pi_{i}^{A} \otimes
\idty\Big\|_\HS^2 \;,
\end{equation}
where the minimum is taken over all families
$\{ \Pi_{i}^{A} \}$  of rank-one orthogonal projectors on $\Hh_A$.
When  $A$ is a qubit,  $\Delta_\Hel^M(\rho) = D^R_\Hel (\rho)/2$
coincides with the Hellinger discord of response up to a factor of one
half.
In this case,  $D^R_\Hel ( \rho)$ is also equal to the local quantum
uncertainty (LQU) measure $\c U_A^\Lambda (\rho)$ for bipartite systems introduced in
Ref.~\cite{Girolami2013}
(we refer the reader to Section~\ref{sec-def_discords} for the
definition of $\c U_A^\Lambda (\rho)$). This measure was
evaluated explicitly for qubit-qudit states in Ref.~\cite{Girolami2013}. 
From this result, one finds (see Sec.~\ref{sec-def_discords})
\begin{equation} \label{eq-closed_formula_discresp_Hel}
 D_\Hel^R (\rho) = \c U_A^\Lambda ( \rho)
 = 2 - 2 (t_0^2 + \| \vec{y}\|^2 + k_{\rm max} ) \; .
\end{equation}
This expression is consistent
with Eqs.~(\ref{eq-explicit_formula_geo_disc_2_qubits}) and the
relation $D_\Hel^R ( \rho) = 2 D_\Hel^G ( \rho) - D_\Hel^G (
\rho)^2/2$ from Table~\ref{tab3}.

\subsection{Inequalities between the geometric measures: qubit-qudit
 systems} \label{eq-results_bounds_qubit}

\begin{figure}
\includegraphics[width=17.7cm]{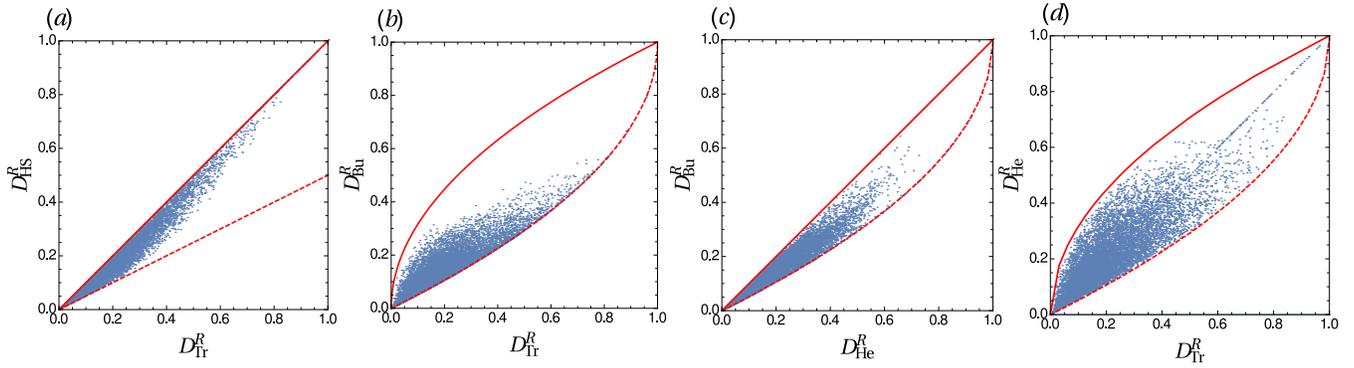}
\caption{Comparison of the discords of response based on the four
distances introduced in Sect.~\ref{sec-intro}. The points
represent $10^4$  randomly generated two-qubit  states (see main text for details).
The lines correspond to equalities in the inequalities of
Eqs.~(\ref{eq-comparison_geomeas_trace_HS}) and~(\ref{eq-bounds_geo_meas_Bu_Hel}).
(a) Hilbert-Schmidt and trace discords of response $D^R_\HS$ and
$D^R_{\tr}$. Red solid line $D^R_\HS= D^R_{\tr}$ achieved for pure
states, see Eq.~(\ref{eq-comparison_geomeas_trace_HS}); red dashed line:
$D^R_\HS= D^R_{\tr}/2$, see Eq.~(\ref{eq-comparison_geomeas_trace_HS}).
(b) Bures and trace discords of response $D^R_\Bu$ and
$D^R_{\tr}$. Red solid line: $D^R_{\Bu} = \sqrt{D^R_{\tr}}$, see Eq.~(\ref{eq-bounds_geo_meas_Bu_Hel}); red
dashed line: $D_\Bu^R = 1 - \sqrt{1-D_{\tr}^R} \;\Leftrightarrow \;
D^R_{\tr}= 2 D_\Bu^R -(D^R_\Bu)^2$ achieved for pure states, see Eq.~(\ref{eq-bounds_geo_meas_Bu_Hel}).
(c) Bures and Hellinger discords of response $D^R_\Bu$ and
$D^R_{\Hel}$.
Red solid line: $D^R_{\Bu} = D^R_{\Hel}$,  see Eq.~(\ref{eq-bounds_geo_meas_Bu_Hel}); dotted line:
$D_\Bu^R = 1 -\sqrt{1-D_\Hel^R}$, achieved for pure
states, see Eq.~(\ref{eq-bounds_geo_meas_Bu_Hel}).
(d) Hellinger and trace discords of response $D^R_\Hel$ and
$D^R_{\tr}$.
Red solid line: $D^R_\Hel = \sqrt{D^R_{\tr}}$, see Eq.~(\ref{eq-bounds_geo_meas_Bu_Hel});
red dashed line: $D_\Hel^R = 1 - \sqrt{1-D_{\tr}^R} \;\Leftrightarrow \; D_{\tr}^R= 2 D_\Hel^R - (D_\Hel^R)^2$,
see Eq.~(\ref{eq-bounds_geo_meas_Bu_Hel}).}
\label{responsegrid}
\end{figure}

In this subsection we consider the specific case of a reference
subsystem $A$ being a qubit ($n_A=2$), while subsystem $B$ is of
arbitrary space dimension. We summarized in Tables~\ref{tab1}-\ref{tab3}
an ample set of bounds holding between the geometric measures.
It is important to establish whether these bounds are tight or not. A lower or
upper bound on a measure $D$ varying in the interval $[0,D_\mmax]$ is said to be
{\it tight} if for every value $d \in [0,D_\mmax]$,
there exists a bipartite state $\rho$ such that $D(\rho)=d$ and $\rho$ saturates
the bound.

In general, proving that a bound is tight can be challenging.
To get some insight into this problem, we have generated numerically random two-qubit
states, computed $D^M$ and $D^R$ for the four distances considered in
this paper, and drawn
in Figs.~\ref{responsegrid}-\ref{RespMeas} the
corresponding distributions in the planes formed by the pairs of measures we wish
to compare.
The random two-qubit states of rank $k$ are obtained by taking the partial trace over a $k$-dimensional ancillary
system of  randomly generated pure states of the composed  (two qubits + ancilla)
system.  The ensemble of pure states is distributed according to the unitarily invariant Fubini-Study
measure on the projective space of the composed
system~\cite{Bengtsson}. The rank $k$ is chosen randomly.
  Consider the set of points formed by the values of a given pair of
measures we wish to compare for all
randomly generated states. When  the line defined by the
equality in a bound between these two measures is very close
to the border of this set, we say that the
bound is {\it almost tight}. Note that, although this gives an indication that the 
bound could be
tight, the tightness property can only be established 
by finding analytically a family of states with arbitrary 
discord fulfilling the equality. 

Our analytical bounds involve the following real increasing functions
from $[0,2-\sqrt{2}]$ onto $[0,1]$:
\begin{equation} \label{eq-def_function_g_and_f}
g (d) \equiv 2 d - \onehalf d^2
\quad , \quad
h( d) \equiv 2 g(d) - \big( g(d) \big)^2 \; .
\end{equation}
The inverse of $g$ is $g^{-1} (d) = 2 - 2 \sqrt{1-d/2}$.


{\it Comparison of the geometric measures for the trace and Hilbert Schmidt
distances}. From  Tables~\ref{tab1}-\ref{tab3} we get
\begin{equation} \label{eq-comparison_geomeas_trace_HS}
D^G_{\tr}= D_{\tr}^M=D^R_{\tr} \quad , \quad D^G_\HS = D_\HS^M = \onehalf  D^R_\HS \quad , \quad
\frac{1}{n_B} D^R_{\tr}  \leq  D^R_\HS \leq D^R_{\tr} \; .
\end{equation}
The equality between $D^G$ and $D^M$  for the
trace and Hilbert-Schmidt distances is already known,  see
Refs.~\cite{Nakano2013} and~\cite{Luo2010}. The relations
$D^G_{\tr}=D^R_{\tr}$ and $D^G_\HS= D^R_\HS/2$ are proven
in Appendix~\ref{app-proof_theo5-7}.
The first inequality in Eq.~(\ref{eq-comparison_geomeas_trace_HS})  is a trivial bound analog to
Eq.~(\ref{eq-trivial_bounds}). The second one is the content
of Theorem~\ref{qtmeqdisc} proven in
Appendix~\ref{app-proof_theo5-7}. 

The numerical results   displayed
in Fig.~\ref{responsegrid}(a) indicate that the first
inequality is unlikely to be tight in the whole range $[0,1]$,
although it is
almost optimal for weakly discordant states.
The second inequality is 
saturated for pure states, as can be checked by using the values reported in
Table~\ref{tab3} for such states. Hence this inequality is tight. 


{\it Comparison of the three geometric measures based on the Bures
  distance}. The discords $ D^G_\Bu$, $D_\Bu^M$, and $D^R_\Bu$ are ordered as follows:
\begin{equation} \label{eq-bounds_Bu_geodisc_measdisc}
 D^G_\Bu \leq D_\Bu^M \leq g^{-1} \circ h ( D^G_\Bu )
  \leq D^R_\Bu = g ( D^G_\Bu ) \leq g ( D_\Bu^M )
\quad , \quad  D^G_\Bu  \leq  (2 -\sqrt{2}) D^R_\Bu\;.
\end{equation}
The first inequality is a straightforward consequence of the
definitions of $D^G$ and $D^M$ and   is almost tight for two qubits
(see below). The second one is established by using
$D_\Bu^R = g ( D_\Bu^G)$ and $D_\Bu^M \leq 2 - (2  + 2 ( 1 -D_\Bu^R)^2)^{1/2}$ from Table~\ref{tab3}.
It is saturated for pure states (see
Theorem~\ref{theo-upper_bound_D_Bu^M_intermsof_D_Bu^R} below). The
third inequality comes from $g^{-1} \circ h (d) \leq g (d)$ for any $d \in
[0,2-\sqrt{2}]$. It is not tight (it is an equality for
classical-quantum states only). The fourth inequality follows from
the first one and the monotonicity of the function $g$. We find
numerically that this inequality is almost tight for
  two-qubit states, see Fig.~\ref{RespMeas}(a). This indicates that the same is
true for the first inequality.

The last inequality in Eq. (\ref{eq-bounds_Bu_geodisc_measdisc}) can be proven by exploiting
$D_\Bu^R = g ( D_\Bu^G)$, the bound
$d \leq (2- \sqrt{2}) g(d)$ for any $d \in [0 , 2 - \sqrt{2}]$,
and the fact that the highest value of $D_\Bu^G$ is $2 - \sqrt{2}$ when $A$ is a qubit~\cite{Spehner2013}.
This inequality  is not tight (it is saturated for
classical-quantum and maximally entangled states only).
It has been conjectured in
Ref.~\cite{Roga2014} to hold for Bell-diagonal two-qubit states by
relying 
on numerical investigations with randomly generated states. The above
argument provide an analytical proof valid for arbitrary
states of a bipartite system with $n_A=2$ and $n_B \geq 2$. 

Let
us also point out that closed expressions for
$D^G_{\Bu} ( \rho) $ and $D^R_\Bu (\rho)$  have been 
determined for Bell-diagonal states $\rho$ in
Refs.~\cite{Aaronson2013,Spehner2014} and  Ref.~\cite{Roga2014},
respectively. It is
straightforward to verify that these results are
related to  each other by specializing $D_\Bu^R = g ( D_\Bu^G)$ to the case
of Bell-diagonal states.


{\it Comparison of the geometric measures based on the
  Hellinger distance}. Similarly, one has 
\begin{equation}  \label{eq-bounds_Hel_geodisc_measdisc}
D^G_\Hel \leq D_\Hel^M \leq D^R_\Hel = g ( D^G_\Hel ) \leq g( D_\Hel^M )
\quad , \quad D^G_\Hel \leq  (2 -\sqrt{2} ) D^R_\Hel\;.
\end{equation}
Here, the second inequality  is established by using
$D_\Hel^R = g ( D_\Hel^G)$ from Table~\ref{tab3} and
$D_\Hel^M \leq g ( D_\Hel^G)$ from  Table~\ref{tab2}.
This inequality does not seem
to be tight according to Fig.~\ref{RespMeas}(b). 
In contrast, the third inequality (and thus also the first inequality)  
is almost tight for   two-qubit states.


{\it Comparison of the discords of response for the Bures,
Hellinger, and trace distances}. One has
\begin{equation} \label{eq-bounds_geo_meas_Bu_Hel}
1 - \sqrt{1-D_\Hel^R} \leq
D_\Bu^R \leq D_\Hel^R  \leq \sqrt{D_{\tr}^R}
\leq \sqrt{2 D_\Bu^R - (  D_\Bu^R)^2}
\leq \sqrt{2 D_\Hel^R - (  D_\Hel^R)^2} \; .
\end{equation}
The first inequality is equivalent to the bound
$D_\Hel^R \leq 1-(1-D_\Bu^R)^2$ from Table~\ref{tab3}. It is
tight and saturated for pure states (see
Theorem~\ref{prop_bound_D_Bu^R_D_Hel^R} below). The other
inequalities are trivial consequences of general bounds between the
Bures, Hellinger, and trace distances, see
Eq.~(\ref{eq-trivial_boundsbis}).
In particular, the last inequality follows from the second one and the monotonicity of $d \mapsto 2d - d^2$ on $[0,1]$.
 
 According to our numerical results in
Figs.~\ref{responsegrid}(c) and~\ref{responsegrid}(d), 
the second and third inequalities in
Eq.~(\ref{eq-bounds_geo_meas_Bu_Hel}) are almost tight for two qubits. The
fourth inequality is saturated for pure states (see Sec.~\ref{wdtc}
below). 
One sees on Fig.~\ref{responsegrid}(d) that the bound $D_{\tr}^R \leq
2 D_\Hel^R - (D_\Hel^R)^2$ is almost tight. This seems to indicate
that there exists
a family of two-qubit states saturating the second and fourth
inequalities in Eq.~(\ref{eq-bounds_geo_meas_Bu_Hel}). Such states
cannot be pure since $D_\Bu^R < D_\Hel^R$ for pure states when
$D_\Bu^R \not= 0,1$ (see Table~\ref{tab3}). 


{\it Comparison of the (measurement-induced) geometric discords for the Bures and
Hellinger distances}. We find
\begin{equation} \label{eq-bounds_geo_disc_Bu_Hel}
D_\Bu^G \leq D_\Hel^G  \leq g^{-1} \circ h ( D_\Bu^G )
\quad , \quad D_\Bu^M \leq D_\Hel^M \;.
\end{equation}
The first and last inequalities are as in Eq.~(\ref{eq-trivial_boundsbis}). The
second one is a consequence of
$D_\Bu^R= g( D_\Bu^G)$, $D_\Hel^R= g( D_\Hel^G)$, and
$D_\Hel^R \leq 1- (1-D_\Bu^R)^2$ from  Table~\ref{tab3}.
The latter bound being saturated for pure states (see
Theorem~\ref{prop_bound_D_Bu^R_D_Hel^R} below), the same is true for the
second inequality,
which is tight.
Note the similarities between the lower and upper bounds on
 $D_\Bu^M$ and $D_\Hel^G$ in
Eqs.~(\ref{eq-bounds_Bu_geodisc_measdisc})
and~(\ref{eq-bounds_geo_disc_Bu_Hel}).
 
According to the numerical results shown in
Fig.~\ref{measurementgrid}(b), the last inequality on the measurement-induced discords is almost
tight for two qubits.


{\it Comparison of the
geometric and measurement-induced geometric discords for the trace,
Bures, and Hellinger distances}.
Finally, we obtain a set of inequalities enabling to compare $D^G$ and $D^M$:
\begin{equation}
\label{eq-relations_trace_and_HS}
D^G_{\tr} = D^M_{\tr} \leq h( D_\Bu^G )
\leq \min \{ h ( D_\Hel^G ) \, ,\, h ( D_\Bu^M) \} \leq h ( D_\Hel^M )
\quad, \quad
D_\Bu^G \leq D_\Hel^G \leq ( 2 -\sqrt{2}) \sqrt{ D_{\tr}^M} \; .
\end{equation}
The first inequality follows by combining the relations
$D_{\tr}^G =D_{\tr}^M = D_{\tr}^R$ and
$D_\Bu^R = g( D_\Bu^G)$ with the fourth bound in Eq.~(\ref{eq-bounds_geo_meas_Bu_Hel}).
As the latter bound, this inequality is saturated for pure states.
The second and third inequalities in Eq.~(\ref{eq-relations_trace_and_HS}) are straightforward consequences of the  monotonicity of the function $h$
and of the trivial bounds $D_\Bu^G \leq D_\Hel^G $,
$ D_\Bu^G\leq D_\Bu^M \leq D_\Hel^M$, and $D_\Hel^G \leq D_\Hel^M$.
The fifth inequality follows from the last
bound in Eq.~(\ref{eq-bounds_Hel_geodisc_measdisc}) and the
third  bound in Eq.~(\ref{eq-bounds_geo_meas_Bu_Hel}).

The numerical
results shown in panels (a) and (c) of
Fig.~\ref{measurementgrid} indicate that there exists a family of
mixed two-qubit states which nearly saturate the first,
second, and third inequalities in
Eq.~(\ref{eq-relations_trace_and_HS}), \ie, such that
$D_{\tr}^M \simeq h ( D_\Bu^G ) \simeq h ( D_\Hel^G ) \simeq  h (D_\Bu^M) \simeq h ( D_\Hel^M )$. 
This provides a numerical hint that
these three inequalities could be tight. Moreover, since $h$ is an increasing
function, this also tells us that the
first bounds in
Eqs.~(\ref{eq-bounds_Bu_geodisc_measdisc}),~(\ref{eq-bounds_Hel_geodisc_measdisc}),
and~(\ref{eq-bounds_geo_disc_Bu_Hel}) and the last bound in
Eq.~(\ref{eq-bounds_geo_disc_Bu_Hel}) (trivial bounds) are    nearly saturated
by this family of two-qubit states. This is in agreement with
the conclusions drawn above, namely, the saturation
of the first bounds in Eqs.~(\ref{eq-bounds_Bu_geodisc_measdisc}) and~(\ref{eq-bounds_Hel_geodisc_measdisc}) inferred respectively from
Figs.~\ref{RespMeas}(a) and~\ref{RespMeas}(b), and the saturation
of the last bound in Eq.~(\ref{eq-bounds_geo_disc_Bu_Hel}) inferred from
Fig.~\ref{measurementgrid}(b).


\begin{conjecture}
{\rm 
From the result of Fig.~\ref{measurementgrid}(c) we conjecture that the
following inequality, which is stronger than the fifth inequality in Eq.~(\ref{eq-relations_trace_and_HS}), holds for two-qubit systems:}
\begin{equation} \label{eq-conjecture_bound}
D_\Hel^M \leq ( 2 - \sqrt{2} ) \sqrt{D_{\tr}^M} \; .
\end{equation}
\end{conjecture}

However, so far, we could not find an analytical proof of Eq.~(\ref{eq-conjecture_bound}).

\begin{figure}
\includegraphics[height=5.5cm]{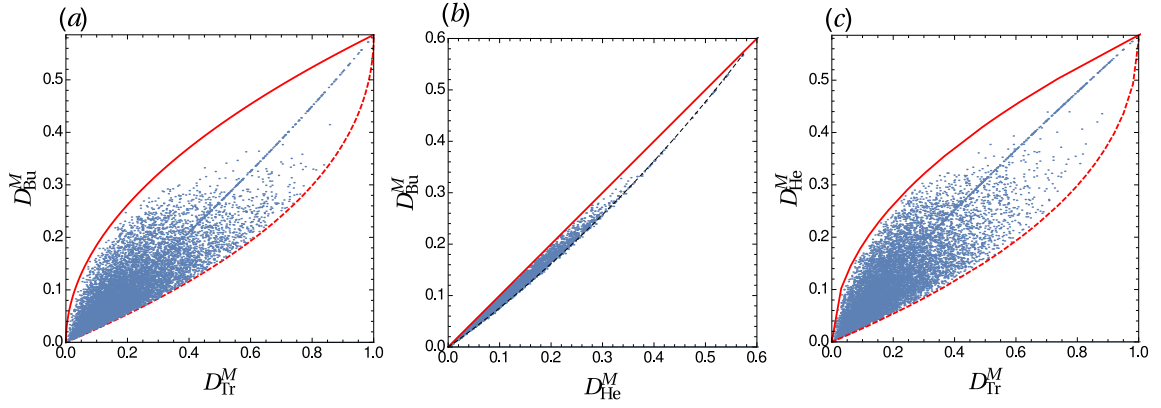}
\caption{
Comparison of the measurement-induced geometric discords based on the four
  distances introduced in Sect.~\ref{sec-intro}.  The points represent
  $10^4$ randomly generated two-qubit states (see main text
  for details).
The lines correspond to equalities in the inequalities of
Eqs.~(\ref{eq-bounds_geo_disc_Bu_Hel})-(\ref{eq-conjecture_bound}).
(a) Bures and trace measurement-induced geometric discords $D^M_\Bu$
and $D^M_{\tr}$. Red solid line: $D^M_\Bu =
(2-\sqrt{2})\sqrt{D^M_{\tr}}$, see
Eqs.~(\ref{eq-bounds_geo_disc_Bu_Hel}) and~(\ref{eq-conjecture_bound}); red dashed line: $D_\Bu^M = h^{-1} (D_{\tr}^M )$, see Eq.~(\ref{eq-relations_trace_and_HS}).
(b) Bures and Hellinger measurement-induced geometric discords
$D^M_\Bu$ and $D^M_{\Hel}$. Red solid line:
$D^M_\Bu = D^M_{\Hel}$, see Eq.~(\ref{eq-bounds_geo_disc_Bu_Hel});
blue dashed line: relation $D_\Hel^M = 2 - [(1+ \sqrt{1-g(D_\Bu^M)})^{3/2}+(1- \sqrt{1-g( D_\Bu^M)})^{3/2}]/\sqrt{2}$
satisfied by pure states.
(c) Hellinger and
trace measurement-induced geometric discords $D^M_\Hel$ and $D^M_{\tr}$. Red solid line:
$D^M_\Hel = (2-\sqrt{2})\sqrt{D^M_{\tr}}$, corresponding to the conjectured
upper bound in Eq.~(\ref{eq-conjecture_bound});
red dashed line: $D_{\Hel}^M = h^{-1}( D_{\tr}^M)$, see Eq.~(\ref{eq-relations_trace_and_HS}).}
\label{measurementgrid}
\end{figure}

\begin{figure}
\includegraphics[height=6cm]{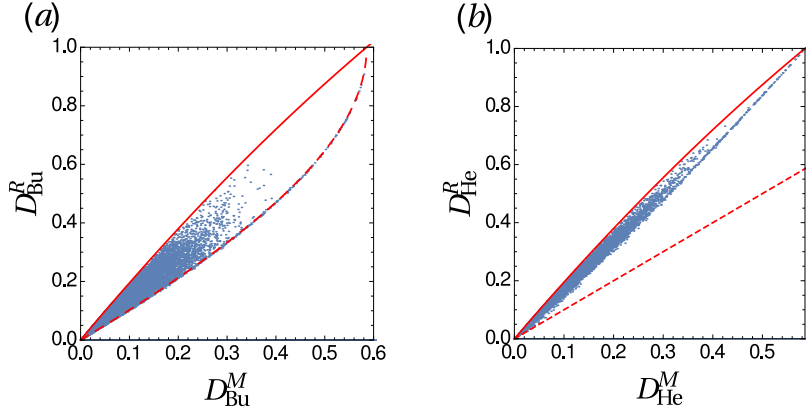}
\caption{
Comparison of the discords of response and measurement-induced geometric discords based on the same distances.
The points represent $10^4$ randomly generated two-qubit states 
(see main text for details).
The lines correspond to equalities in the inequalities of
Eqs.~(\ref{eq-bounds_Bu_geodisc_measdisc}) and~(\ref{eq-bounds_Hel_geodisc_measdisc}).
(a) Bures discord of response $D^R_\Bu$ and measurement-induced geometric discord $D^M_{\Bu}$.
Red solid line:
$D^R_{\Bu} = g( D^M_{\Bu})$, see  Eq.~(\ref{eq-bounds_Bu_geodisc_measdisc}); red dashed line:
$D_\Bu^R = g \circ h^{-1} \circ g ( D_\Bu^M) =1 - \sqrt{1-g ( D_\Bu^M  )}$,
corresponding to the saturation achieved for pure states of the second
inequality in Eq.~(\ref{eq-bounds_Bu_geodisc_measdisc}) with $D_\Bu^G=g^{-1} ( D_\Bu^R)$.
(b) Hellinger discord of response $D^R_\Hel$ and measurement-induced
geometric discord $D^M_{\Hel}$. Red solid line: $D_\Hel^R = g (D_\Hel^M)$, see Eq.~(\ref{eq-bounds_Hel_geodisc_measdisc});
red dashed line: $D^R_\Hel = D^M_{\Hel}$, see Eq.~(\ref{eq-bounds_Hel_geodisc_measdisc}).}
\label{RespMeas}
\end{figure}

\section{Review of previous results} \label{sec-preliminaries}

In this section we review some results already known in the literature that will be used later
on in the paper. We begin by recalling the definitions of the entropic
quantum discord and the local quantum uncertainty (see Sec.~\ref{sec-def_discords}).
Next, we discuss some important properties of the four distances of
interest on the set of quantum states, in particular some bounds between them (see Sec.~\ref{wdtc}). Finally, we
summarize the main arguments showing that $D^G$, $D^M$, and $D^R$  are
{\em bona fide} measures of quantum correlations
(see Sec.~\ref{sec-def_geo_discords_proper}) and briefly review the
results of Ref.~\cite{Spehner2013} on the Bures geometric discord
(see Sec.~\ref{sec-Bures_geometric_disc}).

\subsection{The entropy-based quantum discord and local quantum uncertainty} \label{sec-def_discords}

We recall in this subsection the definitions of the entropic
quantum discord introduced by Ollivier and
Zurek~\cite{Zurek2000,Ollivier2001} and by Henderson and
Vedral~\cite{Henderson2001} and of the LQU introduced by Girolami,
Tufarelli, and Adesso~\cite{Girolami2013}.  

Consider a bipartite quantum system
composed of subsystems $A$ and $B$, in the state $\rho \equiv
\rho_{AB}$. The total correlations between the two parties are
characterized by the  mutual information
\begin{equation}
I_{A:B} ( \rho) \equiv S(\rho_B) + S(\rho_A) - S(\rho ) \; ,
\label{mutual1}
\end{equation}
where the information (ignorance) about the state of $AB$
is given by the von Neumann entropy $S(\rho)\equiv -\tr \rho \log\rho $, and similarly for subsystems $A$ and $B$. In
classical information theory, the mutual information is equal to the
difference between the Shannon entropy of $B$ and the conditional
entropy of $B$ conditioned  on $A$. In the quantum setting, the corresponding quantity is
\begin{equation}
J_{A:B}^{\{\Pi^{A}_i \}} ( \rho) \equiv  S(\rho_B)-\sum_ip_i S(\rho_{B|i} ) \; ,
\end{equation}
where $p_i=\tr ( \rho\, \Pi_{i}^{A}\otimes \idty  )$ is the probability
of outcome $i$ in a local von Neumann measurement  on $A$ defined by
the orthogonal projectors  $\Pi_{i}^{A}$, and $\rho_{B |i}= p_i^{-1}
\tr_A( \rho \, \Pi_{i}^{A} \otimes \idty)$  is the corresponding
conditional post-measurement state of $B$.
It turns out that $I_{A:B}$ and $J_{A:B}$ are
in general not equal for quantum systems. Moreover, $I_{A:B}$ is never smaller than $J_{A:B}$, whatever
the local measurement. Therefore, one can define the {\em entropic quantum discord} as~\cite{Zurek2000,Ollivier2001,Henderson2001}
\begin{equation}
D_{A}^{\rm \,ent} ( \rho ) \equiv  I_{A:B} ( \rho)  - \max_{\{\Pi_{i}^{A} \}} J_{A:B}^{\{\Pi^{A}_i \}} (\rho)  \geq 0 \; .
\label{def:discord}
\end{equation}
This quantity is interpreted as a quantifier of the non-classical
features of the bipartite state $\rho$. Note that it is not
symmetric under the exchange of the two parties. One defines the
quantum discord $D_{B}^{\rm \,ent} ( \rho )$  analogously, by considering local
measurements on subsystem $B$. The two entropic discords $D_A^{\rm \,ent}$ and
$D_B^{\rm \,ent}$ give the amount of bipartite mutual
information   that cannot be retrieved by measuring one of the
subsystems. Actually,  one has
\begin{equation} \label{def:discord_bis}
D_{A}^{\rm \,ent} ( \rho ) = \min_{\{\Pi^{A}_i \}} \bigl\{ I_{A:B} ( \rho) - I_{A:B} ( \rho_{\rm p.m.}^{\{ \Pi_{i}^A \} } )\bigr\} \; ,
\end{equation}
where $\rho_{\rm p.m.}^{\{ \Pi_{i}^A \} }$ is the post-measurement
state after a local measurement on $A$ with
rank-one orthogonal projectors $\Pi_{i}^A$, see
Eq.~(\ref{def:geomdisc2}), and the minimum is taken over the set of all such measurements.

It can be shown that the entropic discord $D_A^{\rm \,ent} $ obeys Axioms
(i-v) stated in the Introduction, so that it is a proper measure of
quantum correlations (see Ref.~\cite{Spehner_review} for a thorough
review). For instance, condition (iv) is fulfilled since $D_A^{\rm \,ent} $
coincides for pure states with the entanglement of formation,
an entanglement monotone that reduces  on pure states to the von Neumann
entropy of the reduced states~\cite{Bennett96}. For mixed
states, the entropic discord captures quantum correlations different
from entanglement. Indeed, most separable (\ie, unentangled) mixed states
have non-vanishing discords. Moreover, a nonzero discord might
be responsible for the improvement (quantum speed-up) of certain quantum
algorithms with respect to their classical analogs~\cite{Datta2008},
although this claim is still debated (see e.g. the survey article in Ref.~\cite{Modi_review}). On the other hand, the analytical
evaluation of the entropic discord remains a formidably
challenging task, even for two-qubit states, because of the
difficulty of the optimization problem over the quantum measurements.
In this respect, distance-based measures of quantum correlations are usually less challenging
to evaluate, as it has been illustrated in
Sec.~\ref{sec-computability}. Moreover, they often have
operational interpretations in terms of state and channel discrimination.

A different measure of quantum correlations called the {\it local quantum
uncertainty} (LQU) was introduced in Ref.~\cite{Girolami2013}.
It is defined as follows.
One says that a quantum observable is measured in a
given state $\rho$ without quantum uncertainty if the measurement does not
disturb  $\rho$. The degree of perturbation is quantified
by the {\em skew information}~\cite{Wigner1963,Luo2003}
\begin{equation}
\c L(\rho,K)\equiv-\frac{1}{2}\tr\left([\sqrt{\rho},K]^2\right) \; ,
\label{skewinform}
\end{equation}
where $K$ is
the measured observable and $[X,Y]\equiv XY-YX$ denotes the
commutator. The skew information is considered a proper measure of the
measurement uncertainty because of the following properties:
it is non-negative, it vanishes if and only if the operators commute,
and it is convex  in $\rho$. Moreover, it is bounded by
the variance of the  measured observable:
$\c L(\rho,K) \leq  \tr K^2 \rho - (\tr K\rho )^2$, with
equality holding for pure states~\cite{Luo2003} (see also
Ref.~\cite{Audenaert2008}).
In particular, in protocols of parameter
estimation the skew information
is related to the quantum Fisher information and to the Cram{\'e}r-Rao bound~\cite{Girolami2014, Luo2003}.
Applying the skew information to local observables $K_{A}^{\Lambda}$ acting on subsystem $A$
with spectrum $\Lambda$ and
minimizing over all such self-adjoint operators,
one defines the LQU as~\cite{Girolami2013}:
\begin{equation}
\c U_A^{\Lambda}(\rho )\equiv \min_{K^{\Lambda}_A}
\c L(\rho ,K^{\Lambda}_A \otimes \idty ) \; ,
\label{locquun}
\end{equation}
where the spectrum $\Lambda$  is fixed and non-degenerate in order to
exclude the identity operator.
If $A$ is a qubit then the dependence on $\Lambda$ of  $\c U_A^{\Lambda}(\rho )$ reduces to a multiplication by a constant
factor~\cite{Girolami2013}. Therefore, one can restrict the definition
of the LQU to the case where $\Lambda$ is the
harmonic spectrum: $\Lambda=\{1,-1\}$.  All Hermitian $2\times 2$ matrices with harmonic spectrum are
unitary matrices. It follows from this observation that if $A$ is a
qubit then the LQU
is equal  to the Hellinger discord of  response $D^R_\Hel$.

The LQU was evaluated explicitly
for arbitrary qubit-qudit states  in Ref.~\cite{Girolami2013}. It was found in
this reference that
$1-\c U_A^{\{1,-1\}}(\rho ) = \lambda_\mmax ( W)$ is the highest eigenvalue of the $3\times 3$
matrix with elements $W_{ij}= \tr \sqrt{\rho}\,
\sigma_i \otimes \idty \sqrt{\rho}\, \sigma_j \otimes \idty$. A simple
calculation shows that $\c U_A^{\{1,-1\}} (\rho)$ is then given by the
right-hand side of Eq. (\ref{eq-closed_formula_discresp_Hel}).

\subsection{Properties and comparison of the trace, Hilbert-Schmidt,
  Bures, and Hellinger distances}\label{wdtc}

In this subsection we recall some known facts about the four distances defined in
Eqs.~(\ref{eq-trace_and_HS_dist})-(\ref{eq-Q_Hellinger_distance}).
In quantum information theory, well-behaved
distances $d$ on the set of quantum states must be contractive
under Completely Positive Trace Preserving (CPTP) maps, that is, they
must be such that  $d (\Phi(\rho ),\Phi(\sigma ) ) \leq d (\rho ,
\sigma )$ for any states $\rho$ and $\sigma$ and any CPTP map $\Phi$
acting on the algebra of operators from
the  system Hilbert space ${\cal H}$ to a (possibly different)
space ${\cal H}'$ (see Sec.~\ref{qbc}
for a definition of CPTP maps)~\cite{Nielsen,Spehner_review}.
Notice also that a contractive distance $d$ is in particular unitarily invariant
(\ie, $d ( U \rho \,U^\dagger ,  U \sigma \,U^\dagger)  = d( \rho ,\sigma )$
for any unitary operator $U$).

Let us first focus on the trace and Hilbert-Schmidt distances
$d_{\tr}$ and  $d_\HS$. The former is contractive under CPTP
maps, but this is not the case for the latter (more generally, as already mentioned,
the distances associated to the $p$-norms $\| X\|_p = ( \tr |X|^p )^{1/p}$
are not contractive for
$p>1$)~\cite{Perez-Garcia2006,Nielsen,Bengtsson,Spehner_review}.
The trace distance achieves an operational meaning in the light of the
Helstrom formula $P_{\rm err} (\rho ,\sigma) =1/2 -d_{\tr}(\rho,\sigma)/4$ for the minimum probability of error in
discriminating two equiprobable quantum states $\rho$ and $\sigma$.
The Hilbert-Schmidt distance can be used to bound from above and below
the trace distance as follows:
\begin{equation}
\frac{1}{\sqrt{r}} d_{\tr}(\rho ,\sigma ) \leq  d_{\HS}(\rho ,\sigma ) \leq d_{\tr}(\rho ,\sigma ) \; ,
\label{univineq}
\end{equation}
where $r$ is the rank of $\rho -\sigma $.
The inequalities on the geometric discords given in Eq.~(\ref{eq-trivial_bounds}) are trivial
consequences of Eq.~(\ref{univineq}). 

The Bures distance $d_\Bu$  defined in Eq.~(\ref{eq-Burea_dist}) coincides with the Fubini-Study distance for pure
states. It is Riemannian and contractive
under CPTP maps. In fact,
it is the smallest distance featuring these two properties~\cite{Petz1996} (see also
Ref.~\cite{Spehner_review}).
It can be bounded in terms of the trace distance and vice versa thanks to the following inequalities~\cite{Fuchs1999,Nielsen}:
\begin{equation}
d_{\Bu}(\rho ,\sigma )^2\leq d_{\tr}(\rho ,\sigma )
\leq  \sqrt{2 g \big(  d_{\Bu}(\rho ,\sigma )^2 \big)} \; ,
\label{burestracerelation}
\end{equation}
where $g$ is the function given in Eq.~(\ref{eq-def_function_g_and_f}).

The (quantum) Hellinger distance $d_\Hel$ defined in
Eq.~(\ref{eq-Q_Hellinger_distance}) appears naturally in the context
of the quantum Chernoff bound~\cite{Audenaert2007}. For commuting states, it reduces, just like the
Bures distance, to the (classical) Hellinger distance
between probability distributions.
Although its definition bears some similarity with that of the
Hilbert-Schmidt distance, it shares many properties of the Bures
distance. Indeed, $d_{\Hel}$ is Riemannian and contractive with
respect to CPTP maps (the contractivity can be derived from Lieb's
concavity theorem~\cite{Lieb1973}, see Ref.~\cite{Audenaert2007}).
Since $d_\Bu$ is the smallest Riemannian contractive
distance, for any states $\rho$ and $\sigma$ one has
\begin{equation}
d_{\Bu}(\rho ,\sigma )\leq d_\Hel (\rho ,\sigma ) \; .
\label{buchel}
\end{equation}
Remarkably, there exists also an important bound of the trace distance
in terms of the Hellinger distance~\cite{Holevo1972}:
\begin{equation}
d_{\Hel}(\rho ,\sigma )^2 \leq d_{\tr}(\rho ,\sigma )\leq
\sqrt{2 g \big(  d_{\Hel} (\rho ,\sigma )^2\big)} \; .
\label{Holevo2}
\end{equation}

The first, second, and third inequalities on the geometric discords in
Eq.~(\ref{eq-trivial_boundsbis}) are
trivial consequences of
Eqs.~(\ref{buchel}),~(\ref{Holevo2}), and~(\ref{burestracerelation}), respectively.
The corresponding bounds for the discord of response, which
are obtained
by taking into account the normalization factor in
Eq.~(\ref{def:quantumn}), read
\begin{equation} \label{eq-trivial_bound_disc_response}
D_\Bu^R ( \rho)^2
\leq D_\Hel^R ( \rho)^2
\leq D_{\tr}^R ( \rho) \leq 1 - \big( 1 - D_\Bu^R ( \rho) \big)^2 \; .
\end{equation}
The last inequality holds as an equality for pure
states. Indeed, the upper bound on the trace distance
in Eq.~(\ref{burestracerelation}) is saturated when both $\rho$
and $\sigma$ are pure~\cite{Nielsen}.
This explains why the fourth bound in
Eq.~(\ref{eq-bounds_geo_meas_Bu_Hel}) and the first bound in
Eq.~(\ref{eq-relations_trace_and_HS}) are tight and saturated for pure
two-qubit states.
Our numerical results displayed in Fig.~\ref{responsegrid}(c) indicate that the first bound
in Eq.~(\ref{eq-trivial_bound_disc_response}) is almost tight when $n_A=n_B=2$.

Let us also point out that the Bures and Hellinger distances are monotonic under tensor products, that is,
if two states $\rho_1$ and $\sigma_1$ are closer to each other than
two other states $\rho_2$ and $\sigma_2$, then the same is true for the states
$\rho_1^{\otimes 2}$, $\sigma_1^{\otimes 2}$ and $\rho_2^{\otimes 2}$,
$\sigma_2^{\otimes 2}$, when considering two identical
copies of the system. We remark that the trace distance does not enjoy this property.

\subsection{Geometric measures as proper measures of quantum correlations} \label{sec-def_geo_discords_proper}

By using known results in the literature, we show in this subsection
that  $D^G$,  $D^M$, and   $D^R$ are {\em bona fide} measures of quantum
correlations satisfying Axioms (i-iv) of Section~\ref{sec-intro}
when the distance is the trace, Bures, or Hellinger distance.

Let us first prove that $D^G$, $D^M$, and $D^R$ satisfy Axiom (i), irrespective of the choice of the distance.
This is obvious for the geometric discord. For the  measurement-induced geometric discord, this comes from the fact that
a state is classical-quantum if and only if it is invariant under a
von Neumann measurement on $A$ with rank-one
projectors~\cite{Ollivier2001,Spehner_review}. Note that
this would not be true if the minimization in
Eq.~(\ref{def:geomdisc2}) was performed over projectors $\Pi_i^A$
with ranks larger than  one.
For the discord of response, the validity of Axiom (i) is a consequence of the following
result derived in Ref.~\cite{Roga2014}: the only bipartite states for which
there exists a local unitary
transformation on $A$ such that the unitary operator has a non-degenerate
spectrum and the state is
invariant under the transformation are the classical-quantum states.

The fact that the three geometric measures obey Axiom (ii) holds for any
unitary invariant distance, and thus in particular for the
trace, Hilbert-Schmidt, Bures, and Hellinger distances.

It is easy to verify that $D^G$,  $D^M$, and  $D^R$ satisfy the
monotonicity Axiom (iii) provided that $d$ is contractive under
CPTP maps (see, e.g., Ref.~\cite{Spehner_review}). As explained above, this is the case for the trace,
Bures, and Hellinger distances, but not for the Hilbert-Schmidt distance.

It remains to prove that the geometric measures satisfy Axiom (iv), \ie, that
they reduce to
entanglement monotones for pure states. To this end, one can exploit
the following known results:
\begin{itemize}
\item[(a)] For pure states, the Bures discord of response is given by~\cite{Roga2014}:
\begin{equation} \label{eq-Bures_disc_resp_pure_states}
D_\Bu^R(\ket{\Psi})= 1  - (1- E^R ( \ket{\Psi} ))^\onehalf \; ,
\end{equation}
where $E^R(\ket{\Psi})$ is the {\it entanglement of response}~\cite{Monras2011} (or {\em unitary entanglement}), originally introduced in
Ref.~\cite{Giampaolo2007}, which is equal to one minus the maximum fidelity
between the pure state $\ket{\Psi}$ and the state obtained  
 by perturbing $\ket{\Psi}$ with a local unitary operator $U_A \in \Uu_\Lambda$:
\begin{equation} \label{eq-entanglement_of_response}
E^R ( \ket{\Psi})
\equiv  \min_{U_A \in \Uu_\Lambda} \big\{ 1 -  | \bra{\Psi} U_A \otimes
\idty \ket{\Psi} |^2 \big\} \; .
\end{equation}
The entanglement of response is an entanglement monotone and can be extended to mixed states via the convex roof construction~\cite{Monras2011}.

\item[(b)] The trace and Hilbert-Schmidt discords of response
  $D^R_{\tr}$ and $D_\HS^R$ coincide exactly with the entanglement of
  response $E^R$ on pure states.

\item[(c)] The Bures geometric discord $D_\Bu^G$ reduces for pure states to~\cite{Spehner2013,Spehner_review}:
\begin{equation} \label{eq-Bures_geo_disc_pure_states}
D_\Bu^G(\ket{\Psi})= 2  - 2 (1- E^W ( \ket{\Psi} ))^\onehalf \; ,
\end{equation}
where
$E^W ( \ket{\Psi} ) \equiv \min \{ 1 -  | \braket{\Phi_\sep}{\Psi} |^2 \}$ is
the Wei-Goldbart measure of global geometric
entanglement~\cite{Wei03}. Here, the minimum is taken over all separable
pure  states $\ket{\Phi_\sep}$ (\ie, product states). The  measure
$E^W$ is an entanglement monotone (see
e.g. Ref.~\cite{Spehner_review}). It has been extended to pure multipartite hierarchies in Ref.~\cite{Blasone08} and to
mixed bipartite states in Ref.~\cite{Streltsov10}.

\item[(d)]   Let $\delta (\rho,\sigma)$ 
be a non-negative function defined on the set of pairs of quantum
states, which is contractive 
with respect to CPTP maps and  satisfies the `flags' condition 
$\delta ( \sum_i p_i \rho_i \otimes \ketbra{i}{i}, \sum_i p_i \sigma_i \otimes \ketbra{i}{i} )=
\sum_i p_i \delta ( \rho_i, \sigma_i )$. Let $\Ll_A$ be a family of CPTP maps on subsystem $A$ which is closed 
under unitary conjugations. It is proven in Ref.~\cite{Piani2014} that the measure of quantumness $Q_{\delta,\Ll_A}$ defined by
\begin{equation}
Q_{\delta, \Ll_A} (\rho) = \inf_{\Lambda_A \in \Ll_A} \delta ( \rho, \Lambda_A \otimes 1 ( \rho))
\end{equation}
is entanglement monotone when restricted to pure
states.
\end{itemize}

Exploiting statements (a-b) above and Eq.~(\ref{eq-rel_disc_resp_HS_Hel}), it follows that on pure states
\begin{equation} \label{eq-def-entanglement_of_response}
D^R_\Hel ( \ket{\Psi})  =  D^R_\HS ( \ket{\Psi})= D^R_{\tr} (\ket{\Psi}) = 1- ( 1 - D^R_\Bu( \ket{\Psi}))^2
 = E^R ( \ket{\Psi}) \; .
\end{equation}
Therefore, the discord of response satisfies Axiom (iv) for the
trace, Hilbert-Schmidt, Bures, and Hellinger distances. 

Note that in this paper we call {\it entanglement monotone} any
measure $E$ on pure states such that 
$E(\ket{\Psi}) \geq E ( \ket{\Phi} )$ whenever $\ket{\Phi}$ can be obtained from $\ket{\Psi}$
by local operations and classical communication between the two parties
$A$ and $B$. We do not ask here that
$E$ obeys the stronger monotonicity requirement $ E(\ket{\Psi}) \geq \sum_i
p_i E (\ket{\Phi_i})$, where the right-hand side is the average entanglement of the
post-selected state
$\ket{\Phi_i}$ conditioned to the  measurement outcome $i$ and
$p_i$ is the corresponding probability, the inequality being true for 
arbitrary measurements (including non-projective ones). Hence, according to our definition, if
$E$ is an entanglement monotone and $f$ is a non-decreasing function
(not necessarily concave), then $f(E)$ is still an
entanglement monotone.
Given that the entanglement of response $E^R$ is strongly
entanglement monotone~\cite{Monras2011}, by
Eq. (\ref{eq-def-entanglement_of_response}) 
the discords of response $D^R_\Hel$, $D^R_\HS$, and $D^R_{\tr}$
actually satisfy a stronger version of Axiom (iv) in which entanglement
monotone is replaced by
strongly entanglement monotone. 
This statement is also true for the
Bures discord of response $D^R_\Bu$, as it can be shown by using
Eqs. (\ref{eq-Bures_disc_resp_pure_states}) and (\ref{eq-entanglement_of_response}) 
 and the
characterization of strongly entanglement monotones in
Ref.~\cite{Vidal00}.  

By the property (c) 
above, the Bures geometric discord $D_{\Bu}^G$ satisfies  Axiom (iv).
Later on, we will show that this is also the case for the Hellinger geometric discord $D_{\Hel}^G$ (see Section~\ref{sec-geometric_discord}) and
for the Bures and Hellinger measurement-induced geometric discords
$D_{\Bu}^M$ and $D_{\Hel}^M$ (see
Section~\ref{sec-meas_ind_geo_disc}).
Note that the right-hand side of Eq.~(\ref{eq-Bures_geo_disc_pure_states}) is not strongly entanglement
monotone~\cite{Spehner_review}, so that $D_{\Bu}^G$ does not satisfy the aforementioned stronger version of
Axiom~(iv).  
 
The general property (d) implies that 
 the measurement-induced geometric discord $D^M$ and the discord of
 response $D^R$ fulfill  Axiom (iv) for the trace, Bures, and Hellinger distances.
Indeed, one easily checks that the trace distance and the square Bures and
Hellinger
distances obey the `flags' condition. By taking $\Ll_A$ to be the
 family 
of rank-one projective measurement CPTP maps $\rho \mapsto \rho_{\rm p.m.}^{\{ \Pi_i^A\}}$
on $A$ defined by Eq. (\ref{def:geomdisc2}), one finds that 
$D^M$ 
is entanglement monotone on pure states for the trace, Bures, and Hellinger
distances. 
The same statement holds for $D^R$ and is obtained
by taking $\Ll_A$ to be 
the family of local unitary conjugations $\rho \mapsto U_A \otimes \idty \,\rho \, U_A^\dagger \otimes \idty $  
for all unitary operators $U_A$ with spectrum $\Lambda$.
Actually, one can deduce from Theorem 2.4 of Ref.~\cite{Piani2014} that
$(D^M_{\tr})^{1/2}$, $D^M_\Bu$ and $D^M_\Hel$ satisfy the 
strong version of Axiom (iv), and similarly for $(D^R_{\tr})^{1/2}$,
$D^R_\Bu$, and $D^R_\Hel$, in agreement with our observation above.

Finally, the trace geometric discord $D^G_{\tr}$ obeys Axiom
(iv) at least when subsystem $A$ is a qubit,  because in this case one
has $D^G_{\tr} = D^M_{\tr}=D_{\tr}^R $, see Eq.~(\ref{eq-comparison_geomeas_trace_HS}).

In summary,  for the trace,
Bures, and Hellinger distances, $D^G$, $D^M$, and $D^R$
are all {\em bona fide} measures of quantum correlations whatever the dimensionality $n_A$ of
subsystem $A$ (except perhaps for $D_{\tr}^G$, for which  we are not
aware of 
a proof of Axiom (iv) in the case $n_A>2$).
In contrast, if one chooses the Hilbert-Schmidt
distance, then $D^G$,  $D^M$, and $D^R$ are not contractive under
local CPTP maps acting on subsystem $B$ and are thus not reliable measures of
quantum correlations. This comes from the lack of monotonicity of the Hilbert-Schmidt
distance already mentioned several times in this paper. Indeed, an explicit counter-example
showing that $D_{\HS}^G=D_{\HS}^M$ does not satisfy Axiom~(iii)
has been constructed in Ref.~\cite{Piani2012}.
It is enough to consider the quantum operation which consists
in tracing out over an ancilla $C$, \ie, the
CPTP map $\Phi_B (\rho )=\tr_C ( \rho )$. Here, the subsystem $B$ consists of
two parts $B'$ and $C$ (that is, $C$ is included in $B$). If the total
system $AB$  is in state $\rho = \rho_{AB'} \otimes \rho_C$ with no correlations between $AB'$ and $C$,
adding or removing the ancilla $C$ cannot affect the
quantum correlations between $A$ and $B$.
However, due to the property $d_\HS ( \rho_{AB'} \otimes \rho_C , \sigma_{AB'}
\otimes \rho_C) =  d_\HS ( \rho_{AB'} , \sigma_{AB'} ) \| \rho_C\|_\HS$, one finds
\begin{equation}
D_\HS^M ( \rho )  = D_\HS^M ( \Phi_B( \rho ) )
\tr (\rho_C^2) \; .
\end{equation}
If the ancilla $C$ is not in a pure state then $\tr (\rho_C^2)<1$ and thus
$D_\HS^M (\Phi_B ( \rho )) > D_\HS^M ( \rho )$. By the same argument, $D_\HS^R (\Phi_B ( \rho ))  > D_\HS^R ( \rho )$.
Therefore, $D_\HS^G=D_\HS^M$ and $D_\HS^R$ are not {\it bona fide} measures of quantum correlations.

\subsection{Bures geometric discord and quantum state discrimination} \label{sec-Bures_geometric_disc}

We describe in this subsection the operational interpretation of the Bures geometric
discord $D_\Bu^G$ and introduce the \CCQ states of a given state
$\rho$ for the Bures distance in terms of a quantum
state discrimination problem. We briefly review the
results of Ref.~\cite{Spehner2013}, which 
justify Eq.~(\ref{eq-formula_Bures_geo_discord}) and will be used several times in what follows.

Let us first introduce the maximum probability of success in
discriminating the states $\rho_i$ with prior probabilities $\eta_i$
by means of projective measurements  on
$AB$ of rank $n_B$, which is given by
\begin{equation} \label{eq-max_proba_vN}
P_{\rm S}^{\,\rm{opt\,v.N.}} ( \{ \rho_i,\eta_i \})
 =  \max_{ \{ \Pi_{i} \} } \sum_{i=1}^{n_A} \eta_i \tr ( \Pi_{i} \rho_i) \; ,
\end{equation}
the maximum being taken over all families $\{ \Pi_{i} \}$ of
orthogonal  projectors on $\Hh_A \otimes \Hh_B$ with rank $n_B$. Consider now that the Bures geometric discord
is expressed in terms of the maximum fidelity $F (\rho , \classQ  )$ between $\rho$
and a classical-quantum state
by (see Eqs.~(\ref{eq-Burea_dist}) and~(\ref{def:geomdisc1})):
\begin{equation} \label{eq-Bures_geo_disc_as_max_fidelity}
D^G_\Bu ( \rho) = 2 -2 \sqrt{F (\rho ,\classQ )} \quad , \quad  
F (\rho , \classQ  )\equiv   \max_{\sigma_\Aclas \in \classQ } F (\rho, \sigma_\Aclas )\; .
\end{equation}
The fidelity $F (\rho ,\classQ )$ is equal to~\cite{Spehner2013,Spehner_review}:
\begin{equation} \label{eq-variationnal_formula_bis}
F (\rho , \classQ  ) =
\max_{\{ \ket{\alpha_i} \} } P_{\rm S}^{\,\rm{opt\,v.N.}} ( \{ \rho_i,\eta_i \}) \; ,
\end{equation}
where
the states $\rho_i$ to be discriminated and their probabilities
$\eta_i$ depend on the \ONB $\{ \ket{\alpha_i}\}$ for $A$ and are given by
Eq.~(\ref{eq-state_Q_discrimination}), and the maximum is taken over all
such bases.
  If the density matrix $\rho$ is invertible, 
the optimal measurement 
in ambiguous quantum state discrimination to discriminate the $\rho_i$'s 
is a von Neumann measurement with projectors of rank $n_B$. Hence 
$F (\rho , \classQ  )$ is also equal to the maximum over all
orthonormal bases $\{ \ket{\alpha_i}\}$ of the optimal success
probability $P^{\,\rm{opt}}_{\rm S}  ( \{ \rho_i,\eta_i \})$ 
to  discriminate the states $\rho_i$ using arbitrary POVMs.

For the Bures distance, the classical-quantum states $\sigma_{\Bu, \rho}$ closest to $\rho$ are
expressed in terms of the optimal basis vectors $\ket{\alpha_i^{\opt}}$ and
optimal orthogonal projectors $\Pi_i^{\rm{opt}}$ maximizing
the right-hand sides of Eqs.~(\ref{eq-variationnal_formula_bis}) and~(\ref{eq-max_proba_vN}):
\begin{equation} \label{eq-again_I_was_stupid}
\sigma_{\Bu, \rho} = \frac{1}{F (\rho , \classQ )}
\sum_{i=1}^{n_A} \ketbra{\alpha_i^{\opt}}{\alpha_i^{\opt}}
\otimes \bra{\alpha_i^{\opt}} \sqrt{\rho}\, \Pi_i^{\opt} \sqrt{\rho} \ket{\alpha_i^{\opt}}  \; .
\end{equation}

When $A$ is a qubit, one has to discriminate $n_A=2$ states. It is then
a simple task to evaluate the maximum success probability. The latter is given by the
celebrated Helstrom  formula~\cite{Helstrom1976}
(see also Ref.~\cite{Spehner_review}, Corollary 11.B.6):
\begin{equation} \label{eq-Helstom_formula}
 P_{\rm S}^{\,\rm{opt\,v.N.}} ( \{ \rho_i,\eta_i \})  = \frac{1}{2} \Bigl( 1 + \| \eta_2 \rho_2 - \eta_1 \rho_1 \|_{\tr} \Bigr) \; .
\end{equation}

With the help of Eq.~(\ref{eq-variationnal_formula_bis}), one can
show that $D_\Bu^G$ 
satisfies the last Axiom (v) of
Section~\ref{sec-intro}, in addition to the other Axioms
(i-iv)~\cite{Spehner2013,Spehner_review}. When $n_A \leq n_B$, the
maximum value  of $D_\Bu^G$ is
$D_\mmax = 2 - 2/\sqrt{n_A}$, as reported in Table~\ref{tab1}.

The bounds on $D_\Bu^G$ reported in Tables~\ref{tab1} and~\ref{tab2}
will be derived in
Sections~\ref{sec-geometric_discord} and~\ref{sec-meas_ind_geo_disc} by combining
Eq.~(\ref{eq-variationnal_formula_bis}) with some known
bounds on the maximum success probability in quantum state
discrimination theory.

\section{Hellinger geometric discord: computable and {\em bona fide} measure of quantum correlations}
\label{sec-geometric_discord}

Geometric discords have been studied
in Refs.~\cite{Nakano2013,Paula2013,Ciccarello2014} for the trace distance and
in Refs.~\cite{Streltsov2011c,Abad2012,Aaronson2013,Spehner2013,Spehner2014}
for the Bures distance.
In contrast, to the best of our knowledge, the Hellinger geometric discord has not been
studied in previous works for finite-dimensional systems.

For two-mode Gaussian states $\rho$ of a continuous-variable system,
the Gaussian geometric discord, defined as the minimal Hellinger distance
between $\rho$ and a classical-quantum Gaussian state, has been investigated in
Ref.~\cite{Marian15}. However, since, quite remarkably, for Gaussian
states the classical-quantum states coincide with product states, this
Gaussian discord is
actually a measure of total (classical plus quantum) rather than
quantum correlations. Thus it only provides  an upper bound on the  
Hellinger geometric discord $D_\Hel^G (\rho)$ for Gaussian states
$\rho$ (in fact, the closest
classical-quantum state to $\rho$ is not necessarily Gaussian). A general study that analyzes and compares
geometric and entropic measures of quantum correlations for Gaussian
states is in preparation and  will appear in a forthcoming paper~\cite{Buono2015}.

In this section, we show that if the Hilbert spaces of $A$ and $B$ have finite dimensions,
for any pure state $\ket{\Psi}$,
$D_\Hel^G(\ket{\Psi})$ is a simple function of the
Schmidt number of $\ket{\Psi}$ and for any mixed state $\rho$, $D_\Hel^G (\rho)$ is simply related to the Hilbert-Schmidt geometric discord
$D_\HS^G$ of the square root of $\rho$. One deduces from the first
statement that $D_\Hel^G$ enjoys all the properties (i-iv) of {\em bona fide} measures of
quantum correlations,   in contrast to the Hilbert-Schmidt
geometric discord (see Sec.~\ref{sec-def_geo_discords_proper}). Given the computability of
$D_\HS^G$, it follows from the second statement
that $D_\Hel^G$ is also easy to compute, as confirmed by
the  closed formula for qubit-qudit states derived in
Sec.~\ref{sec-computability}.
We also determine the closest classical-quantum state of a given state $\rho$
with respect to the Hellinger distance.
Finally, we show that the  Hellinger geometric
discord provides upper and lower bounds on the Bures geometric
discord.

\subsection{General expressions and closest classical-quantum states} \label{sec-geo_disc_Hell_mixed_states}

Recall that the Schmidt number of a pure state $\ket{\Psi}$ of a
bipartite system $AB$ is defined as $K ( \ket{\Psi}) = ( \sum_i \mu_i^2 )^{-1}$, where $\mu_i$ are the eigenvalues
of the reduced state $[\rho_\Psi]_A=\tr_B \ketbra{\Psi}{\Psi}$, that is, the non-negative coefficients
appearing in the Schmidt decomposition
\begin{equation} \label{eq-Schmidt_decomposition}
\ket{\Psi} = \sum_{i=1}^{n} \sqrt{\mu_i} \ket{\varphi_i} \ket{\chi_i} \; .
\end{equation}
Here, $n =\min \{ n_A,n_B\}$ and $\{ \ket{\varphi_i} \}_{i=1}^{n_A}$
(respectively $\{ \ket{\chi_j} \}_{j=1}^{n_B}$) is
an \ONB for $A$ (for $B$).

\begin{theorem} \label{eq-theo_geo_disc_Hell_mixed_states}
\begin{itemize}
\item[(a)] If $\ket{\Psi}$ is a pure state of $AB$, then
\begin{equation}  \label{eq-Hellinger_geo_discord_for_pure_states}
D_\Hel^G ( \ket{\Psi}) = 2  - 2 K  ( \ket{\Psi} )^{-\onehalf} \; ,
\end{equation}
where $K ( \ket{\Psi})$ is the Schmidt number of $ \ket{\Psi}$.
Furthermore, the \CCQ state to $\ket{\Psi}$ for the Hellinger distance is the classical-classical state
\begin{equation} \label{eq-Hellinger_CCQ_state}
\sigma_{\Hel, \Psi} = K ( \ket{\Psi} ) \sum_{i=1}^n \mu_i^2 \ketbra{\varphi_i}{\varphi_i} \otimes \ketbra{\chi_i}{\chi_i} \; .
\end{equation}
\item[(b)] If $\rho$ is a mixed state of $AB$, then
\begin{equation} \label{eq-formula_Hellinger_geo_discordbis}
D_\Hel^G ( \rho) = 2  - 2  \max_{ \{ \ket{\alpha_i} \} }  \biggl\{ \sum_{i=1}^{n_A}  {\tr}_B [ \bra{\alpha_i} \sqrt{\rho} \ket{\alpha_i}^2 ]\biggr\}^\onehalf
\; ,
\end{equation}
where the maximum is over all \ONBs  $ \{ \ket{\alpha_i} \} $ for
$A$. Let this maximum be reached for the basis
$ \{ \ket{\alpha_i^\opt} \} $. Then the \CCQ state to $\rho$ for the
Hellinger distance is
\begin{equation} \label{eq-CCL_Hel}
\sigma_{\Hel , \rho}
  =
   \left( 1 - \frac{D_\Hel^G ( \rho)}{2} \right)^{-2} \sum_{i=1}^{n_A} \ketbra{\alpha_i^\opt}{\alpha_i^\opt} \otimes
    \bra{\alpha_i^\opt} \sqrt{\rho} \ket{\alpha_i^\opt}^2\;.
\end{equation}
\end{itemize}
\end{theorem}

Since $K ( \ket{\Psi})$ is an entanglement monotone,
one infers from Eq.~(\ref{eq-Hellinger_geo_discord_for_pure_states})
that $D_\Hel^G$  satisfies Axiom  (iv) of Section~\ref{sec-intro}.
Moreover, it also satisfies the remaining Axioms (i-iii) because the Hellinger distance is contractive under
CPTP maps  (see
Sec.~\ref{sec-def_geo_discords_proper}). Therefore one has:

\begin{corollary}
The Hellinger geometric discord $D_\Hel^G$ is a \emph{bona fide} measure of
quantum correlations.
Its maximum value is given by $D_\mmax = 2 -2 /\sqrt{n_A}$ when $n_A
\leq n_B$.
\end{corollary}

The last statement follows from
Eq.~(\ref{eq-Hellinger_geo_discord_for_pure_states}) and the fact that
if $n_A \leq n_B$ then
$D_\Hel^G$ is maximum for maximally entangled pure states (see the Introduction).  
 
It is  enlightening  to compare the results of Theorem~\ref{eq-theo_geo_disc_Hell_mixed_states} to
the corresponding results
for the Bures distance. For pure states, one has (see
Table~\ref{tab1} and Ref.~\cite{Spehner2013})
\begin{equation} \label{eq-geo_disc_pure_state}
D_\Bu^G ( \ket{\Psi}) = 2 ( 1 - \sqrt{\mu_{\rm max} } )
\end{equation}
and
\begin{equation} \label{eq-closest_A_clas_state_Bures}
\sigma_{\Bu, \Psi} = \ketbra{\varphi_{{\rm max}}}{\varphi_{{\rm max}} } \otimes \ketbra{\chi_{{\rm max}}}{ \chi_{{\rm max}}} \; ,
\end{equation}
where $\mu_{\rm max}= \max \{ \mu_i\} $ is the largest Schmidt
coefficient of $\ket{\Psi}$ and
$\ket{\varphi_{{\rm max}}}$, $\ket{ \chi_{{\rm max}}}$ are the Schmidt vectors corresponding to $\mu_{\rm max}$
in Eq.~(\ref{eq-Schmidt_decomposition}).
If the largest Schmidt coefficient is degenerate, $\ket{\Psi}$ admits
an infinite family of \CCQ states  for the Bures distance, formed by
convex combinations of the states given in Eq.~(\ref{eq-closest_A_clas_state_Bures}). Remarkably, these states are also the closest separable states
to $\ket{\Psi}$.
This means that $D^G_\Bu$ coincides for pure states with the {\it Bures geometric measure of entanglement}
$E^G_\Bu  (\rho) \equiv \min_{\sigma_\sep} d_\Bu ( \rho, \sigma_\sep )^2$~\cite{Vedral98},
where the minimum is over all separable states $\sigma_\sep$. It has
been shown in Ref.~\cite{Streltsov10} that the latter entanglement measure
is simply related to the convex-roof extension $E^W(\rho)$ to mixed
states $\rho$ of  the
Wei-Goldbart geometric entanglement   by
$E^G_\Bu  (\rho)=  2 - 2 (1- E^W(\rho))^{1/2}$.
As already remarked in Sec.~\ref{sec-def_geo_discords_proper}, this
implies that the Bures geometric discord is a {\em bona fide} measure
of quantum correlations,  just like the Hellinger geometric discord.

We now proceed to establish Theorem~\ref{eq-theo_geo_disc_Hell_mixed_states}.

\vspace{1mm}

\proof
Let us  first prove part (b) of the theorem.
From Eqs.~(\ref{eq-Q_Hellinger_distance}) and~(\ref{def:geomdisc1}) it follows that
\begin{equation}
D_\Hel^G (\rho) = 2  - 2 \max_{\sigma_{\Aclas} \in \classQ}  \tr \sqrt{\rho} \sqrt{\sigma_{\Aclas}}  \; .
\end{equation}
By using the spectral decompositions of the states $\rho_{B|i}$ in Eq.~(\ref{eq:cq}), any classical-quantum state can be written as
\begin{equation} \label{eq-A_class_state}
\sigma_\Aclas
=
\sum_{i=1}^{n_A} \sum_{j=1}^{n_B} q_{ij} \ketbra{\alpha_i}{\alpha_i} \otimes \ketbra{\beta_{j|i}}{\beta_{j|i}} \; ,
\end{equation}
where  $\{ q_{ij} \}$ is a probability distribution, $\{ \ket{\alpha_i} \}_{i=1}^{n_A}$ is an orthonormal basis for $A$
and, for any $i$, $\{ \ket{\beta_{j|i}} \}_{j=1}^{n_B}$ is an orthonormal basis for $B$
(note that the $\ket{\beta_{j|i}}$ need not be orthogonal for distinct $i$'s).
The square root of $\sigma_\Aclas $ is obtained by replacing $q_{ij}$ by $\sqrt{q_{ij}}$ in the r.h.s. of Eq.~(\ref{eq-A_class_state}).
Hence
\begin{equation} \label{eq-Hellinger_fidelity_to classical_state}
\tr \sqrt{\rho} \sqrt{\sigma_\Aclas}
 =  \sum_{i,j} \sqrt{q_{ij}} \bra{\alpha_i \otimes \beta_{j|i}} \sqrt{\rho} \ket{\alpha_i \otimes \beta_{j|i}}
\leq
\biggl( \sum_{i,j} \bra{\alpha_i \otimes \beta_{j|i}} \sqrt{\rho } \ket{\alpha_i \otimes \beta_{j|i}}^2 \biggr)^{\onehalf} \; .
\end{equation}
The last bound follows from the Cauchy-Schwarz inequality and the identity $\sum_{i,j} q_{ij} =1$. It is saturated when
\begin{equation} \label{eq-optimal_q_ij}
q_{ij} = \frac{\bra{\alpha_i \otimes \beta_{j|i}} \sqrt{\rho } \ket{\alpha_i \otimes \beta_{j|i}}^2}
{\sum_{i,j} \bra{\alpha_i \otimes \beta_{j|i}} \sqrt{\rho } \ket{\alpha_i \otimes \beta_{j|i}}^2}
\; .
\end{equation}
Therefore
\begin{equation} \label{eq-Hellinger_fidelity_to classical_state2}
\max_{\{ q_{ij}\}} \tr \sqrt{\rho} \sqrt{\sigma_\Aclas }
=
\biggl( \sum_{i,j} \bra{\beta_{j|i}} B_i \ket{\beta_{j|i}}^2 \biggr)^{\onehalf}
\end{equation}
with $B_i = \bra{\alpha_i} \sqrt{\rho } \ket{\alpha_i}$. Note that
$B_i$ is a self-adjoint operator acting on $\Hh_B$.
Now, for any fixed $i$, one has
\begin{equation} \label{eq-ineq_sum_matrix_element_square}
\sum_j  \bra{\beta_{j|i}}  B_i \ket{\beta_{j|i}}^2 \leq \tr [ B_i^2 ] \; .
\end{equation}
This inequality is saturated when $\{ \ket{\beta_{j|i}} \}$ is an eigenbasis of $B_i$.
Maximizing over all classical-quantum states amounts to maximize over all
probability distributions $\{ q_{ij}\}$ and all \ONBs $\{ \ket{\alpha_i} \}$ and $\{ \ket{\beta_{j|i}} \}$. Thus
\begin{equation} \label{eq-formula_Hellinger_geo_discordbis_classicalquantum}
\left( 1 - \frac{D_\Hel^G (\rho)}{2} \right)^2
 =
\max_{ \{ \ket{\alpha_i} \} } \sum_{i} {\rm Tr}_B \big[ \bra{\alpha_i} \sqrt{\rho } \ket{\alpha_i}^2  \big]
\; .
\end{equation}
Eq.~(\ref{eq-formula_Hellinger_geo_discordbis}) in Theorem~\ref{eq-theo_geo_disc_Hell_mixed_states}
follows immediately from this relation.
The \CCQ state $\sigma_{\Hel ,\rho}$ to $\rho$ is given by Eq.~(\ref{eq-A_class_state})
where $\ket{\alpha_i} = \ket{\alpha_i^\opt}$ are the vectors realizing the maximum
in Eq.~(\ref{eq-formula_Hellinger_geo_discordbis_classicalquantum}),
$\ket{\beta_{j|i}} = \ket{\beta_{j|i}^\opt}$ are the eigenvectors
of $B_i^\opt = \bra{\alpha_i^\opt} \sqrt{\rho} \ket{\alpha_i^\opt}$, and (see Eq.~(\ref{eq-optimal_q_ij})):
\begin{equation} \label{eq-optimal_q_ij_bis}
q_{ij} = \frac{\bra{\beta_{j|i}^\opt} (B_i^\opt)^2 \ket{\beta_{j|i}^\opt}}{\sum_i \tr [ (B_i^\opt)^2]} \; .
\end{equation}
The expression  for $\sigma_{\Hel ,\rho}$  in
Theorem~\ref{eq-theo_geo_disc_Hell_mixed_states} readily follows.

We now establish part (a) of the theorem. Let $\rho = \ketbra{\Psi}{\Psi}$ be a pure state with reduced state
$\rho_A = \tr_B \ketbra{\Psi}{\Psi}$. Then $B_i =
\ketbra{\beta_i}{\beta_i}$, where $\ket{\beta_i} =
\braket{\alpha_i}{\Psi}$ has square norm $\| \beta_i\|^2 =
\bra{\alpha_i } \rho_A \ket{\alpha_i}$. Thus Eq.~(\ref{eq-formula_Hellinger_geo_discordbis_classicalquantum}) yields
\begin{equation} \label{eq-formula_Hellinger_geo_discord_oure_state}
\left( 1 - \frac{D_\Hel^G ( \ket{\Psi} )}{2} \right)^2
=
\max_{ \{ \ket{\alpha_i } \} } \sum_{i} \bra{\alpha_i} \rho_A \ket{\alpha_i}^2 \; .
\end{equation}
In analogy with Eq.~(\ref{eq-ineq_sum_matrix_element_square}), the sum in the r.h.s. is bounded from above by
$\tr \rho_A^2 = K( \ket{\Psi})^{-1}$, and the bound is saturated
when $\{ \ket{\alpha_i} \}$ is an  eigenbasis of $\rho_A$. This leads to Eq.~(\ref{eq-Hellinger_geo_discord_for_pure_states}). The \CCQ
state to $\ket{\Psi}$ is given by Eq.~(\ref{eq-CCL_Hel}) with
$\ket{\alpha_i} = \ket{\varphi_i}$. In view of the Schmidt
decomposition~(\ref{eq-Schmidt_decomposition}), one obtains Eq.~(\ref{eq-Hellinger_CCQ_state}).
\finpro

\subsection{Relation to the Hilbert-Schmidt geometric discord} \label{sec-rel_geo_dic_HS_Hel}

The Hilbert-Schmidt geometric discord  can be determined in a similar way as
the Hellinger geometric discord. Let us give here for completeness a self-contained
short derivation of the result, originally derived in Ref.~\cite{Luo2010}.

By definition, the Hilbert-Schmidt geometric discord is
\begin{equation}
D_\HS^G  ( \rho )
= \min_{\sigma_\Aclas \in {\mathcal CQ}} \| \rho - \sigma_\Aclas \|_\HS^2
=  \tr \rho^2 + \min_{\sigma_\Aclas \in {\mathcal CQ}} \tr ( \sigma_\Aclas^2- 2 \rho \sigma_\Aclas ) \; .
\end{equation}
Thanks to Eq.~(\ref{eq-A_class_state}), the last trace is equal to
\begin{equation}
\sum_{i,j}
  \Bigl\{ \bigl( q_{ij} - \bra{\alpha_{i} \otimes \beta_{j|i}} \rho \ket{\alpha_{i} \otimes \beta_{j|i}} \bigr)^2
    - \bra{\alpha_{i} \otimes \beta_{j|i}} \rho \ket{\alpha_{i} \otimes \beta_{j|i}}^2  \Bigr\} \; .
\end{equation}
The minimum over the  probability distribution $\{ q_{ij}\}$ is
obviously achieved  for $ q_{ij} = \bra{\alpha_{i} \otimes
  \beta_{j|i}} \rho \ket{\alpha_{i} \otimes \beta_{j|i}}$.
Minimizing also over the \ONBs $\{ \ket{\alpha_i} \}$ and
$\{ \ket{\beta_{j|i}}\}$ and using
Eq.~(\ref{eq-ineq_sum_matrix_element_square}) again, one finds
\begin{equation}  \label{eq-HS_geo_discord}
D_\HS^G  ( \rho )
 = \tr  \rho^2 - \max_{ \{ \ket{\alpha_i} \}}  \sum_{i=1}^{n_A} {\rm Tr}_B \bra{\alpha_i} \rho \ket{\alpha_i}^2
= \min_{ \{ \ket{\alpha_i} \}} \sum_{i \not= j}^{n_A} {\rm Tr}_B  | \bra{\alpha_i} \rho \ket{\alpha_j} |^2 \; ,
\end{equation}
which is the expression originally found by Luo and Fu~\cite{Luo2010}.
The last equality follows from the relation
$\tr  \rho^2 = \sum_{i,j} \tr_B  | \bra{\alpha_i} \rho \ket{\alpha_j} |^2 $.

By the same argument as above, the closest classical-quantum
state  $\sigma_{\HS, \rho}$ to $\rho$ according to the Hilbert-Schmidt distance coincides
with the  post-measurement state
after a local  measurement on $A$,  namely
\begin{equation} \label{eq-closest_states_meas_ind_geo_disc_HS}
\sigma_{\HS, \rho} = \rho_{\rm p.m.}^\opt
= \sum_{i=1}^{n_A} \ketbra{\alpha_i^\opt }{\alpha_i^\opt }
\otimes \bra{\alpha_i^\opt} \rho \ket{\alpha_i^\opt} \; ,
\end{equation}
where the measurement basis  $\{ \ket{\alpha_i^\opt} \}$ is the \ONB
of $A$ maximizing  the first sum in Eq.~(\ref{eq-HS_geo_discord}). Therefore, as already observed in
Ref.~\cite{Luo2010}, the Hilbert-Schmidt geometric discord $D^G_\HS$
and the measurement-induced geometric discord
$D^M_\HS$ coincide, whatever the space dimensions of the two subsystems.

Next, by comparing Eqs.~(\ref{eq-HS_geo_discord}) and~(\ref{eq-formula_Hellinger_geo_discordbis}), one easily deduces the
following  result:

\begin{theorem} \label{theo-rel_geo_disc_Hel_HS}
For any bipartite state $\rho$ of a composite system $AB$, the
Hellinger geometric discord $D_\Hel^G ( \rho)$ is related to
the Hilbert-Schmidt geometric discord $D_\HS^G ( \sqrt{\rho} )$
of the square root of $\rho$ by
\begin{equation} \label{eq-rel_geo_disc_Hel_HS}
D_\Hel^G ( \rho ) = g^{-1} \big( 2  D_\HS^G ( \sqrt{\rho} ) \big) \equiv 2  - 2 \bigl( 1 - D_\HS^G ( \sqrt{\rho} ) \bigr)^{\onehalf} \; .
\end{equation}
\end{theorem}

Note that the Hilbert-Schmidt geometric discord is evaluated for the
square root of $\rho$, which is not a state but is nevertheless a
non-negative operator. Thus $\sigma = \sqrt{\rho} \, /\tr  \sqrt{\rho}
$ is a density operator and $D_\HS^G ( \sqrt{\rho})$ is defined as
\begin{equation}
D_\HS^G ( \sqrt{\rho}) \equiv  \| \sqrt{\rho} \|_{\tr}^2 D_{\HS}^G ( {\sigma } ) 
=  D_{\HS}^M ( \sqrt{\rho} ) \; ,
\end{equation}
where $D_{\HS}^M ( \sqrt{\rho} ) $ is given by replacing $\rho$ by
$\sqrt{\rho}$ in Eq.~(\ref{def:geomdisc2}).

For pure states, Eq.~(\ref{eq-rel_geo_disc_Hel_HS}) yields a direct relation between the Hellinger and Hilbert-Schmidt geometric discords. Namely,
\begin{equation}
D_\Hel^G ( \ket{\Psi} ) = 2  - 2 (1 - D_\HS^G ( \ket{\Psi} ))^\onehalf \; .
\end{equation}
Consequently, as a further corollary, 
Eq.~(\ref{eq-Hellinger_geo_discord_for_pure_states}) can be recast in the form
$D_\HS^G ( \ket{\Psi} )  = 1 - K( \ket{\Psi})^{-1}$, a result already
known in the literature, see e.g. Refs.~\cite{Luo12,Luo2013}.

As explained in Sec.~\ref{sec-computability},
the calculation of $D_\Hel^G ( \rho) $ is straightforward for
qubit-qudit states $\rho$ once one has
determined the decomposition~(\ref{eq-Boch_dec_square_root}) of
the square root of $\rho$. One can use for this purpose the formula
given in Eq.~(\ref{eq-explicit_formula_geo_disc_2_qubits}).
An alternative derivation of this formula for two-qubit states may be obtained
by combining Eq.~(\ref{eq-rel_geo_disc_Hel_HS}) 
with the result of Ref.~\cite{Daki'c2010} on the Hilbert-Schmidt geometric
discord. Since a generalization of the latter result to bipartite
systems with arbitrary finite space dimensions
$n_A$ and $n_B$ is available~\cite{Luo2010}, a corresponding formula
for $D^G_\Hel(\rho)$ for higher-dimensional systems can be obtained as
well. 

\subsection{Comparison between the Bures and Hellinger geometric discords}

\begin{theorem}  \label{theo-bounds_Bures_geo}
Let us recall the increasing function $g(d)$ defined in Eq.~(\ref{eq-def_function_g_and_f}), and its inverse $g^{-1} (d) \equiv 2 - 2 \sqrt{1-d/2}$.
The Bures and Hellinger geometric discords satisfy
\begin{equation} \label{eq-bounds_Bures_geo}
g^{-1} ( D_\Hel^G ( \rho ))
\leq D_\Bu^G ( \rho ) \leq D_\Hel^G ( \rho ) \; .
\end{equation}
\end{theorem}
In particular, $D_\Bu^G ( \rho )$  lies in the interval bounded by $D_\Hel^G  ( \rho ) /2$ and  $D_\Hel^G  ( \rho )$.
Note that the result of Theorem~\ref{theo-bounds_Bures_geo} can be
rewritten as:
\begin{equation} \label{eq-theo3bis}
D^G_\Bu (\rho) \leq D^G_\Hel (\rho)  \leq g ( D^G_\Bu ( \rho)) \; .
\end{equation}

\proof This is a consequence of
Eq.~(\ref{eq-variationnal_formula_bis}) and of the Barnum--Knill upper
bound on the probability of success in quantum state
discrimination~\cite{Barnum2002}. According to such bound, the maximum
probability of success $P_{\rm S}^{\,\rm{opt\,v.N.}} ( \{ \rho_i,\eta_i \})
$ is at most equal to the square root of the probability of success
obtained by discriminating the states $\rho_i$ with the least--square
measurement. This yields (see Ref.~\cite{Spehner_review} for more  details):
\begin{equation} \label{eq-inequality_square_root_meas}
 \max_{\{ \ket{\alpha_i } \} }  \sum_{i=1}^{n_A}  {\rm Tr}_B \bra{\alpha_i}  \sqrt{\rho} \ket{\alpha_i}^2
\leq
F (\rho , \classQ )
\leq
\max_{\{ \ket{\alpha_i} \} }  \biggl\{ \sum_{i=1}^{n_A}   {\rm Tr}_B  \bra{\alpha_i} \sqrt{\rho} \ket{\alpha_i}^2  \biggr\}^{\frac{1}{2}} \; .
\end{equation}
The second inequality in Eq.~(\ref{eq-inequality_square_root_meas})
together with Eqs.~(\ref{eq-Bures_geo_disc_as_max_fidelity}) and~(\ref{eq-formula_Hellinger_geo_discordbis})  lead
to the first bound in Eq.~(\ref{eq-bounds_Bures_geo}). The second
bound in Eq.~(\ref{eq-bounds_Bures_geo}) is an immediate
consequence of the fact that the Bures distance is always smaller or
at most equal to the  Hellinger distance.
We remark for completeness that by exploiting
Eqs.~(\ref{eq-Bures_geo_disc_as_max_fidelity})
and~(\ref{eq-formula_Hellinger_geo_discordbis}), this second bound
 is equivalent
precisely to  the lower bound in
Eq.~(\ref{eq-inequality_square_root_meas}). \finpro

\section{Measurement-induced geometric discord} \label{sec-meas_ind_geo_disc}

In this Section we derive an upper bound on  the measurement-induced
geometric discord $D^M$ in terms of the geometric discord $D^G$,  both
for the Hellinger and the Bures distances.
We also determine for these two metrics the value that $D^M$ acquires
for a pure state $\ket{\Psi}$ and
the closest post-measurement state to $\ket{\Psi}$ for local
measurements.

\subsection{Hellinger measurement-induced geometric discord}

In view of the definitions in Eqs.~(\ref{eq-Q_Hellinger_distance})
and~(\ref{def:geomdisc2}), the measurement-induced geometric discord
based on the Hellinger distance can be expressed as
\begin{equation} \label{eq-meas_ind_geo_disc_Hell}
D_\Hel^M ( \rho) =
2 - 2 \max_{\{ \ket{\alpha_i}\} } \sum_{i=1}^{n_A} {\tr}_B \bigl[ \bra{\alpha_i} \sqrt{\rho} \ket{\alpha_i} \sqrt{\bra{\alpha_i} \rho \ket{\alpha_i}} \bigr]
\; .
\end{equation}
Here, we have used the expression $( \rho_{\rm p.m.}^{ \{ \ket{ \alpha_i } \} })^{1/2}= \sum_i \ketbra{\alpha_i}{\alpha_i}
\otimes \sqrt{\bra{\alpha_i} \rho\, \ket{\alpha_i}}$ of the square
root of the post-measurement state in Eq.~(\ref{def:geomdisc2}). Let us first study the restriction of $D^M_\Hel$ to pure states.

\begin{theorem} \label{eq-theo_Hel_mes_ind_geo_disc_pure_states}
On pure states, the Hellinger measurement-induced geometric discord is
given by
\begin{equation} \label{eq-meas_ind_geo_disc_pure_states_Hel}
D_\Hel^M ( \ket{\Psi} ) = 2 -2 \sum_{i=1}^n \mu_i^{\frac{3}{2}} \; ,
\end{equation}
where $\mu_i$ are the Schmidt coefficients of
$\ket{\Psi}$. 
The measurement basis $\{ \ket{\alpha_i^\opt} \}$  on
subsystem $A$ which
produces the 
closest post-measurement state to $\ket{\Psi}$ is the
\ONB $\{ \ket{\varphi_i} \}$ formed by the Schmidt vectors in Eq.~(\ref{eq-Schmidt_decomposition})
(\ie, the eigenbasis of the reduced state $[\rho_\Psi]_A$).
\end{theorem}

As a consequence, the  post-measurement state closest to $\ket{\Psi}$ after a
local von Neumann measurement on party $A$ takes the form
\begin{equation} \label{eq-closest_states_meas_ind_geo_disc}
[\rho_\Psi]_{\rm p.m.}^{\opt, \Hel} = \sum_{i=1}^n \mu_i \ketbra{\varphi_i}{\varphi_i} \otimes \ketbra{\chi_i}{\chi_i} \; .
\end{equation}
With the exception of the uniform case $\mu_i = 1/n$ $\forall\; i$,
that is, in all cases in which $\ket{\Psi}$  {\em is not} maximally
entangled, $[\rho_\Psi]_{\rm p.m.}^{\opt, \Hel}$ is distinct from the
closest classical-quantum state
 to $\ket{\Psi}$ (compare with
Eq.~(\ref{eq-Hellinger_CCQ_state})). Therefore, for such non-maximally
entangled pure states, $D_\Hel^M ( \ket{\Psi})$ is always strictly  
larger than the Hellinger geometric discord  $D_\Hel^G (  \ket{\Psi} )$.

\proof
Equation~(\ref{eq-meas_ind_geo_disc_Hell}) yields for $\rho_\Psi= \ketbra{\Psi}{\Psi}$
\begin{equation} \label{eq-proof-meas_ind_geo_disc_pure_states}
D_\Hel^M ( \ket{\Psi} ) =
2 - 2 \max_{\{ \ket{\alpha_i}\} } \sum_{i=1}^{n_A} \| \beta_i \|^3 \; ,
\end{equation}
with the unnormalized vector $\ket{\beta_i}$ in the Hilbert space of
$B$ defined as $\ket{\beta_i} = \braket{\alpha_i}{\Psi}$. The  Schmidt decomposition gives
\begin{equation} \label{eq-proof-meas_ind_geo_disc_pure_states2}
\| \beta_i \|^2   = \sum_{j=1}^n \mu_j | \braket{\alpha_i}{\varphi_j} |^2
 \leq \Biggl( \sum_{j=1}^n \mu_j^{\frac{3}{2}} | \braket{\alpha_i}{\varphi_j} |^2 \biggr)^{\frac{2}{3}} \; ,
\end{equation}
where the upper bound is obtained by combining the H\"older inequality
and $\sum_j  | \braket{\alpha_i}{\varphi_j} |^2 \leq 1$. It follows
that
\begin{equation}
\sum_{i=1}^{n_A} \| \beta_i \|^3 \leq \sum_{j=1}^n
\mu_j^{\frac{3}{2}}\;.
\end{equation}
This  bound is saturated by taking $\ket{\alpha_i} =
\ket{\varphi_i}$ for all $i=1,\ldots,n$, with $n=\min\{ n_A, n_B\}$ (if $n<n_A$, the remaining
vectors $\ket{\alpha_i}$ are chosen arbitrarily to form an orthonormal
basis of $\Hh_A$).  Equation~(\ref{eq-meas_ind_geo_disc_pure_states_Hel})
then follows upon replacing the sum in
Eq.~(\ref{eq-proof-meas_ind_geo_disc_pure_states}) by its upper bound.
\finpro

\vspace{2mm} 

Although we already know from Theorem 2.4 of Ref.~\cite{Piani2014}
that $D_\Hel^M$ satisfies Axiom (iv) (see
Sec.~\ref{sec-def_geo_discords_proper}), it is instructive to deduce
this result from
Eq.~(\ref{eq-meas_ind_geo_disc_pure_states_Hel}). Introducing the
unitarily invariant function $f(\rho_A) = 2 -2 \tr \rho_A^{3/2} $ on
the set of density operators on $\Hh_A$,
this equation can be rewritten as 
$D_\Hel^M ( \ket{\Psi}) = f ( [\rho_\Psi]_A )$.
The fact that this quantity is an entanglement
monotone follows directly from the characterization of convex 
strongly monotone entanglement measures provided by Vidal~\cite{Vidal00}. Indeed, according to Ref.~\cite{Vidal00},
$E(\ket{\Psi})= f([\rho_\Psi]_A )$ defines an entanglement monotone on pure states
if $\rho_A\mapsto f ( \rho_A)$ is concave (notice, however, that this condition is not necessary and sufficient: notable
exceptions are provided by the logarithmic negativity, which is not  convex but is
nevertheless strongly entanglement monotone, see
Ref.~\cite{Plenio2005}, and the Bures geometric
measure of entanglement $E_\Bu^G$,
which is convex but entanglement monotone in the weak sense discussed in
Section~\ref{sec-def_geo_discords_proper}, see Ref.~\cite{Spehner_review}). The concavity of $\rho_A\mapsto f ( \rho_A)$ 
is a consequence of the convexity of $\rho_A \mapsto \tr k (\rho_A)$
for real convex functions $k$,  in particular for $k(x)= x^{3/2}$, as
proved for instance in Ref.~\cite{Carlen}. Hence $D_\Hel^M$ satisfies  Axiom (iv)  of Section~\ref{sec-intro}.
Since the Hellinger distance is
contractive under CPTP maps, $D_\Hel^M$ also fulfills Axioms (i-iii) (see 
Sec.~\ref{sec-def_geo_discords_proper}). Summing up, $D_\Hel^M$ is
a {\em bona fide} measure of quantum correlations. The maximum value of $D_\Hel^M$  when $n_A \leq n_B$ is equal to
$2-2/\sqrt{n_A}$, as reported in Table~\ref{tab2} (this follows from
Eq.~(\ref{eq-meas_ind_geo_disc_pure_states_Hel}) and
the fact that $D_\Hel^M$ is maximum for maximally entangled pure
states).

\begin{theorem} \label{theo-comparison_meas_ind_geo_and_geo_discord_Hel}
The Hellinger geometric discord and Hellinger measurement-induced geometric discord satisfy
\begin{equation} \label{eq-meas_ind_geo_and_geo_discord_Hel}
D^G_\Hel(\rho) \leq D^M_\Hel (\rho) \leq  g ( D^G_\Hel ( \rho) ) \; ,
\end{equation}
with the function $g(d)$ defined by Eq.~(\ref{eq-def_function_g_and_f}).
\end{theorem}
In particular, the ratio $D^M_\Hel(\rho)/D^G_\Hel(\rho)$  always lies in the interval $[1,2]$.

\proof
The first inequality follows  trivially from the definitions.
By the operator concavity of $f(x)=\sqrt{x}$ and the Jensen-type
inequality applied to the CPTP map $\rho \mapsto \sum_i
\ketbra{\alpha_i}{\alpha_i} \otimes \idty \,\rho
\,\ketbra{\alpha_i}{\alpha_i} \otimes \idty$ (see, e.g.,
Refs.~\cite{Carlen,Spehner_review}), the  following operator bound holds
\begin{equation}
 \sqrt{\bra{\alpha_i} \rho \ket{\alpha_i}} \geq \bra{\alpha_i} \sqrt{\rho} \ket{\alpha_i} \; .
\end{equation}
As a consequence, the trace in Eq.~(\ref{eq-meas_ind_geo_disc_Hell}) is bounded from below by $\tr    \bra{\alpha_i} \sqrt{\rho} \ket{\alpha_i}^2 $.
Comparing with the expression~(\ref{eq-formula_Hellinger_geo_discordbis}) of the Hellinger geometric discord,
the upper bound on $D^M_\Hel$ follows.
\finpro

\subsection{Bures measurement-induced geometric discord}
\label{sec-Bures_meas_ind_geo_disc}

For any state $\rho$ of the bipartite system $AB$, let us denote by
${\cal LM}_\rho$ the set of all post-measurement states
obtained from $\rho$ after local rank-one projective measurements on $A$,
that is,  the states given by Eq.~(\ref{def:geomdisc2}).
In analogy with Eq.~(\ref{eq-Bures_geo_disc_as_max_fidelity}), the Bures measurement-induced geometric discord
is equal to
$D_\Bu^M ( \rho) = 2 - 2 \sqrt{F( \rho, {\cal LM}_\rho)}$, where $F ( \rho , {\cal LM}_\rho)$ is the maximum fidelity between $\rho$ and
a state belonging to ${\cal LM}_\rho$. One easily finds
\begin{equation} \label{eq-max_fidelity_meas_ind_geo_disc}
F ( \rho ,{\cal LM}_\rho) = \max_{ \{ \ket{\alpha_i} \}} \left\{ \tr \sqrt{\sum_{i=1}^{n_A} \eta_i^2 \rho_i^2} \right\}^2 \; ,
\end{equation}
with the states $\rho_i$ and probabilities $\eta_i$ given by
Eq.~(\ref{eq-state_Q_discrimination}). This proves
Eq.~(\ref{eq-formula_Bures_meas_ind_discord})  reported in Section~\ref{sec-main_results_general_expressions_Bures_Hellinger}.
Moreover, the following theorem holds:

\begin{theorem} \label{theo-Bures_meas_ind_disc_pure_states}
\begin{itemize}
\item[(a)]
On  pure states, the Bures measurement-induced geometric discord is
given by
\begin{equation} \label{eq-meas_ind_geo_disc_pure_states}
D_\Bu^M ( \ket{\Psi} )  = D_\Hel^G ( \ket{\Psi} ) = 2 - 2 K ( \ket{\Psi} )^{-\onehalf} \; ,
\end{equation}
where $K ( \ket{\Psi} ) = (\sum_i \mu_i^2 )^{-1}$ is the Schmidt number of $\ket{\Psi}$.
In particular, $D_\Bu^M$ satisfies Axiom (iv) of Section~\ref{sec-intro}
and is thus a \emph{bona fide} measure of quantum
correlations. Moreover,
the measurement basis $\{ \ket{\alpha_i^\opt} \}$ which
produces the closest post-measurement state $[\rho_\Psi]_{\rm p.m.}^\opt \in {\cal LM}_\Psi$ to $\ket{\Psi}$
for the Bures distance is the eigenbasis $\{  \ket{\varphi_i}\}$ of
the reduced state $[\rho_\Psi]_A$.
\item[(b)] For any mixed state $\rho$, if $n_A \leq n_B$ then the maximum value $D_\mmax$ of
$D_\Bu^M (\rho)$ is $D_\mmax = 2 - 2/\sqrt{n_A}$. Moreover, $D_\Bu^M(\rho) = D_\mmax$ if and only if
$\rho$ has maximum entanglement of formation $\EoF ( \rho) = \ln n_A$. Thus, $D_\Bu^M$ is a proper measure of quantum correlations that, besides Axioms (i)-(iv), satisfies also the additional Axiom (v).
\end{itemize}
\end{theorem}

Quite remarkably, the post-measurement state 
$[\rho_\Psi]_{\rm  p.m.}^\opt \in {\cal LM}_\rho$ which is closest to the pure state
$\ket{\Psi}$ is the same for the   Hellinger, Bures, trace, and
Hilbert-Schmidt distances. The explicit expression of this state is given by
Eq.~(\ref{eq-closest_states_meas_ind_geo_disc}). This is a consequence of
Theorems~\ref{eq-theo_Hel_mes_ind_geo_disc_pure_states} and~\ref{theo-Bures_meas_ind_disc_pure_states} above and of 
Theorem~3.3. of Ref.~\cite{Piani2014} for the three first distances. For the Hilbert-Schmidt distance,
this follows from Eqs. (\ref{eq-HS_geo_discord}) and
(\ref{eq-closest_states_meas_ind_geo_disc_HS}) and from the bound in
Eq.~(\ref{eq-almost_finished}) below. 

Comparing Eqs.~(\ref{eq-geo_disc_pure_state}) and~(\ref{eq-meas_ind_geo_disc_pure_states}) we find that, as for the case
of the Hellinger distance, the trivial bound
 $D_\Bu^G(\ket{\Psi} ) \leq D_\Bu^M (\ket{\Psi})$ is strict  for all non-maximally
entangled pure states.

Let us also stress that statement (b) holds as well for the Bures geometric
discord $D_\Bu^G$, which has the same maximum value $D_\mmax$ when
$n_A \leq n_B$ (see Sec.~\ref{sec-Bures_geometric_disc}).

\proof We first consider the case of pure states: setting $\rho = \ketbra{\Psi}{\Psi}$ in Eq.~(\ref{eq-max_fidelity_meas_ind_geo_disc}), we have
\begin{equation}
F ( \ket{\Psi} ,{\cal LM}_\Psi ) = \max_{ \{ \ket{\alpha_i} \}} \sum_{i=1}^{n_A} \| \beta_i \|^4
\end{equation}
with $\ket{\beta_i} = \braket{\alpha_i}{\Psi}$ as before. Thanks to the
 first identity in
 Eq.~(\ref{eq-proof-meas_ind_geo_disc_pure_states2}) and to the Cauchy-Schwarz inequality, we also have
\begin{eqnarray} \label{eq-almost_finished}
\nonumber
\sum_i \| \beta_i \|^4    =  \sum_{i,j,k} \mu_j \mu_k | \braket{\alpha_i}{\varphi_j} |^2 | \braket{\alpha_i}{\varphi_k} |^2
& \leq  &  \biggl( \sum_{i,j,k} \mu_j^2  | \braket{\alpha_i}{\varphi_j} |^2 | \braket{\alpha_i}{\varphi_k} |^2 \Biggr)^\onehalf
 \biggl( \sum_{i,j,k} \mu_k^2  | \braket{\alpha_i}{\varphi_j} |^2 | \braket{\alpha_i}{\varphi_k} |^2 \Biggr)^\onehalf
\\
& \leq & \sum_{j=1}^n \mu_j^2 = K ( \ket{\Psi} )^{-1} \; .
\end{eqnarray}
The bound is saturated by taking $\ket{\alpha_i} = \ket{\varphi_i}$
for $i=1,\ldots, n$, hence the maximum of the l.h.s. coincides with $K ( \ket{\Psi} )^{-1}$.
The statement that $D_\Bu^M$ is a {\it bona fide} measure of quantum
correlations follows from the results of
Section~\ref{sec-def_geo_discords_proper} and the fact that $K ( \ket{\Psi} )$ is an entanglement monotone.

We now consider the case of mixed states.
The statement (b) follows from Eq.~(\ref{eq-max_fidelity_meas_ind_geo_disc}) and the following trace inequality: for any finite family of operators
$\{ X_i\}_{i=1}^{n_A}$,
\begin{equation} \label{eq-trace_inequality}
\biggl\| \sum_{i=1}^{n_A} X_i \biggr\|_{\tr} \leq \sqrt{n_A} \tr \sqrt{ \sum_{i=1}^{n_A} |X_i|^2} \; ,
\end{equation}
with equality  if and only if all $X_i$ are equal. This inequality is
a consequence of the operator monotonicity of the square root function
(see e.g. Refs.~\cite{Carlen,Bathia})
and of the operator bound $|\sum_{i=1}^{n_A} X_i |^2 \leq n_A
\sum_{i=1}^{n_A} |X_i|^2$, which in turn follows from
$X_i^\dagger X_j + X_j^\dagger  X_i \leq | X_i|^2 + |X_j|^2$.
Taking $X_i = \eta_i \rho_i$ in Eq.~(\ref{eq-trace_inequality}), using
Eq.~(\ref{eq-max_fidelity_meas_ind_geo_disc}), and recalling that
$\sum_i \eta_i \rho_i = \rho$, one finds that $F ( \rho ,{\cal LM}_\rho) \geq 1/n_A$. Hence
$D^M_\Bu ( \rho) \leq D_\mmax$. Together with the trivial bound
$D^G_\Bu ( \rho) \leq D^M_\Bu (\rho)$ and the identity $D^G_\Bu ( \rho) = D_\mmax$ for maximally
entangled states (as $D_\Bu^G$ satisfies Axiom (v)),  this implies that $D^M_\Bu ( \rho) = D_\mmax$ on
such states. Conversely,  let us consider a state $\rho$ such that
 $D^M_\Bu ( \rho) = D_\mmax$, \ie , $F ( \rho ,{\cal LM}_\rho) =1/n_A$.
Then the inequality in Eq.~(\ref{eq-trace_inequality}) with $X_i =
\eta_i \rho_i$ is saturated, so that
$\eta_i \rho_i$ is independent of $i$, for any \ONB 
$\{ \ket{\alpha_i}\}$. One deduces from  the relations
$\tr \rho_i = 1 =\sum_i \eta_i $ and
$\sum_i \eta_i \rho_i = \rho$  that $\eta_i = 1/n_A$ and $\rho_i=\rho$
for any $i=1,\ldots, n_A$ and any basis $\{ \ket{\alpha_i}\}$.
By using the same arguments as in the proof of the Proposition following Theorem~3 in Ref.~\cite{Spehner2013},
one concludes that $\rho$ is a maximally entangled state according to
the entanglement of formation. More specifically, $\rho$ is a convex combination of maximally
entangled pure states $\ket{\Psi_k}$, whose  expression is provided by
Eq.~(\ref{eq-max_entangled_state}) below and which satisfy the orthogonality
conditions given after this equation. Note that such maximally
entangled states are not necessarily pure if $n_B \geq 2 n_A$. 
\finpro

\vspace{1mm}

\begin{theorem} \label{theo-comparison_meas_ind_geo_and_geo_discord_bures}
The Bures geometric discord $D_\Bu^G$ and the Bures measurement-induced geometric
discord $D_\Bu^M$ satisfy a bound analogous to that established in
Theorem~\ref{theo-comparison_meas_ind_geo_and_geo_discord_Hel} for
$D^G_\Hel$ and $D^M_\Hel$, namely:
\begin{equation} \label{eq-meas_ind_geo_and_geo_discord_buresbis}
 D^G_\Bu ( \rho) \leq D^M_\Bu ( \rho) \leq g( D^G_\Bu ( \rho) ) \; ,
\end{equation}
with $g(d) = 2 d - d^2/2$.
\end{theorem}

Let us observe that the lower and upper bounds on $D^M_\Bu$ in
Eq.~(\ref{eq-meas_ind_geo_and_geo_discord_buresbis}) are identical to
the bounds  on $D_\Hel^G$ in Eq.~(\ref{eq-theo3bis}).
It is clear that the upper bound is not optimal
for strongly quantum-correlated states (in fact, one has
$g(2-2/\sqrt{n_A})= 2- 2/n_A > D_\mmax$). On the other hand, this  upper bound
is optimal in the limit of almost non-discordant states. Indeed,
consider a pure state $\ket{\Psi_\varepsilon}$  with maximum Schmidt
coefficient $\mu_\mmax = 1- \varepsilon$ with
$0 \leq \varepsilon \ll 1$. From Eqs.~(\ref{eq-Hellinger_geo_discord_for_pure_states}), 
(\ref{eq-geo_disc_pure_state}), and~(\ref{eq-meas_ind_geo_disc_pure_states}) it follows that
$g ( D_\Bu^G (\ket{\Psi_\varepsilon})) = D_\Bu^M (\ket{\Psi_\varepsilon}) =D_\Hel^G ( \ket{\Psi_\varepsilon } ) = 2
\varepsilon$ up to terms of order
$\varepsilon^2$. This means that the upper bounds  in
Eqs.~(\ref{eq-theo3bis})
and~(\ref{eq-meas_ind_geo_and_geo_discord_buresbis}) are saturated asymptotically
by pure states that are arbitrarily close to product states.
According to Eq.~(\ref{eq-trivial_boundsbis}), the r.h.s. of
Eq.~(\ref{eq-meas_ind_geo_and_geo_discord_buresbis}) places an upper bound on
$D_{\tr}^G(\rho)/2$ as well. However,  in this case it  is not clear
whether  the bound can be saturated asymptotically.

If subsystem $A$ is a qubit,
 inequalities stronger than the upper bounds in Eqs.~(\ref{eq-theo3bis})
and~(\ref{eq-meas_ind_geo_and_geo_discord_buresbis}) can be obtained from
Theorems~\ref{prop_bound_D_Bu^R_D_Hel^R}
and~\ref{theo-upper_bound_D_Bu^M_intermsof_D_Bu^R} below, by inserting the expressions
(\ref{eq-disc_resp_Hel_n=2,3}) and (\ref{burrg})  into
Eqs.~(\ref{eq-bound_D_Bu^R_D_Hel^R}) and~(\ref{eq-upper_bound_D_Bu^M_intermsof_D_Bu^R}). This yields the improved
tight bounds reported in Section~\ref{eq-results_bounds_qubit}, given
in Eqs.~(\ref{eq-bounds_Bu_geodisc_measdisc}) and~(\ref{eq-bounds_geo_disc_Bu_Hel}).

\proof
The first inequality is trivial.
When the density matrix $\rho$ is invertible, the second inequality follows by
combining Eqs.~(\ref{eq-variationnal_formula_bis}) and~(\ref{eq-max_fidelity_meas_ind_geo_disc}) with the upper bound  by
Ogawa and Nagaoka on the maximum probability of success  in quantum state discrimination~\cite{Ogawa99}:
\begin{equation} \label{eq-upper_bound_P_S}
P^\opt_{\rm S} ( \{ \rho_i, \eta_i\} ) \leq  \tr \sqrt{\sum_{i=1}^{n_A} \eta_i^2 \rho_i^2} \; \; .
\end{equation}
When $\rho$ is not invertible, we may approximate it by the invertible
density matrix
$\rho_\varepsilon = (1- \varepsilon) \rho + \varepsilon\,\idty/(n_A n_B)$
with $\varepsilon \in ( 0,1]$ and obtain the desired result
by continuity, letting $\varepsilon \rightarrow 0$. It is worth
noting that it is also known that the maximum
probability of success is bounded from below by the square of the
r.h.s. of Eq.~(\ref{eq-upper_bound_P_S}), see Ref.~\cite{Tyson09}. However, in our context
such bound is not very useful, as it yields the trivial inequality $D_\Bu^G ( \rho ) \leq D_\Bu^M ( \rho )$.
\finpro

\section{Discord of response: computable and {\em bona fide} measure of quantum correlations}
\label{sec-disc_of_response}

In this Section we show that whenever the reference party $A$ is a
qubit or a qutrit, the Hellinger discord of response is a simple
function of the Hellinger geometric discord, and the same holds true in
the Bures case when $A$ is a qubit
(Theorems~\ref{prophel} and~\ref{propbur}).
As a consequence, the Hellinger discord of response is computable for
all qubit-qudit states, as anticipated in Sec.~\ref{sec-computability}.

We also derive lower and upper bounds on $D^R_\Hel$ valid for arbitrary subsystems $A$ and $B$, first in terms
of $D_\Hel^G$ (Theorem~\ref{prophel}) and then in terms of
$D_\Bu^R$ (Theorem~\ref{prop_bound_D_Bu^R_D_Hel^R}).
Moreover, we obtain an upper bound on the Bures measurement-induced geometric discord
in terms of the Bures discord of response
(Theorem~\ref{theo-upper_bound_D_Bu^M_intermsof_D_Bu^R}).
Finally, we prove that for the trace distance, all geometric measures $D_{\tr}^G$, $D_{\tr}^M$,
and $D_{\tr}^R$ coincide whenever $A$ is a qubit, and we show that the Hilbert-Schmidt discord of response
$D^R_\HS$ is always smaller or equal to the trace discord of response $D_{\tr}^R$ (Theorem~\ref{qtmeqdisc}).

\subsection{Hellinger discord of response: {\em bona fide} and computable measure of quantum correlations} \label{sec-Bures_Hellinger}

The following theorem yields that the Hellinger discord of response
enjoys a simple, exact relation to the Hellinger geometric discord
whenever $A$ is a qubit or a qutrit. For subsystems $A$ with 
space dimensions $n_A >3$, there is no direct relation between these
two measures (see  Appendix~\ref{app-proof_theo5-7}), however we
are able to derive lower and upper bounds on $D^R_\Hel$ in terms of $D^G_{\Hel}$.

\begin{theorem}\label{prophel}
The Hellinger discord of response is bounded in terms of the Hellinger geometric discord as follows:
\begin{equation} \label{bound_disc_resp_Hel}
 \sin^2 \left( \frac{\pi}{n_A} \right) g ( D^{G}_{\Hel}  (\rho) ) \leq
 D_{\Hel}^{R}(\rho) \leq g ( D^{G}_{\Hel}  (\rho) ) \; ,
\end{equation}
with $g(d) = 2 d - d^2/2$.
If subsystem $A$ is a qubit or a qutrit then the first inequality is an equality, that is,
\begin{equation}
D^{R}_{\Hel}(\rho)=\left\{ \begin{array}{ll}
          g ( D^{G}_{\Hel}  (\rho) )  & \mbox{if $n_A=2$,}  \\[6pt]
       \frac{3}{4}  g ( D^{G}_{\Hel}  (\rho) )   & \mbox{if $n_A = 3$.}  \end{array} \right.
\label{eq-disc_resp_Hel_n=2,3}
\end{equation}
\end{theorem}

The proof of this theorem is reported in
Appendix~\ref{app-proof_theo5-7}. 

By combining
Eqs.~(\ref{eq-meas_ind_geo_and_geo_discord_Hel}) and~(\ref{bound_disc_resp_Hel}) one obtains the
upper bound $D_\Hel^M (\rho) \leq D_{\Hel}^R(\rho)/\sin^2(\pi/n_A)$ on the
Hellinger measurement-induced geometric discord, as reported in
Table~\ref{tab3}. Let us also recall that the closed analytical expression of
$D^{R}_{\Hel}(\rho)$ for arbitrary qubit-qudit states $\rho$ in Eq.~(\ref{eq-closed_formula_discresp_Hel})
can be obtained from Eq.~(\ref{eq-disc_resp_Hel_n=2,3}) and the expression of the Hellinger
geometric discord given by Eq.~(\ref{eq-explicit_formula_geo_disc_2_qubits}).

Interestingly, the Hellinger discord of response provides
lower and upper bounds on the Bures discord of response, which is harder to compute. Optimal bounds are provided by the
following theorem:

\begin{theorem}\label{prop_bound_D_Bu^R_D_Hel^R}
For subsystems $A$ and $B$ with arbitrary space dimensions, one has
\begin{equation}
\label{eq-bound_D_Bu^R_D_Hel^R}
1 - \sqrt{1-D_\Hel^R(\rho)} \leq D_\Bu^R ( \rho) \leq D_\Hel^R ( \rho) \; .
\end{equation}
The first bound is saturated for pure states.
\end{theorem}

The numerical results reported in
Fig.~\ref{responsegrid}(c) indicate that the second bound
is almost tight for two-qubit systems.

\proof
The second inequality in the theorem is a
trivial consequence of the fact that the Bures distance is bounded from
above by the Hellinger distance, see Eq.~(\ref{buchel}).
In order to prove the first inequality, we exploit the
definitions of $D_\Bu^R$ and $D_\Hel^R$
(see Eqs.~(\ref{eq-Burea_dist}),~(\ref{eq-Q_Hellinger_distance}), and~(\ref{def:quantumn})) and the identities
$( U \rho \, U^\dagger )^\onehalf = U \sqrt{\rho} \, U^\dagger$ and
$F ( \rho, U \rho\,U^\dagger  ) =
 \|  \sqrt{\rho} \,U \sqrt{\rho} \|_{\tr}^2$ holding for any
unitary operator $U$. This gives 
\begin{eqnarray}
\label{eq-Bures_disc_resp_bis}
D_\Bu^R ( \rho)
& = &
1 - \max_{U_A \in \Uu_\Lambda} \bigl\|  \sqrt{\rho} \,U_A \otimes \idty\, \sqrt{\rho}
\bigr\|_{\tr} \; ,
\\
\label{eq_Hellinger_disc_resp}
D_\Hel^R ( \rho)
& = &
1 - \max_{U_A \in \Uu_\Lambda} \tr \big( \sqrt{\rho} \, U_A \otimes \idty
\sqrt{\rho} \,U_A^\dagger \otimes \idty \big) \; .
\end{eqnarray}
We now take advantage of the bound
\begin{equation} \label{eq-trace_inequality_proof_prop_bound_D_Bu^R_D_Hel^R}
\big\| \sqrt{\rho} \,U_A \otimes \idty \sqrt{\rho} \bigr\|^2_{\tr}
\leq
\tr \big( \sqrt{\rho} \, U_A \otimes \idty
\sqrt{\rho} \, U_A^\dagger \otimes \idty \big) \;.
\end{equation}
This bound follows from
the identity $\| A \|_{\tr} = \tr ( V^\dagger A)$
with $V$ a unitary operator such that $A = V |A|=\sqrt{\rho} \,U_A \otimes \idty \sqrt{\rho} $ (polar
decomposition), from the Cauchy-Schwarz inequality
$|\tr ( B^\dagger C ) |^2 \leq  \tr |B|^2 \tr |C|^2$ with
$B= \rho^{\frac{1}{4}} V \rho^{\frac{1}{4}}$ and
$C= \rho^{\frac{1}{4}} U_A  \otimes \idty \rho^{\frac{1}{4}}$, and
from $\tr |B|^2 = \tr  \sqrt{\rho} \,V^\dagger \sqrt{\rho}\, V \leq \tr \rho = 1$ (again by the Cauchy-Schwarz inequality).
Then, combining Eqs.~(\ref{eq-Bures_disc_resp_bis})-(\ref{eq-trace_inequality_proof_prop_bound_D_Bu^R_D_Hel^R}), it holds that $1 - D_\Hel^R ( \rho)  \geq  ( 1 - D_\Bu^R (\rho) )^2$. This inequality is an equality for pure states, as can easily be inferred
from Eqs.~(\ref{eq-Bures_disc_resp_pure_states}) and~(\ref{eq-def-entanglement_of_response}).
\finpro

\subsection{Bures discord of response} \label{sec-Bures_Hellinger2}

If $U_A$ is a local unitary operator acting on $\Hh_A$
with spectrum $\Lambda$ given by the roots of the unity, then
\begin{equation} \label{eq-unitaries_harm_spectrum}
U_A = \sum_{j=1}^{n_A}  e^{- \I \frac{2 \pi j}{n_A}}  \ketbra{\alpha_j}{\alpha_j}
\end{equation}
for some  orthonormal basis $\{ \ket{\alpha_j} \}$ of
$\Hh_A$.
By inserting
this spectral decomposition into Eq.~(\ref{eq-Bures_disc_resp_bis})
we  obtain
\begin{equation} \label{eq-fidelity_of_response}
D_\Bu^R ( \rho)
 = 1- \max_{\{ \ket{\alpha_i}\}} \biggl\|  \sum_{j=1}^{n_A} \eta_j
 e^{- \I \frac{2\pi j}{n_A}} \rho_j \biggr\|_{\tr} 
\end{equation}
with  the states $\rho_i$ and probabilities $\eta_i$ given by
Eq.~(\ref{eq-state_Q_discrimination}).
This proves the general expression of $D_\Bu^R ( \rho)$ anticipated in
Section~\ref{sec-main_results_general_expressions_Bures_Hellinger}.

\begin{theorem}\label{propbur}
If $A$ is a qubit ($n_A=2$) and $B$ is a qudit ($n_B \geq 2$), then the Bures discord of
response is related to the Bures geometric discord by
\begin{equation}
D^R_{\Bu}  (\rho)  = g (D^G_{\Bu}  (\rho)) \; ,
\label{burrg}
\end{equation}
with $g(d) = 2 d - d^2/2$.
\end{theorem}

\proof
The proof relies on Eq.~(\ref{eq-variationnal_formula_bis}) and the Helstrom formula~(\ref{eq-Helstom_formula}).
Accordingly, the maximum
fidelity $F ( \rho, \classQ ) = ( 1 - D_\Bu^G ( \rho)/2)^2$ between $\rho$ and a classical-quantum state is given for $n_A=2$ by
\begin{equation}
\label{eq-max_fidelity_to_CQ_qubit}
F ( \rho, \classQ ) = \onehalf \Bigl( 1 +  \max_{\{ \ket{\alpha_i}\}} \|   \eta_2 \rho_2 - \eta_1 \rho_1 \|_{\tr} \Bigr) \; .
\end{equation}
For $n_A=2$, the two operators inside the trace norms in
Eqs.~(\ref{eq-fidelity_of_response})
and~(\ref{eq-max_fidelity_to_CQ_qubit}) coincide.
\finpro

\vspace{2mm}

An optimal upper bound on the measurement-induced geometric discord
$D_\Bu^M$ in terms of $D^R_\Bu$ is given by:

\begin{theorem}\label{theo-upper_bound_D_Bu^M_intermsof_D_Bu^R}
If $A$ is a qubit or a qutrit ($n_A=2$ or $n_A=3$) and $B$ is a qudit
($n_B \geq 2$), the Bures
measurement-induced geometric discord and the Bures discord of response satisfy
\begin{equation} \label{eq-upper_bound_D_Bu^M_intermsof_D_Bu^R}
D_\Bu^M ( \rho) \leq
2 - \frac{2}{\sqrt{n_A}} \sqrt{ 1 + (n_A-1) ( 1 - D_\Bu^R ( \rho)
  )^2} \; .
\end{equation}
The bound is saturated for pure states. For $n_A>3$, the following
weaker bound holds:
\begin{equation}  \label{eq-upper_bound_D_Bu^M_intermsof_D_Bu^R_weaker}
D_\Bu^M ( \rho) \leq
2 - \frac{2}{\sqrt{n_A}} \big(  1 - D_\Bu^R ( \rho) \big) \; .
\end{equation}
\end{theorem}

\proof  The second bound is valid for
any space dimension $n_A$. It follows from the expressions for $D_\Bu^M(\rho)$ and $D_\Bu^R ( \rho)$ given
in
Eqs.~(\ref{eq-max_fidelity_meas_ind_geo_disc}) and~(\ref{eq-fidelity_of_response}), and from the  inequality
in Eq.~(\ref{eq-trace_inequality}).
We now show that when $n_A=2$ or $3$, the
stronger bound of Eq.~(\ref{eq-upper_bound_D_Bu^M_intermsof_D_Bu^R}) holds.
In view of  Eqs.~(\ref{eq-max_fidelity_meas_ind_geo_disc}) and~(\ref{eq-fidelity_of_response}), it is enough to show
that
\begin{equation} \label{eq-proof_upper_bound_D_Bu^M_intermsof_D_Bu^R}
n_A \left\| \sqrt{\sum_{j=1}^{n_A} \eta_j^2 \rho_j^2} \right\|_{\tr}^2
\geq
1 + (n_A -1)
\left\| \sum_{j=1}^{n_A}  e^{- \I \frac{2 \pi j}{n_A}} \eta_j \rho_j \right\|_{\tr}^2 \; .
\end{equation}
Let us consider the operators
\begin{equation} \label{eq-def_operators_A_k}
A_k =  \sum_{j=1}^{n_A}  e^{- \I \frac{2 \pi j k}{n_A}} \eta_j
  \rho_j \quad , \quad k=1,\ldots,n_A\;.
\end{equation}
Then $\sum_{k=1}^{n_A} | A_k|^2 = n_A
\sum_{j=1}^{n_A} \eta_j^2 \rho_j^2$.
We now make use of the inverse triangle inequality
$(\tr \sqrt{|A|^2+|B|^2} )^2\geq  (\tr |A| )^2 + ( \tr |B|)^2$ (see
e.g. Ref.~\cite{Bathia2000}, Lemma~1), which follows from  a
majorization argument (see Ref.~\cite{Bathia}, Exercise II.1.14 and Theorem~II.3.1),
the concavity of the square root function, and Minkowski's inequality
for sequences.
Thanks to the inverse triangle inequality,
we obtain
\begin{equation} \label{eq-inverse_triangle_inequality}
n_A \left\| \sqrt{\sum_{j=1}^{n_A} \eta_j^2 \rho_j^2} \right\|_{\tr}^2
 = \left\| \sqrt{\sum_{k=1}^{n_A} \big| A_k \big|^2 } \right\|_{\tr}^2
 \geq \sum_{k=1}^{n_A} \Big\| A_k \Big\|_{\tr}^2
 = 1 + \sum_{k=1}^{n_A-1}  \Big\| A_k \Big\|_{\tr}^2 \; ,
\end{equation}
where we have used
$A_{n_A} = \sum_j \eta_j \rho_j =\rho$
in the last equality. For $n_A=2$, the bound in
Eq.~(\ref{eq-proof_upper_bound_D_Bu^M_intermsof_D_Bu^R}) follows
immediately from Eqs.~(\ref{eq-def_operators_A_k}) and~(\ref{eq-inverse_triangle_inequality}).
For $n_A=3$, we exploit the identity $A_2=A_1^\dagger$
and the property $\| A^\dagger\|_{\tr} =\| A \|_{\tr}$ of the
trace norm. 

It remains to show that the bound is saturated for pure states $\rho_\Psi = \ketbra{\Psi}{\Psi}$.
In view of Eqs.~(\ref{eq-Bures_disc_resp_pure_states})
and~(\ref{eq-entanglement_of_response}), the Schmidt decomposition
of $\ket{\Psi}$ (see Eq.~(\ref{eq-Schmidt_decomposition})),
and the spectral decomposition of $U_A$ (see Eq.~(\ref{eq-unitaries_harm_spectrum})), we obtain
\begin{equation} \label{eq-proof_upper_bound_D_Bu^M_intermsof_D_Bu^R_pure1}
1 + (n_A-1) \big( 1 - D_\Bu^R (\ket{\Psi})\big)^2
= \max_{\{ \ket{\alpha_i} \}} \sum_{i,j,k,l} \mu_i \mu_k
\biggl( 1 + (n_A-1) \cos \Big( \frac{2\pi (j-l)}{n_A} \Big) \biggr)
\big| \braket{\varphi_i}{\alpha_j} \big|^2  \big| \braket{\varphi_k}{\alpha_l} \big|^2\;.
\end{equation}
If $n_A = 2$ or $3$, only the terms $j=l$ contribute to the sum.
Thanks to Eqs.~(\ref{eq-meas_ind_geo_disc_pure_states}) and~(\ref{eq-almost_finished}),
the r.h.s. of
Eq.~(\ref{eq-proof_upper_bound_D_Bu^M_intermsof_D_Bu^R_pure1}) is
equal to $n_A K ( \ket{\Psi})^{-1}=n_A ( 1 -D_\Bu^M (\ket{\Psi}) /2)^2$. Hence the inequality in
Eq.~(\ref{eq-upper_bound_D_Bu^M_intermsof_D_Bu^R}) is saturated for pure states.
\finpro

\subsection{Trace discord of response} \label{sec-trace_HS}

\begin{theorem}\label{qtmeqdisc}
If party $A$ is a qubit ($n_A=2$) and party $B$ is a qudit ($n_B \geq 2$), the trace discord of response,  trace
geometric discord, and  trace measurement-induced geometric discord
all coincide:
\begin{equation}
D^{R}_{\tr}(\rho)=D^M_{\tr}(\rho)=D^G_{\tr}(\rho) \; .
\label{trrowno}
\end{equation}
Furthermore, one has
\begin{equation}
D_{\tr}^{R}(\rho)\geq D_{\HS}^{R}(\rho) \; .
\label{trhsineq}
\end{equation}
\end{theorem}

It is worth remarking that the bound in
Eq.~(\ref{trhsineq}) is stronger than the trivial bound $2D^{R}_{\tr}(\rho) \geq D^{R}_{\HS}(\rho)$ that follows from
Eq.~(\ref{univineq}). Moreover, it is saturated for pure states  (see Sect.~\ref{eq-results_bounds_qubit}).

The proof of Theorem~\ref{qtmeqdisc} is reported in
Appendix~\ref{app-proof_theo5-7}. It makes explicit use of the fact that the
spectrum of the unitary operators in the definition
of the discord of response is given by the roots of the unity
(harmonic spectrum). The identity $D^M_{\tr} = D^G_{\tr}$ and the fact
that it holds only for $n_A=2$ have been originally established in
Ref.~\cite{Nakano2013}.

The trace geometric discord has been computed in different works for
specific classes of two-qubit states: a closed expression for
Bell-diagonal states has been found in Refs.~\cite{Nakano2013,Paula2013}, and it was further generalized to the class of the
so-called $X$ states and to the quantum-classical states in Ref.~\cite{Ciccarello2014}. Due
to Theorem~\ref{qtmeqdisc}, we can immediately extend these
results to the trace discord of response.

\section{Maximally quantum-correlated states}\label{mqcs}

\begin{figure}
\includegraphics[width=17cm]{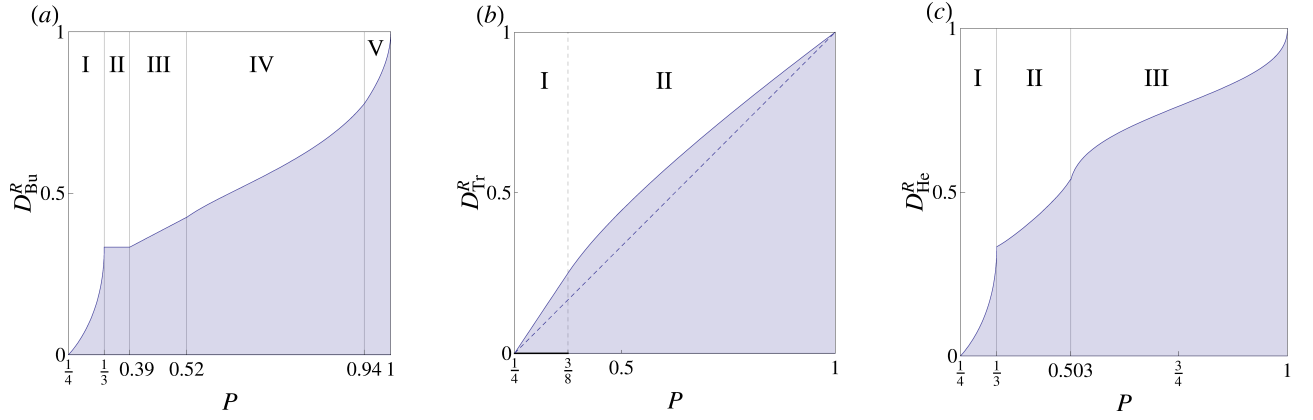}
\caption{(a) Bures discord of response accessible for the two-qubit
  states $\rho$ with purity $P=\tr\rho^2$  (from Ref.~\cite{Roga2014}). The
  possible values of $D^R_{\Bu} (\rho )$ as $P$ varies between $1/4$
  and $1$  are represented by the shadowed
  region. The solid line bounding this region on its upper side gives
  the discord of response of the maximally quantum-correlated states
  with respect to $D^R_{\Bu}$. These states are determined explicitly in
  Ref.~\cite{Roga2014} and have different forms in the five regions
  (I)-(V) delimited by the vertical lines. (b)~Same for the trace discord
  of response $D^R_{\tr}$. The values of $D^R_{\tr}$ for the Werner states are
  represented here by the dashed line (these values are the same for
the two branches in Eq.~(\ref{eq-Werner_state}) for $P \in
[1/4,1/3]$). The maximally
  quantum-correlated states are the  Bell-diagonal states defined in
  Eqs.~(\ref{state1}) and~(\ref{state2}) in
  Appendix~\ref{app-max_Q_corr_states} for the two regions (I) and
  (II) delimited by the vertical line. (c) Same for
  the Hellinger discord of response $D^R_{\Hel}$. The maximally
  quantum-correlated states are characterized in  Appendix~\ref{app-max_Q_corr_states}.}
\label{bup}
\end{figure}

In this Section we study the maximally quantum-correlated states with
respect to various discords of response.
With the help of numerical investigations,
we determine the two-qubit states $\rho$ with a fixed purity
$P=\tr \rho^2$ having the highest trace discord of response
and  Hellinger  discord of response.
In this way, we complete the previous analysis carried
out in Refs.~\cite{Giampaolo2013} and~\cite{Roga2014},
respectively for $D_\HS^R$ and $D_\Bu^R$. We find that the four discords of response
yield different families of maximally quantum-correlated states with
fixed purity $P<1$.
Nevertheless, if $P$ is not
fixed,  $D_{\tr}^R$, $D_{\Bu}^R$, and $D_{\Hel}^R$ take their maximal value
(equal to unity) for the same class of states, namely the maximally
entangled states, in agreement with Axiom (v) of
Section~\ref{sec-intro}.

\subsection{States with a fixed purity maximizing the discords of response}

We restrict here
our analysis to the case of two qubits $A$ and $B$ and use a numerical
approach. We have computed the values of $D^R_{\tr}$, $D^R_\Bu$, and $D^R_\Hel$ for
$3 \times 10^5$ randomly generated two-qubit states, whose
eigenvalues are chosen randomly with a uniform distribution (with the
constraint that they are non-negative and sum up to unity) and
eigenvectors are  the column vectors of a random unitary
matrix distributed according to the Haar measure.
We identify among these random states those
with purity $P$ maximizing the various discords of response. These
families of most discordant states are tested by applying small
disturbance analysis.

Our numerical analysis indicates that the states $\rho^P_\mmax$ with purity $P$ maximizing the
trace discord of response are given by Eqs.~(\ref{state1})  and~(\ref{state2}) in Appendix~\ref{app-max_Q_corr_states}.
Since these states  belong to the class of Bell-diagonal
states, their trace geometric discord can be evaluated by using the
results of Ref.~\cite{Nakano2013} and one can take advantage of
$D_{\tr}^R=D_{\tr}^G$ (recall that we are considering two qubits) to obtain:

\begin{conjecture}\label{prop_trmax}
{\rm The maximal value of $D^{R}_{\tr}(\rho)$ for all two-qubit states
$\rho$ with
purity $P = \tr \rho^2$ is given  by the following function of $P$:
\begin{equation} \label{eq-maximal_disc_resp_for_purity_P}
 D_{\tr}^{R}(\rho^{P}_\mmax)
=
\begin{cases}
\frac{1}{2} (4 P-1) & \text{ if }\;1/4 \leq P \leq 3/8 \\
\frac{1}{9} \left(\sqrt{6 P-2}+1\right)^2 & \text{ if }\;3/8 \leq P \leq 1\;.
\end{cases}
\end{equation}
}
\end{conjecture}

This result, which relies on our conjecture for the maximally
quantum-correlated states with respect to $D_{\tr}^R$, Eqs.~(\ref{state1}) and~(\ref{state2}),
is derived  for completeness in Appendix~\ref{app-max_Q_corr_states} without relying on the results of Ref.~\cite{Nakano2013}.
One can find in a similar way the values of
$D_{\tr}^{R}$ for the Werner states
\begin{equation} \label{eq-Werner_state}
\rho_{\rm W} =a_\pm \frac{\idty}{4}+(1-a_\pm)\ket{\Psi_-}\bra{\Psi_-}
\quad , \quad 
a_\pm = \begin{cases} 1 + \sqrt{(4P-1)/3} & \text{for $P \in   [1/4,1/3]$} \\ 
                 1-  \sqrt{(4P-1)/3} & \text{for $P \in    [1/4,1]$} 
   \end{cases}
\;,
\end{equation}
where $\ket{\Psi_-}= (\ket{01}- \ket{10})/\sqrt{2}$
denotes the Bell state.
The Werner states  maximize the Hilbert-Schmidt
discord of response at fixed purity $P$ (for $P \in
[1/4,1/3]$ this is true for  the two branches
of Werner states in Eq.~(\ref{eq-Werner_state}), which yield to the same value
of $D^R_\HS ( \rho_{\rm W})$), see Ref.~\cite{Giampaolo2013}.
We display in panel (b) of Fig.~\ref{bup} the accessible values of
$D_{\tr}^R(\rho)$ as a function of $P$ (shadowed region).
We clearly see in this figure that if $P$ is distinct from the
smallest and highest possible purities $P=1/4$
and $P=1$ (\ie, if $\rho$ is neither a maximally mixed nor a pure state), the
trace discord of response of the Werner state
(dashed line) is below the maximal
value of $D_{\tr}^R$ given by
Eq.~(\ref{eq-maximal_disc_resp_for_purity_P}). This
shows that the maximally quantum-correlated states with fixed purity $P$ are
different for the two measures $D_{\tr}^R$ and $D_\HS^R$.

A similar statement holds for the other discords of response.
The maximal values of $D_\Bu^R$ and $D_\Hel^R$
at fixed purity $P$ cannot be characterized by such simple
functions as in Eq.~(\ref{eq-maximal_disc_resp_for_purity_P}),
therefore we do not report them here.
The maximal Bures discord of response has been determined explicitly
as a function of $P$ in
  Ref.~\cite{Roga2014} and is shown in panel (a) of Fig.~\ref{bup}.
By using  the numerical method described above, we have identified the
 two-qubit states with purity $P$ maximizing the
Hellinger discord of response, which are given in Appendix~\ref{app-max_Q_corr_states}.
From the analytical expression given in
Eq.~(\ref{eq-closed_formula_discresp_Hel}), it is then easy to
compute numerically the maximal Hellinger discord of response  as  a
function of $P$. The latter is represented by the upper solid line in
panel (c) of Fig.~\ref{bup}.

In Fig.~\ref{hstr} we report the distributions in the planes defined
by pairs of discords of response associated to different
distances for random two-qubit states. These distributions are
analogous to those of Fig.~\ref{responsegrid} except that they correspond here to
a fixed purity $P=0.6$. Random states with this purity are
generated as described above: their eigenvectors are obtained from
random unitary matrices distributed according to the Haar measure, while
their eigenvalues are picked randomly from the set of non-negative numbers
$p_i$ with fixed sums $\sum_i p_i =1$ and  $\sum_i p_i^2 =0.6$.
For states of rank $r>2$, we first generate $r-2$ random eigenvalues with a
uniform distribution on  sufficiently small intervals. The remaining
eigenvalues are given by the constraints on the trace and the purity.

Since there is no exact relation between the discords of response
associated to different
distances, the points in Fig.~\ref{hstr} are not distributed on a line but in a region of the plane with a
non-vanishing area. This means that  the different discords of response do not define the
same ordering on the set of bipartite states: for instance, it is possible to find
two states $\rho_1$ and $\rho_2$ which satisfy $D^R_{\tr}(\rho_1) <
D^R_{\tr}(\rho_2)$ and at the same time $D^R_{\Bu}(\rho_1) >
D^R_{\Bu}(\rho_2)$ (see \eg the points lying on  the thick line in Fig.~\ref{hstr}(b)).
In other words, changing the distance
modifies the ordering of the states. In Fig.~\ref{hstr}, the different locations of the
points having the highest discord of response for the different distances
illustrate the fact that the maximally quantum-correlated states with purity
$P<1$ are not the same for $D^R_{\tr}$, $D^R_\Bu$, and $D^R_\Hel$.

We have observed a similar behavior as in Fig.~\ref{hstr} for the
trace, Bures, and Hellinger measurement-induced geometric discords
(not reported here).

\begin{figure}
\includegraphics[width=17.0cm]{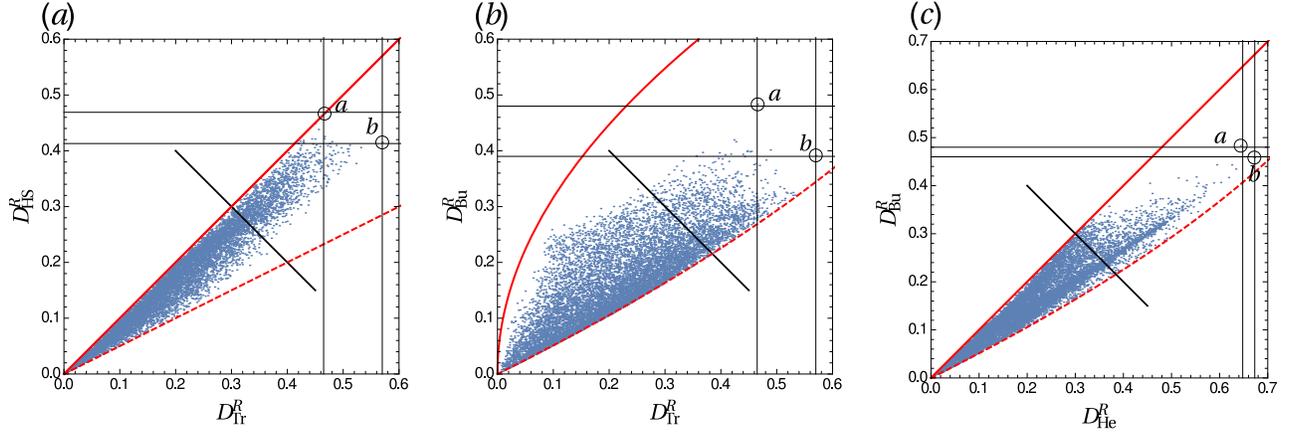}
\caption{(a) Values of the Hilbert-Schmidt and trace  discords of response
  $D^R_\HS$ and $D^R_{\tr}$ for
  $10^4$ random two-qubit states of different ranks with the same fixed
  global state purity $P=0.6$. The points $a$ and $b$ represent,
  respectively, some maximally quantum-correlated states  with purity
  $P$ with respect
  to $D^R_\HS$ and $D^R_{\tr}$. Note that $a$ is not
  maximally quantum correlated with respect to $D^R_{\tr}$, and
  similarly for $b$ and $D^R_\HS$. The red solid and dashed lines represent
  the bounds discussed in Sec.~\ref{eq-results_bounds_qubit} like in Fig.~\ref{responsegrid}.
  States on the thick black line have a hierarchy with respect to
  $D^R_{\HS}$ that is reversed compared to the hierarchy with respect
  to $D^R_{\tr}$. (b) Same for the Bures and trace discords of response $D^R_\Bu$
  and $D^R_{\tr}$.  (c) Same for  the Bures and Hellinger discords of response $D^R_\Bu$
  and $D^R_{\Hel}$.}
\label{hstr}
\end{figure}

\subsection{States with arbitrary purity maximizing the discords of
  response}
\label{sec-axiom(v)}

In spite of the annoying fact that the maximally quantum-correlated
states with a fixed purity depend on the distance used to define the
discord of response, a universal family of states emerges when one
considers the maximal value of $D^R$ irrespective of the purity
$P$. As it should be expected
for any proper measure of quantum correlations, for the trace, Bures,
and Hellinger distances these maximally discordant states are
the maximally entangled states, that is, the states $\rho$ with
largest entanglement of formation $\EoF ( \rho) = \ln ( \min\{ n_A,
n_B\})$. Let us recall that if $n_A \leq  n_B$,  such  states
are convex combinations of pure states of the form:
\begin{equation} \label{eq-max_entangled_state}
\ket{\Psi_k} = \frac{1}{\sqrt{n_A}} \sum_{i=1}^{n_A} \ket{\varphi_{ki}} \ket{\chi_{ki}}
\end{equation}
with $\braket{\varphi_{ki}}{\varphi_{kj}} = \delta_{ij}$ and
$\braket{\chi_{ki}}{\chi_{lj}} = \delta_{ij} \delta_{kl}$ (see, e.g.,
Ref.~\cite{Spehner_review}, Proposition 9.E.1). Notice that the subspaces $\Span \{
\ket{\chi_{ki}}, i = 1 , \ldots, n_A\}$ are orthogonal for different
$k$'s, so that the aforementioned convex combinations involve at most
$r$ pure states $\ket{\Psi_k}$ if $r n_A \leq n_B < (r+1) n_A$. In
particular, if $n_A \leq n_B < 2 n_A$ then the  maximally entangled
states are necessarily pure states given by
Eq.~(\ref{eq-max_entangled_state}). The following theorem is proven in Appendix~\ref{app-maximal_discod_resp}.

\begin{theorem} \label{theo-axiom_v_disc_resp}
Let subsystems $A$ and $B$ have arbitrary space dimensions
$n_A$ and $n_B$, with $n_A \leq n_B$. Then the maximal value of the
trace, Bures, and Hellinger discords of response is equal to unity and
these three discords of response satisfy Axiom (v)  of Section~\ref{sec-intro},
namely, $D^R ( \rho) = 1$ if and only if
$\EoF ( \rho) = \ln n_A$. In contrast, if $2 n_A \leq n_B$ then the
Hilbert-Schmidt discord of response does not enjoy this property.
\end{theorem}

It is shown in  Appendix~\ref{app-maximal_discod_resp} that the
trace, Bures, and Hellinger geometric discords and measurement-induced
geometric discords satisfy Axiom (v) as well, at least when $A$ is a qubit,
as reported in Table~\ref{tab3} (see also Theorem~\ref{theo-Bures_meas_ind_disc_pure_states}).

\section{Quantumness breaking channels}\label{qbc}

Quantum channels (also called quantum operations) are by definition
dynamical maps on the set of quantum states which are completely
positive (CP) and trace-preserving (TP). Let us recall that a linear map
$\Phi_A$ acting on the set of states of a system $A$ is CP if its trivial extensions $\Phi_A\otimes\idty_B$
are positive (\ie, they transform non-negative matrices into non-negative
matrices)  for any system $B$ with finite-dimensional Hilbert space $\Hh_B$.
A TP linear map $\Phi_A$ is CP (and hence is a quantum channel) if and only if~\cite{Bengtsson}
\begin{equation}
\rho^{\Phi_A} \equiv \Phi_A \otimes\idty_B (\ket{\Psi_+}\bra{\Psi_+})
\end{equation}
is non-negative. Here, we have introduced a system $B$ having the same
space dimension $n_B=n_A$ as $A$, a fixed orthonormal product basis
$\{ \ket{i_A  , j_B } \equiv \ket{i}_A \otimes \ket{j}_B \}_{i,j=1}^{n_A}$ for the
composite system $AB$, and the maximally entangled state
$\ket{\Psi_+} = \frac{1}{\sqrt{n_A}}\sum_{i=1}^{n_A}\ket{i_A,  i_{B}}$
of $AB$.
The state $\rho^{\Phi_A}$ is called the
Jamio{\l}kowski  state corresponding to the CPTP map $\Phi_A$.

In Ref.~\cite{Korbicz2012}, the authors have characterized the local
quantum channels $\Phi_A$ acting on system $A$
that destroy all the quantum correlations existing in an arbitrary
bipartite quantum state of $AB$. Such channels are called
{\em  quantumness breaking channels} and are such that the output state
$\Phi_A \otimes \idty_B ( \rho_{AB} )$ is a classical-quantum state
for any bipartite input state $\rho_{AB}$. It turns
out that a channel $\Phi_A$ is quantumness breaking if and
only if its Jamio{\l}kowski  state $\rho^{\Phi_A}$ is
classical-quantum (see Ref.~\cite{Korbicz2012}).

In this Section, we derive a user-friendly, necessary condition
for a channel to be quantumness breaking. This condition is formulated in terms of
the rank of the superoperator $\widehat{\Phi}_A$ associated to the
quantum channel $\Phi_A$. This superoperator is defined as the operator on the
tensor-product (doubled) Hilbert space $\Hh_A \otimes \Hh_A$ with  matrix elements
\begin{equation} \label{eq-superoperator_form}
\bra{i_A, j_A} \widehat{\Phi}_A \ket{k_A, l_A} \equiv \bra{i_A} \Phi_A ( \ketbra{k_A}{l_A} ) \ket{j_A} \; .
\end{equation}
If we represent the states $\rho_A$ of $A$ as vectors $\ket{\rho_A}$
on $\Hh_A \otimes \Hh_A$ with components
$\braket{i_A,j_A}{\rho_A}=\bra{i_A}\rho_A \ket{j_A}$, then
$\widehat{\Phi}_A$ realizes the transformation of these vector-states under the channel $\Phi_A$.

Our treatment relies on the so-called
{\em reshuffling operation} $\c R$, which is a widely employed tool in
the theory of quantum channels, see, e.g., Ref.~\cite{Bengtsson}. The
operation $\c R$  exchanges the matrix entries of a block matrix in
the following way: given an operator $X$ on $\Hh_A \otimes \Hh_B$, one
associates to it the operator $X^{\c R}$ from
$\Hh_B \otimes \Hh_B $ to $\Hh_A \otimes \Hh_A$ with matrix elements
\begin{equation}
\bra{i_A, j_A} X^{\c R} \ket{k_B , l_B} \equiv   \bra{i_A , k_B} X \ket{j_A , l_B} \; .
\label{reshuf}
\end{equation}
Thus, the reshuffling operation transforms a $n_A n_B\times n_An_B$ matrix onto a $n_A^2\times n_B^2$ matrix and vice versa. It provides a connection between the superoperator $\widehat{\Phi}_A$ associated to the quantum channel ${\Phi}_A$ and the corresponding Jamio{\l}kowski state thanks to the following relation~\cite{Bengtsson}:
\begin{equation} \label{eq-relation_superop_Jamiolkowski}
\big( \widehat{\Phi}_A \big)^{\c R} =n_A\rho^{\Phi_A} \; .
\end{equation}

Our necessary condition for a channel to be quantumness breaking
is based on  the following theorem:

\begin{theorem}\label{thm99}
For any state $\rho$ of a bipartite quantum system $AB$, the Hilbert-Schmidt geometric discord $D^{G}_{\HS}(\rho)$ is bounded from below in the following fashion:

\begin{equation} \label{eq-inka}
D^{G}_{\HS}(\rho) \geq \mu_{n_A+1}+\ldots + \mu_{n_A^2} \; ,
\end{equation}
where $\mu_{1} \geq \cdots \geq \mu_{n_A^2}$ are the squared singular values of the reshuffled density matrix $\rho^{\c R}$ in decreasing order.
\end{theorem}

Recall that the squared singular values of $\rho^{\c R}$ are the
eigenvalues $\mu_i$ of the $n_A^2 \times n_A^2$ matrix $\rho^{\c R} (\rho^{\c R})^\dagger$.
In Appendix~\ref{app-lower_bound_disc_resp_HS} we prove
that the inequality in Eq.~(\ref{eq-inka}) turns into an  equality
provided that party $A$ is a qubit and $\rho$ has
maximally mixed marginals $\rho_A = \idty/2$ and
$\rho_B = \idty/n_B$.   A bound similar to that of Eq. (\ref{eq-inka})
has been derived in Ref.~\cite{Rana2012}, where it was also found 
that this bound is saturated for Bell-diagonal states.

 It is worthwhile mentioning that the singular values of  $\rho^{\c
   R}$ appear  in the generalized Schmidt decomposition of
mixed states. Given an
arbitrary  density matrix $\rho$ on $\Hh_A \otimes \Hh_B$, this decomposition reads
 (see, e.g., Refs.~\cite{Bengtsson,Spehner_review})
\begin{equation}
\rho = \sum_{m=1}^{n^2_A} \sqrt{\mu_m} X_m \otimes Y_m \; ,
\end{equation}
where $\{ X_m\}_{m=1}^{n_A^2}$ and $\{ Y_p\}_{p=1}^{n_B^2}$ are
orthonormal bases of the Hilbert spaces formed by all $n_A \times n_A$ matrices
and all $n_B \times n_B$ matrices, respectively
(\ie , $\tr_A( X_m^\dagger X_n ) = \delta_{mn}$ and $\tr_B
(Y_p^\dagger Y_q) = \delta_{pq}$) and we have assumed  $n_A\leq
n_B$. The matrices $X_m$ and $Y_p$
are given in terms of the eigenvectors of $\rho^{\c R} (\rho^{\c R})^\dagger $
and $(\rho^{\c R})^\dagger \rho^{\c R} $, respectively.
Note that $\sum_m \mu_m$ coincides with the state purity
$P = \tr \rho^2$. Moreover, $\sum_m \sqrt{\mu_m} >1$ is a
sufficient (but not necessary)
condition for entanglement~\cite{Chen03}. Analogously, it follows from
Theorem~\ref{thm99}  that $\sum_{m>n_A} \mu_m >0$ is a sufficient
condition for $\rho$ to be quantum correlated. Indeed,
although $D^{G}_{\HS}$ is not a proper measure of quantum correlations, it
satisfies Axiom (i) of Section~\ref{sec-intro} (see Sec.~\ref{sec-def_geo_discords_proper}).

In view of the relation~(\ref{eq-relation_superop_Jamiolkowski}) and of the aforementioned
characterization of quantumness breaking channels in
Ref.~\cite{Korbicz2012}, one deduces from Theorem~\ref{thm99} the following result:

\begin{corollary}
If the rank of $\widehat{\Phi}_A$ is larger than $n_A$, then the channel $\Phi_A$ is not quantumness breaking.
\end{corollary}

Theorem~\ref{thm99} actually provides a quantitative estimate which can
be used to discriminate  channels
that are not quantumness breaking, since it gives a lower bound on
the amount of quantum correlations that survive after the
action of a local quantum channel  $\Phi_A$ if the input state is the maximally entangled state
$\ket{\Psi_+}$. Such residual amount, as measured e.g. by the trace
geometric discord $D^G_{\tr}$, cannot be smaller than the sum of the $n_A^2-n_A$ smallest squared
singular values of $\widehat{\Phi}_A/n_A$ (recall that $D^G_{\tr} \geq D_\HS^G$, see Eq.~(\ref{eq-trivial_bounds})).

\vspace{2mm}

\Proofof{Theorem~\ref{thm99}}
Since the reshuffling procedure consists only in exchanging matrix
entries, it neither changes the Hilbert-Schmidt norm of a matrix, nor
the Hilbert-Schmidt distance between two matrices which are both
reshuffled by $\c R$. Observe that the reshuffling operation
transforms a classical-quantum  state $\sigma_\Aclas$
into a matrix of rank equal to the dimensionality $n_A$ of subsystem
$A$. More precisely, consider a classical-quantum state
$\sigma_\Aclas$ given by Eq.~(\ref{eq:cq}). Then
$\Pi \sigma_\Aclas^{\Rr}= \sigma_\Aclas^{\Rr}$, where $\Pi$ is the
projector of rank $n_A$ defined by
$\Pi= \sum_i \ketbra{\alpha_i}{\alpha_i} \otimes
\ketbra{\overline{\alpha}_i}{\overline{\alpha}_i}$ (the bars denote
complex conjugation in the standard basis, \ie,
$\braket{j_A}{\overline{\alpha}_i}= \braket{j_A}{\alpha_i}^\ast$ for
any $j_A=1,\ldots, n_A$). We will now estimate the geometric discord
$D^G_\HS ( \rho)$ by the Hilbert-Schmidt distance from $\rho^{\c R}$
to the nearest  $n_A^2 \times n_B^2$ matrix  $M_{n_A}$ of rank
$n_A$:
\begin{equation}
D_\HS^G ( \rho) = \min_{\sigma_\Aclas } \left\|\rho -\sigma_\Aclas \right\|_{\HS}^2
\geq\min_{M_{n_A}} \left\| \rho^{\c R}-M_{n_A}\right\|_{\HS}^2 \; .
\label{chain2}
\end{equation}
On the other hand, by the low-rank approximation theorem (see, e.g.,
Ref.~\cite{HJ}) one has
\begin{eqnarray}
\label{eq-proof_theo9bis}
\nonumber
\left\| \rho^{\c R}-M_{n_A}\right\|_{\HS}^2
& = & \tr \big(  \rho^{\c R}  - M_{n_A} \big)^\dagger \Pi \big( \rho^{\c R}  - M_{n_A} \big)
 + \tr \big( \rho^{\c R}  - M_{n_A} \big)^\dagger ( 1- \Pi) \big( \rho^{\c R}  - M_{n_A} \big)
\\
\nonumber
& = &
 \left\| \Pi \rho^{\c R}-M_{n_A}\right\|_{\HS}^2 + \tr \rho^{\c R} \big(  \rho^{\c R} \big)^\dagger  (1-\Pi )
\\
& \geq &
\tr \rho^{\c R} \big( \rho^{\c R} \big)^\dagger  (1-\Pi) \; \geq \;
 \mu_{n_A+1}+\cdots + \mu_{n_A^2} \; ,
\end{eqnarray}
where $\Pi$ is the projector of rank $n_A$ on the range of $M_{n_A}$
and $\mu_m$ are the eigenvalues of
$\rho^{\c R} \big(  \rho^{\c R} \big)^\dagger$ (\ie, the squared singular values of $\rho^{\c R}$) in
decreasing order. The last inequality comes from the min-max
principle (see, e.g., Ref.~\cite{Bathia}). Comparing Eqs.~(\ref{chain2})
and~(\ref{eq-proof_theo9bis}), the desired  bound~(\ref{eq-inka})
follows. Note that, according to Eq.~(\ref{eq-proof_theo9bis}), this inequality turns into an  equality if
$\Pi^\opt \rho^{\c R} = \sigma_{\HS, \rho}^{\Rr}$, where $\Pi^\opt$ is
the projector on the sum of the eigenspaces associated to the largest
$n_A$ eigenvalues of
$\rho^{\c R} \big( \rho^{\c R} \big)^\dagger$ and $\sigma_{\HS, \rho}$ is the closest
classical-quantum state to $\rho$  (as measured by the Hilbert-Schmidt
distance), which is given by Eq.~(\ref{eq-closest_states_meas_ind_geo_disc_HS}).
\finpro

\section{Discussion and conclusions}\label{som}

One of the most relevant results of the present study is that the
Hellinger geometric discord $D^G_\Hel $ and the Hellinger discord of
response $D^R_\Hel$ provide the first instance of {\em bona fide}  measures of
quantum correlations which are at the same time easy to compute (see
Eqs.~(\ref{eq-explicit_formula_geo_disc_2_qubits})
and~(\ref{eq-closed_formula_discresp_Hel})), satisfy all  the axiomatic criteria for
proper measures of quantum discord given in the Introduction, and enjoy clear
operational interpretations  in quantum protocols. 
They thus satisfy all the fundamental requirements of
computability, reliability, and operational viability.

Let us briefly discuss the operational interpretations of  these two
measures and of the other geometric
measures of quantum discord studied in this work. If the
reference subsystem $A$ is a qubit, the Hellinger discord of response
coincides with the LQU. The latter indeed enjoys
a simple operational interpretation described  in Sec.~\ref{sec-def_discords}.
Thanks to the relation between $D^R_\Hel$ and $D^G_\Hel$ provided by 
Eq.~(\ref{eq-disc_resp_Hel_n=2,3}), the same holds for the  Hellinger geometric discord.
Besides this interpretation, $D^R_\Hel$ and $D^G_\Hel$ enjoy further operational meanings in terms of the
minimum probability of error in discriminating two equiprobable
quantum states if infinitely many copies can be used to distinguish
them~\cite{Audenaert2007}. 

In a one-shot scenario, the minimum
probability of error in discriminating two equiprobable
states $\rho_1$ and $\rho_2$ 
is given in terms of the trace norm $\| \rho_2 - \rho_1\|_{\tr}$ 
according to the Helstrom formula~\cite{Helstrom1976}, see
Eq.~(\ref{eq-Helstom_formula}). This formula
grants an operational meaning to all geometric measures of quantum
discord defined with the trace distance, for instance in the
context of quantum reading~\cite{Roga2015}.   In this protocol~\cite{Pirandola2011b}, the task
is to distinguish the output states of a quantum transmitter
which goes through an unknown binary cell changing 
the transmitter states. If the actions of the binary cells on these
states are given by the identity and
local unitary transformations with a harmonic spectrum, the minimum
probability of error maximized over all possible realizations of the
cells is a simple function of the trace discord of response $D^R_{\tr}$.
As we have shown in the present work, if subsystem $A$ is a qubit then $D^R_{\tr}$ coincides with
the trace geometric discord $D^G_{\tr}$ and the trace
measurement-induced geometric discord $D^M_{\tr}$. Hence, also these
last two measures enjoy direct operational
interpretations. Moreover, as shown in Ref.~\cite{Roga2015}, the
minimum probability of error in quantum reading features tight lower and upper bounds 
that are simple functions of the Bures and Hellinger discords of response $D^R_{\Bu}$ and $D^R_\Hel$.

Furthermore, the Bures geometric discord $D^G_{\Bu}$ enjoys a simple
operational interpretation in terms of the minimum probability of
error in the task of
discriminating selected states within a given ensemble, as discussed
in Ref.~\cite{Spehner2013}. Due to the relation obtained above between
$D^G_{\Bu}$ and $D^R_{\Bu}$ when party $A$ is a
qubit, see Eq.~(\ref{burrg}), this
operational interpretation extends also to the Bures discord of
response $D^R_{\Bu}$.

For the sake of completeness, we should also mention some further instances of active
research fields of quantum technologies in which geometric measures of quantum correlations
find interesting applications. These include protocols such as quantum
metrology with unknown disturbance~\cite{GaussianMetrology2014},
quantum illumination~\cite{Farace2014}, and entanglement distribution
between system and apparatus during a measurement process~\cite{Adesso2014}.

In conclusion, we have investigated the main properties of different
classes of geometric measures of quantum correlations. We have
characterized, quantified, and compared them for the most significant contractive
distances (trace, Bures, and Hellinger distances) and operations (geometric discord, measurement-induced
geometric discord, and discord of response). We have proven 
a variety of bounds and algebraic relations
between these geometric measures. The main results are summarized in
the synoptic tables of Section~\ref{sec-main_result}
(see Tables~\ref{tab1}-\ref{tab3}). 
Thanks to the one-to-one correspondence that we
have established between some of these measures, one can extend the
physical interpretation from one class of measures to other classes
that are in direct correspondence with the former. We have found that
direct one-to-one correspondences exist only in the
case of low-dimensional reference subsystems (qubit or
qutrit). Otherwise, in more general situations, we have established a
substantial set of inequalities, some of them being tight.

We have also established that different geometric measures of
quantum discord induce in general different orderings of the discordant states. This phenomenon is quite analogous to the
different ordering of quantum states  established by different
entanglement measures~\cite{Verstraete2003,Adesso2005}.
In particular, different measures of quantum discord identify different classes of maximally
quantum-correlated states with a fixed purity. On the other hand, the set of maximally quantum-correlated states with arbitrary purity
is independent of the choice of the distances and operations, and coincides with the
set of maximally entangled states. 

Finally, we have established a useful role also for the
Hilbert-Schmidt geometric discord $D^{G}_{\HS}$, notwithstanding that
it is strictly speaking  not a measure of quantum correlations due to the fact
that the Hilbert-Schmidt distance is not contractive under
CPTP maps. Indeed, $D^{G}_{\HS}$ provides useful bounds on
{\em bona fide} geometric measures based on contractive distances.
Furthermore, we have exploited this fact in order to determine a necessary
condition for local quantum channels to be quantumness breaking, namely,
to destroy all quantum correlations featured by arbitrary
input states. This condition can be formulated in terms of the Jamio\l kowski state that
corresponds to the given channel.

\begin{acknowledgments}
F.I. and W.R. acknowledge the EU FP7 Cooperation STREP Projects iQIT - integrated Quantum Information Technologies, Grant Agreement No. 270843, and EQuaM -
Emulators of Quantum Frustrated Magnetism, Grant Agreement
No. 323714. They also acknowledge financial support from the Italian
Minister of Scientific Research (MIUR) under the national PRIN
programme. D.S. acknowledges financial support from the ANR project
no. ANR-13-JS01-0005-01. We are grateful to S.~Rana for pointing us
the results of Ref.~\cite{Piani2014} on the trace measurement-induced
geometric discord.
\end{acknowledgments}

\appendix

\section{Proofs of Theorems~\ref{prophel} and~\ref{qtmeqdisc} of Section~\ref{sec-disc_of_response}}
\label{app-proof_theo5-7}

\subsection{Proof of Theorem \ref{prophel} and of some  inequalities in Table~\ref{tab3}}
\label{app-proof_theo5}

We first consider the Hilbert-Schmidt discord of response. We will
afterward make use of the second relation in Eq.~(\ref{eq-rel_disc_resp_HS_Hel})
to deduce the statements on $D^R_\Hel$ in
Theorem~\ref{prophel}.
According to the definition of $D^R$ in Eq.~(\ref{def:quantumn}), one has:
\begin{equation} \label{eq-HS_disc_resp_app-A}
2 D_{\HS}^R ( \rho) = \min_{U_A \in \Uu_\Lambda} \| \rho - U_A \otimes \idty \,\rho \,U_A^\dagger \otimes \idty \|_{\HS}^2 \; ,
\end{equation}
the minimum being over all unitaries $U_A$ acting on subsystem $A$
with spectrum given by the complex roots of unity.  Any
such unitary operator has the form given in Eq.~(\ref{eq-unitaries_harm_spectrum}).
Therefore,
\begin{equation} \label{eq-discord_resp_HS}
2 D_{\HS}^R ( \rho ) = \min_{ \{ \ket{\alpha_i} \} }
\biggl\| \sum_{i, j=1}^{n_A} \Bigl( 1 - e^{-\I \frac{2 \pi (i-j)}{n_A}} \Bigr) \ketbra{\alpha_i}{\alpha_j} \otimes \bra{\alpha_i} \rho \ket{\alpha_j} \biggr\|_{\HS}^2 \; .
\end{equation}
Now, the squared Hilbert-Schmidt norm of a block matrix is the sum of
the squared Hilbert-Schmidt  norms of each blocks, \ie,
$\| \sum_{i,j} \ketbra{\alpha_i}{\alpha_j} \otimes X_{ij} \|_{\HS}^2 =\sum_{i,j} \| X_{ij}\|_{\HS}^2$. This yields the relation
\begin{equation} \label{eq-discord_resp_HS_2}
D_{\HS}^R ( \rho) = 2 \min_{ \{ \ket{\alpha_i}\} }
\sum_{i\not= j} \sin^2 \Bigl(  \frac{\pi (i-j)}{n_A} \Bigr)  \bigl\| \bra{\alpha_i} \rho \ket{\alpha_j} \bigr\|_{\HS}^2 \; .
\end{equation}
Similarly, the Hilbert-Schmidt measurement-induced geometric discord takes the following expression:
\begin{equation} \label{eq-meas_ind_discord_HS}
D_{\HS}^M ( \rho) = \min_{ \{ \ket{\alpha_i}\} }
\sum_{i\not= j} \bigl\| \bra{\alpha_i} \rho \ket{\alpha_j} \bigr\|_{\HS}^2 \; .
\end{equation}
By comparing Eqs.~(\ref{eq-discord_resp_HS_2}) and~(\ref{eq-meas_ind_discord_HS}), we obtain that
\begin{equation}
D^{R}_{\HS}(\rho)=\left\{ \begin{array}{ll}
         2D^M_{\HS}(\rho) =  2D^G_{\HS}(\rho)  & \mbox{if $n_A=2$}  \\[5pt]
       \frac{3}{2}D^M_{\HS}(\rho)  = \frac{3}{2} D^G_{\HS}(\rho)   & \mbox{if $n_A = 3$} \; , \end{array} \right.
\label{drna}
\end{equation}
where we have made use of the equality between the Hilbert-Schmidt
geometric discord and the measurement-induced
geometric discord (see Section~\ref{sec-rel_geo_dic_HS_Hel}). For
$n_A > 3$, comparing Eqs.~(\ref{eq-discord_resp_HS_2}) and~(\ref{eq-meas_ind_discord_HS}) and observing that
\begin{equation}
\sin^2 \left( \frac{\pi}{n_A} \right)  \leq   \sin^2 \left( \frac{\pi (i-j)}{n_A} \right) \leq 1
\quad , \quad i,j=1,\ldots, n_A \; , \; i\not=j \; ,
\end{equation}
we obtain the following inequalities
\begin{equation} \label{thirdprop}
2 \sin^2 \left( \frac{\pi}{n_A} \right) D^G_{\HS}(\rho) \leq
D_{\HS}^{R}(\rho) \leq 2 D^G_{\HS}(\rho) \; .
\end{equation}
The above bounds remain valid for unnormalized non-negative operators
$\rho$. One may thus replace $\rho$ by $\sqrt{\rho}$ in
Eq.~(\ref{thirdprop}). In view of the relations between $D_\Hel^R$ and
$D_\HS^R$ and between $D_\Hel^G$ and $D_\HS^G$ (see Theorem~\ref{theo-rel_geo_disc_Hel_HS}), this yields the inequalities
given in Eq.~(\ref{bound_disc_resp_Hel}). Similarly,
the fundamental identities of Eq.~(\ref{eq-disc_resp_Hel_n=2,3}) follow
directly by substituting $\rho$ by $\sqrt{\rho}$ in Eq.~(\ref{drna}).
This concludes the proof of Theorem~\ref{prophel}.
\finpro

\vspace{1mm}

Let us point out  that the relation given in
Eq.~(\ref{drna}) between the Hilbert-Schmidt discord of response and
the Hilbert-Schmidt geometric discord when $n_A=2$
has been found earlier in Ref.~\cite{Giampaolo2013}. Moreover, the lower bound on $D_{\HS}^{R}(\rho)$
in Eq.~(\ref{thirdprop}) is finer than a previously known bound~\cite{Giampaolo2013}.

Several additional results stated in  the main text can be easily derived from the
above considerations. Firstly,
the general expression of the Hellinger discord of response
in Eq.~(\ref{eq-formula_Hellinger_discord_of_response}) is
obtained by replacing $\rho$ by its square root in Eq.~(\ref{eq-discord_resp_HS_2}).
Secondly, the bounds on the
discord of response in terms of the geometric discord for the trace
and Bures distances  reported in Table~\ref{tab3} are obtained
by combining Eq.~(\ref{thirdprop}) with the trivial bounds of
Eqs.~(\ref{eq-trivial_bounds}) and~(\ref{eq-trivial_boundsbis}).  More precisely,
the bounds on $D_{\tr}^R$ in terms of $D_{\tr}^G$ for $n_A\geq 3$
are consequences of  Eq.~(\ref{thirdprop}), the two inequalities
of Eq.~(\ref{eq-trivial_bounds}), and
the corresponding inequalities for the discord of response (with the
correct factor of two coming from the normalization factor $\c N$ in Eq.~(\ref{def:quantumn})).
The bounds on $D_\Bu^R$ in terms of
$D_{\Bu}^G$  are deduced from the previous bounds on $D_{\tr}^R$ by
using Eqs.~(\ref{eq-trivial_boundsbis}) and~(\ref{eq-trivial_bound_disc_response}).

\subsection{Proof of Theorem~\ref{qtmeqdisc}}
\label{app-proof_theo7}

We now turn to Theorem~\ref{qtmeqdisc} on the trace discord of
response. Indeed, $D_{\tr}^R$ is expressed by a formula analogous to
Eq.~(\ref{eq-discord_resp_HS}) with the Hilbert-Schmidt norm replaced by the
trace norm, excepted for a factor of $4$ instead of $2$
in the right-hand side. For $n_A=2$, setting
$X_{12} = \bra{\alpha_1} \rho \ket{\alpha_2}$, the expression for $D_{\tr}^R$ takes the
form 
\begin{eqnarray} \label{eq-dicord_response_qubit}
\nonumber
D_{\tr}^R ( \rho)
& = &
\min_{ \{ \ket{\alpha_i}\} }
\Bigl\| \ketbra{\alpha_1}{\alpha_2} \otimes X_{12} + \ketbra{\alpha_2}{\alpha_1} \otimes X_{12}^\dagger \Bigr\|_{\tr}^2
= \min_{ \{ \ket{\alpha_i}\} } \Bigl(
\tr \sqrt{  \ketbra{\alpha_1}{\alpha_1} \otimes | X_{12}^\dagger |^2  + \ketbra{\alpha_2}{\alpha_2}\otimes | X_{12}  |^2}
\Bigr)^2
\\
& = &
4 \min_{ \{ \ket{\alpha_i}\} } \| X_{12} \|_{\tr}^2 \; ,
\end{eqnarray}
where we have used the identity
$\| X \|_{\tr}=  \| X^\dagger \|_{\tr}$ in the last line. A similar calculation shows that the trace
measurement-induced geometric discord is given by
$D^M_{\tr}  ( \rho) = 4 \min_{ \{ \ket{\alpha_i}\} } \| X_{12} \|_{\tr}^2$. Thus we arrive at the
important equality
$D_{\tr}^R ( \rho) = D^M_{\tr}  ( \rho)$. 
Furthermore, it is proven in Ref.~\cite{Nakano2013}   
that $D_{\tr}^G ( \rho) = D^M_{\tr}  ( \rho)$ when $n_A=2$.
Finally, the bound $D_{\tr}^R ( \rho) \geq D_{\HS}^R ( \rho)$ follows from
Eqs.~(\ref{eq-discord_resp_HS_2}) and~(\ref{eq-dicord_response_qubit}) and from the trivial inequality $\| X\|_{\tr}
\geq \| X\|_{\HS}$.
\finpro

\section{Geometric measures of quantum correlations satisfying Axiom (v)}
\label{app-maximal_discod_resp}

In this Appendix we show that the Bures, Hellinger, and trace discords
of response satisfy Axiom (v) of Section~\ref{sec-intro}  and that their
maximum value equals unity, as stated in
Theorem~\ref{theo-axiom_v_disc_resp}. We prove as well that the
Hellinger geometric discord obeys Axiom (v) for $n_A=2$ and $n_A=3$, as stated in
Table~\ref{tab3}, and we discuss the same issue for the other
geometric measures of quantum discord.

Let us first focus on the discord of response $D^R$ for the Bures and
Hellinger distances. For such distances  it is obvious from the definitions,
Eqs.~(\ref{eq-Burea_dist}),~(\ref{eq-Q_Hellinger_distance}), and~(\ref{def:quantumn}), that $D^R \leq 1$. Furthermore, $D^R (\rho) =1$ if and only if for any local unitary operator $U_A \in \Uu_\Lambda$ with non-degenerate spectrum
$\Lambda = \{ \lambda_1, \ldots, \lambda_{n_A} \}$, one has that the Uhlmann fidelity
$F( \rho, U_A \otimes \idty \,\rho \,U_A^\dagger \otimes \idty ) = 0$
for the Bures distance and $\tr  \sqrt{\rho} \,(U_A \otimes \idty \,\rho
\,U_A^\dagger \otimes \idty )^{1/2}   = 0$ for the Hellinger
distance. Each of these two  conditions is equivalent to $\rho$
and $U_A \otimes \idty \,\rho \, U_A^\dagger \otimes \idty$ having
orthogonal supports. This is in turn equivalent to the following
statement: for any pair of eigenvectors $\ket{\Psi_k}$ and
$\ket{\Psi_l}$ of $\rho$ with nonzero  eigenvalues, it holds that
\begin{equation} \label{eq-proof_theo_axiom_v_disc_resp}
{\tr}_A ( D_{kl} U_A ) =  \bra{\Psi_l} U_A\otimes \idty  \ket{\Psi_k} = 0
\quad \forall \; U_A \in \Uu_\Lambda \; ,
\end{equation}
where we have set $D_{kl} \equiv \tr_B ( \ketbra{\Psi_k}{\Psi_l} )$.

We now argue that Eq.~(\ref{eq-proof_theo_axiom_v_disc_resp}) implies
that $D_{kl} = (\delta_{kl}/n_A ) \idty_A$. Indeed, if $A$ is a self-adjoint matrix then
\begin{equation} \label{eq-proof_theo_axiom_v_disc_respbis}
\tr ( A U ) = 0\quad \forall\; U \in \, \Uu_\Lambda
\quad  \Rightarrow \quad A = a \,\idty \text{ with $a \in {\mathbb{R}}$ \; .}
\end{equation}
To verify that the implication in
Eq.~(\ref{eq-proof_theo_axiom_v_disc_respbis}) holds true,
let us introduce a fixed \ONB   $\{ \ket{i} \}$
of $\Hh_A$.  We take $U_t = e^{-\I H t} U_0 e^{\I Ht}$ with
$U_0 = \sum_i \lambda_i \ketbra{i}{i}$ and  $H$ self-adjoint. Then
$U_t \in\Uu_\Lambda$
for any real $t$. Let $A$ be such that $\tr ( A U ) = 0$ for all $U \in \, \Uu_\Lambda$.
Differentiating $\tr (A U_t)=0$ with
respect to $t$ and setting $t=0$, one obtains $\sum_i \lambda_i
\bra{i} [ H, A] \ket{i} = 0$. Choosing
$H= e^{-\I \theta} \ketbra{i_0}{j_0} + e^{\I \theta} \ketbra{j_0}{i_0}$, this yields
\begin{equation}
( \lambda_{i_0} - \lambda_{j_0} ) {\rm Im} \{ e^{\I \theta} \bra{i_0} A \ket{j_0} \} = 0 \; .
\end{equation}
Hence, in view of the non-degeneracy assumption on the spectrum,  one
has $\bra{i_0} A \ket{j_0} = 0$ if $i_0 \not= j_0$, so that $A$ is diagonal in the basis $\{ \ket{i}\}$. This basis being
arbitrary, one concludes that $A$ is proportional to the identity
operator. Thus Eq.~(\ref{eq-proof_theo_axiom_v_disc_respbis}) holds
true. Thanks to this property and since $D_{kk}$ is self-adjoint and
has trace one, Eq.~(\ref{eq-proof_theo_axiom_v_disc_resp}) yields
\begin{equation} \label{eq-D_kkisunity}
D_{kk} = {\tr}_B  (\ketbra{\Psi_k}{\Psi_k}) = n_A^{-1} \idty_A\;.
\end{equation}
 Similarly, $(D_{kl}+ D_{lk})/2$ and
$(D_{kl}- D_{lk})/2\I$ are self-adjoint and have vanishing traces for
$k \not= l$, so that according to
 Eqs.~(\ref{eq-proof_theo_axiom_v_disc_resp}) and~(\ref{eq-proof_theo_axiom_v_disc_respbis})
one has $D_{kl} = 0$  for $k \not= l$.

One deduces from Eq.~(\ref{eq-D_kkisunity}) that $\ket{\Psi_k}$ is a maximally entangled pure state,
\ie, it has the form given in Eq.~(\ref{eq-max_entangled_state}) (this follows by
observing that the Schmidt coefficients of $\ket{\Psi_k}$  are the
eigenvalues of the reduced state $D_{kk}$). For $k\not= l$, the
identity $D_{kl} = 0$ is then
equivalent to $\braket{\chi_{lj}}{\chi_{ki}} = 0$ for all
$i,j=1,\ldots, n_A$. As a result, $\rho = \sum_k p_k
\ketbra{\Psi_k}{\Psi_k}$ is a convex combination of some maximally
entangled states $\ket{\Psi_k}$ given by
Eq.~(\ref{eq-max_entangled_state}) and satisfying the orthogonality
conditions stated after this equation. As emphasized in
Section~\ref{sec-axiom(v)}, any maximally entangled state is
given by such a convex combination.
We have thus proven that for
the Bures and Hellinger distances, $D^R(\rho) =
1$ implies that $\rho$ is a maximally entangled state. 

By a  similar
reasoning, the converse statement  is
also true  provided that the eigenvalues $\lambda_i\in
\Lambda$
of the family  of unitary operators  $\Uu_\Lambda$ in the definition of the discord of response
satisfy $\sum_i \lambda_i =0$. This is in particular the case when
the $\lambda_i$ are the roots of unity, as considered in this paper.

To show that $D^R$ satisfies Axiom (v) also for the trace distance, we
make use of the trivial bound of
Eq.~(\ref{eq-trivial_bound_disc_response}). Accordingly,
$D_{\tr}^R (\rho ) \leq 1$ for any bipartite state $\rho$, with equality holding if and only if $D^R_\Bu
(\rho ) = 1$. It has been proven above that this is equivalent to
$\rho$ being maximally entangled, hence the result.

Let us finally discuss the case of the Hilbert-Schmidt distance.
From Eq.~(\ref{eq-HS_disc_resp_app-A}) one gets
\begin{equation}
D_\HS^R (\rho)
 = \tr \rho^2 - \max_{U_A \in \Uu_\Lambda}
 \tr  \rho  \, U_A \otimes \idty \,\rho \,U_A^\dagger \otimes \idty
 \leq 1\; .
\end{equation}
This inequality is saturated if and only if the two following
conditions are satisfied: (a) $\tr \rho^2 =1$ and (b)
$\rho$ and  $ U_A \otimes \idty \,\rho \,U_A^\dagger \otimes \idty$
have orthogonal supports. Thus, by the above arguments, $D^R_\HS(\rho)=1$ if
and only if $\rho$ is a maximally entangled pure state. However, if the space
dimensions of the two subsystems $A$ and $B$ are such that $n_B \geq 2 n_A$, one
can find maximally entangled states $\rho$ of $AB$ which are convex combinations of
two orthogonal maximally entangled pure states given by Eq.~(\ref{eq-max_entangled_state}).
Such states have purities $\tr \rho^2 <1$ and consequently
$D_\HS^R ( \rho) <1$. Therefore, the Hilbert-Schmidt discord of
response does not satisfy Axiom (v).
The
proof  of Theorem~\ref{theo-axiom_v_disc_resp} is now complete.
\finpro

\vspace{2mm}

We have established that the discord of response satisfies Axiom
(v) for the Bures, Hellinger, and trace distances. Let us now study
whether such Axiom holds true as well for the geometric discord and the
measurement-induced geometric discord. For the Bures distance, it is
already known that the answer is positive for the geometric discord $D_\Bu^G$
(see Ref.~\cite{Spehner2013}). Moreover,
Theorem~\ref{theo-Bures_meas_ind_disc_pure_states} above implies that this
is also the case for $D_\Bu^M$. For the Hellinger distance, one finds with the help of 
the bound $D_\Hel^R \leq 1$ and the fundamental relations~(\ref{eq-disc_resp_Hel_n=2,3})   that the geometric discord $D^G_\Hel$ is
bounded from above by $D_\mmax= 2 - 2/\sqrt{n_A}$ when $n_A=2$ and
$n_A=3$, with equality $D^G_\Hel(\rho)= D_\mmax$ holding if and only if
$D_\Hel^R (\rho)= 1$. Thanks to Theorem~\ref{theo-axiom_v_disc_resp}, it
follows   that $D^G_\Hel$ satisfies Axiom
(v) for  $n_A=2$ and $n_A=3$. We believe but so far have not been able
to prove that this remains true for a reference subsystem $A$
with higher-dimensional Hilbert spaces, $n_A>3$. Whether $D^M_\Hel$,
$D^G_{\tr}$, and $D^M_{\tr}$ obey Axiom (v) are other  open
issues. For the last two measures, we only know so far that the
answer is affirmative when $A$ is a qubit,  since then
$D^G_{\tr}= D^M_{\tr}= D_{\tr}^R$ by Theorem~\ref{qtmeqdisc}.
The bound conjectured in Eq.~(\ref{eq-conjecture_bound}) and the
fact that $D^R_{\tr}$ obeys Axiom (v) also yields an affirmative
answer for $D^M_\Hel$ if both $A$ and $B$ are qubits.
All these results are summarized in Table~\ref{tab3}.

\section{Maximally quantum-correlated two-qubit states with a fixed purity} \label{app-max_Q_corr_states}
\subsection{Two-qubit states with a fixed purity maximizing the trace discord of response} \label{app-proof_prop_trmax}

We derive in this Appendix the maximal value of the trace discord of
response $D_{\tr}^{R}(\rho)$ for two-qubit states $\rho$ with a
fixed purity $P$, which is given by Eq.~(\ref{eq-maximal_disc_resp_for_purity_P}). Our calculation is based on the following
conjecture on the most discordant states for
$D_{\tr}^{R}$:

\begin{conjecture}\label{conj}
{\rm Among the two-qubit states $\rho$ with global state purity $P$, those with maximum trace discord of response are
\begin{equation} \label{state1}
\rho^{P<3/8}_\mmax
 = \frac{1}{4}\Big((1+\sqrt{8P-2})\ket{\Phi_+}\bra{\Phi_+}
+ (1-\sqrt{8P-2})\ket{\Phi_-}\bra{\Phi_-}\Big. +\ket{\Psi_+}\bra{\Psi_+}+\ket{\Psi_-}\bra{\Psi_-}\Big)
\end{equation}
for $\frac{1}{4}\leq P \leq \frac{3}{8}$, and
\begin{equation} \label{state2}
\rho^{P>3/8}_\mmax = \frac{1}{6}\Big((2-\sqrt{6P-2}) ( \ket{\Phi_+}\bra{\Phi_+}
+ \ket{\Phi_-}\bra{\Phi_-}) +(2+2\sqrt{6P-2})\ket{\Psi_+}\bra{\Psi_+}\Big) 
\end{equation}
for $\frac{3}{8}< P \leq 1$. Here,
$\ket{\Psi_\pm}= (\ket{01} \pm \ket{10})/\sqrt{2}$
and $\ket{\Phi_\pm}= (\ket{0 0 }\pm \ket{11})/\sqrt{2}$
refer to the four Bell states.}
\end{conjecture}

This conjecture relies on a thorough numerical analysis using randomly
generated states, as described in Sec.~\ref{mqcs}.
In what follows we determine the values of $D_{\tr}^R ( \rho^{P}_\mmax)$
as a function of $P$.
Since $A$ is a qubit, the unitaries $U_A$ acting on $\Hh_A \simeq
\complex^2$ with  spectrum $\Lambda=\{-1,1\}$ can be decomposed in
terms of the three Pauli matrices $\sigma_x$, $\sigma_y$, and $\sigma_z$ as
\begin{equation}
U_A = \sin \theta  \cos \phi \,\sigma_x+\sin \theta \sin \phi\,\sigma_y+\cos \theta \,\sigma_z \; ,
\label{unitaryparam}
\end{equation}
with some arbitrary angles $\theta \in [0,\pi ]$ and $\phi \in [0,2\pi[$.
We will show that the trace distance between $\rho^{P}_{\mmax}$ and
the unitarily perturbed state $U_A \otimes \idty \, \rho^{P }_{\mmax}
U_A^{\dagger}\otimes \idty$ does not depend on $U_A$. Recall that $d_{\tr} (\rho, \sigma)$ is
equal to the sum of the moduli of the eigenvalues of
$\rho-\sigma$. For the density matrix $\rho^{P<3/8}_{\mmax}$ given by Eq.~(\ref{state1}),
the corresponding eigenvalues are $\pm \frac{1}{4}\sqrt{8P-2}(1+\cos
\theta)$ and $\pm \frac{1}{4}\sqrt{8P-2}(1-\cos \theta )$. The sum of
their moduli does not depend on $(\theta, \phi)$, that is, this sum is
independent of $U_A$. The maximum trace discord of response for states
with purity $P \leq 3/8$ reads
\begin{equation}
D_{\tr}^{R}(\rho^{P<3/8}_\mmax )=\frac{1}{4}\big\|\rho^{P<3/8}_{\mmax} -U_A \otimes \idty \, \rho^{P<3/8}_{\mmax}\, U_A^{\dagger}\otimes \idty \big\|_{\tr}^2=2P-\frac{1}{2} \; .
\end{equation}
For the state $\rho^{P>3/8}_\mmax$ given by Eq.~(\ref{state2}), the
corresponding eigenvalues are
$\pm \frac{1}{4}(2-\beta+f(\beta,\theta))$ and $\pm \frac{1}{4}(2-\beta-f(\beta,\theta ))$, where we have set $\beta=\frac{2}{3}(2-\sqrt{6P-2})$ and
\begin{equation}
f(\beta,\theta) = \frac{1}{\sqrt{2}}\sqrt{8 + \beta^2(5-3\cos{(2\theta)}) - 4\beta(3 - \cos{(2\theta)})} \; .
\end{equation}
Once again, the sum of the moduli of these eigenvalues does not depend
on the angles $\theta$ and $\phi$ of the unitary matrix. Therefore,
the corresponding maximum trace discord of response  reads
\begin{equation}
D_{\tr}^{R} (\rho^{P>3/8}_\mmax )
  =\frac{1}{9} \big(1+\sqrt{6P-2} \big)^2 \; .
\end{equation}

\subsection{Two-qubit states with a fixed purity maximizing the Hellinger discord of response} \label{app_hell_max}

We now study the same problem as before but for the Hellinger
discord of response $D_\Hel^R$.
Let us remark that for two-qubit states the analytical expression of
$D_\Hel^R(\rho)$ is given by
Eq.~(\ref{eq-closed_formula_discresp_Hel}). The maximally
quantum-correlated states according to $D^R_{\Hel}$ are found by a
thorough numerical investigation, as described in Section~\ref{mqcs}.
In the range of values $1/4\leq P\leq 1/3$ of the
global state purity $P$, these states are
the Werner states defined in Eq.~(\ref{eq-Werner_state}) with
parameter $1\leq a_+ \leq 4/3$. Note that such Werner states also
maximize the Hilbert-Schmidt discord of
response $D_\HS^R$. Hence in the range $1/4\leq P\leq 1/3$, corresponding to  region I in Fig.~\ref{bup}(c),
the maximal  Hellinger  discord of response reads
\begin{equation}
D^R_{\Hel}(\rho_W)
= \frac{1}{6} \left( -\sqrt{12 P - 3} - \sqrt{6} \sqrt{-6
  P - \sqrt{12 P - 3} + 3} +3\!\right) \; .
\end{equation}

In the range $1/3 \leq P \leq 0.503$, corresponding to
region II in Fig.~\ref{bup}(c),
we find numerically that the maximally quantum-correlated states with respect to
$D_\Hel^R$ belong to the
following family of matrices of rank 3:
\begin{equation}
\rho_{\mmax}^{1/3 \leq P \leq 0.503}  = \begin{pmatrix}
\!\frac{1}{2}\!+\!a\!-\!b\!&0&0&0\\
0&\!2b\cos^2\phi\!&\!2b\cos\phi\sin\phi\!&0\\
0&\!2b\cos\phi\sin\phi\!&\!2b\sin^2\phi\!&0\\
0&0&0&\!\frac{1}{2}\!-\!a\!-\!b\!
\end{pmatrix} \; ,
\label{formpure}
\end{equation}
where the condition of fixed purity yields $a=\frac{1}{2} \sqrt{-12
  b^2+4 b+2 P-1}$. The optimal values of $b$ and $\phi$ are found by
numerical maximization of $D^R_{\Hel}$ for these states.

For global state purities  $P > 0.503$, corresponding to region III in
Fig.~\ref{bup}(c), we find the following maximally quantum-correlated
states of rank 2:
\begin{equation}
\rho_{\mmax}^{P > 0.503} = \begin{pmatrix}
\!1\!-\!2b\!&0&0&0\\
0&\!2b\cos^2\phi\!&\!2b\cos\phi\sin\phi\!&0\\
0&\!2b\cos\phi\sin\phi\!&\!2b\sin^2\phi\!&0\\
0&0&0&0
\end{pmatrix} \; .
\label{formrank2}
\end{equation}
The condition of fixed purity enable us to eliminate one variable:
$b=\frac{1}{4} \left(\sqrt{2 P-1}+1\right)$. 
For a given purity $P$, the parameter $\phi$ for which $D^R_{\Hel}$ achieves
its maximum reads
\begin{equation}
\cos{\phi}=\!\frac{-1}{2 \sqrt{2\!-\!2 P}}\Bigg(\!-\!4 P\!-\!\sqrt{2\!-\!2 P}\!+\!\sqrt{\!-4 P^2\!+\!6 P\!-\!2}\!+\!4\Bigg.
+2\sqrt{(P\!-\!1) \left(3 P\!-\!2 \sqrt{2\!-\!2 P}\!-\!3 \sqrt{2 P\!-\!1}\!+\!2 \sqrt{\!-\!4 P^2\!+\!6 P\!-\!2}\right)}\Bigg.\Bigg)^{1/2}.
\end{equation}
The corresponding maximum Hellinger discord of response as a function
of $P$, determined with the help of Eq.~(\ref{eq-closed_formula_discresp_Hel}),  is shown in Fig.~\ref{bup}(c).

\section{Hilbert-Schmidt geometric discord for qubit-qudit states with maximally mixed marginals}
\label{app-lower_bound_disc_resp_HS}

In this Appendix we show that the Hilbert-Schmidt geometric discord is
equal to its lower bound in Eq.~(\ref{eq-inka}) if the subsystem $A$
is a qubit and the state $\rho$ has maximally mixed
marginals. Although $D^G_\HS$ is not a proper measure of quantum
correlations, it can play a useful role since it provides
bounds on the other geometric measures (e.g.,
the trace geometric discord satisfies $2 D_\HS^G (\rho) \leq D_{\tr}^G
(\rho) \leq 2 n_B D_\HS^G (\rho)$, see Theorem~\ref{qtmeqdisc} and
Eqs.~(\ref{eq-trivial_bounds}) and (\ref{drna})). Moreover, $D^G_\HS$ gives the
value of the Hellinger geometric discord by taking the
square root of the state $\rho$  (see Theorem~\ref{theo-rel_geo_disc_Hel_HS}).

\begin{proposition}\label{prop99}
Let $A$ be a qubit and $B$ a qudit with Hilbert space of arbitrary
finite dimension $n_B$. If the global state  $\rho$ of $AB$ has maximally
mixed marginals $\rho_A=\idty_A /2$ and $\rho_B=\idty_B /n_B$, then the Hilbert-Schmidt
geometric discord of $\rho$ is equal to the sum of the two smallest squared
singular values $\mu_3$ and $\mu_4$ of the reshuffled  density matrix $\rho^{\c R}$:
\begin{eqnarray}
D^G_\HS (\rho) =  \mu_3+\mu_4 \; .
\label{inka}
\end{eqnarray}
\end{proposition}

This Proposition enable us to calculate quite easily the Hilbert-Schmidt
geometric discord for a wide class of qubit-qudit states with
maximally mixed partial states, such as the states known from the theory of quantum channels
as the renormalized dynamical matrices of bistochastic CPTP
maps~\cite{Bengtsson}. The advantage of expressing the
geometric discord in terms of the singular values of $\rho^{\c R}$ is that we can do that
in an arbitrary basis.
As other examples of states which satisfy the
conditions of the Proposition, let us mention the Werner-like rotationally invariant states
defined and analyzed in Ref.~\cite{Chruscinski2007} (see also
Ref.~\cite{Breuer2005}) and the qubit-qudit states given by the r.h.s. of
Eq.~(\ref{eq-Boch_dec_square_root})
with $\vec{x} = \vec{y}= \vec{0}$, $t_0 = 1/\sqrt{2 n_B}$, and
$t_{mp}\geq 0$, $\sum_{mp}t_{mp}=\sqrt{2 n_B}$.

The proof of Proposition~\ref{prop99} relies on the following result
on the Hilbert-Schmidt geometric discord for qubit-qudit systems
(see e.g. Ref.~\cite{Gharibian2012}, Theorem 2). An arbitrary
qubit-qudit density matrix $\rho$ can always be written in the so-called Fano form, namely:
\begin{equation} \label{eq-Fano_decomposition}
\rho=\sum_{m=0}^{3}\sum_{p=0}^{n_B^2-1}M_{m p }\hat{\sigma}_m
\otimes \hat{\gamma}_p \; ,
\end{equation}
where $\hat{\sigma}_0\equiv\idty_A/\sqrt{2}$,
$\hat{\sigma}_1, \hat{\sigma}_2$, and $\hat{\sigma}_3$
are the Pauli matrices renormalized in such a way that
$\tr ( \hat{\sigma}_m \hat{\sigma}_{n}) =\delta_{m n}$,
and $\hat{\gamma}_p$ are Hermitian matrices forming an
orthonormal basis for the space of $n_B \times n_B$ matrices
(\ie , $\tr ( \hat{\gamma}_p \hat{\gamma}_{q}) = \delta_{p  q}$) with
$\hat{\gamma}_0= \idty_B/\sqrt{n_B}$.
The components $M_{mp}$ of $\rho$ in
Eq.~(\ref{eq-Fano_decomposition}) are given by $M_{mp} = \tr (
\rho\,\hat{\sigma}_m \otimes\hat{\gamma}_p)$ and form a $4 \times n_B^2$ real
matrix $M$ (covariance matrix). Define the $3\times n_B^2$  matrix
$\widetilde{M}$ obtained from $M$ by removing the first row, that is,
$\widetilde{M} =[M_{mp}]_{m=1,...,3, p=0,...,n_B^2-1}$. It can be
shown that $D_\HS^G ( \rho)$ is equal to the sum of the two  smallest
squared singular values of $\widetilde{M}$, that is,
\begin{equation} \label{eq-geo_disc_2_n_B_states}
D_\HS^G ( \rho)= {\rm sv}_2^2(\widetilde{M})+{\rm sv}_3^2(\widetilde{M}) \; ,
\end{equation}
where ${\rm sv}_1^2(\widetilde{M}) \geq  {\rm sv}_2^2(\widetilde{M})
\geq {\rm sv}_3^2(\widetilde{M}) $ denote the eigenvalues 
of the $3 \times 3$ non-negative matrix $\widetilde{M} (\widetilde{M})^\dagger$. 

Before proceeding to the proof of Proposition~\ref{prop99}, let us start with some technical considerations on vectorization. This operation transforms a $n \times n$ matrix $Y$ into the vector $\ket{Y}$ obtained by ordering the matrix entries into a one-column vector, that is, $\ket{Y}$ has $n^2$ components given by
\begin{equation}
\<i,j|Y\> \equiv \<i|Y|j\> \equiv\<\overline{Y}|i,j\> \; ,
\label{vct}
\end{equation}
where the bar denotes complex conjugation in the standard basis $\{ \ket{i,j} \}$.
Let  $X$, $Y$, and $Z$ be matrices of sizes 
 $n_A n_B  \times n_A n_B$, $n_A\times n_A$, and $n_B \times n_B$,
respectively.  The following chain of identities will be useful:
\begin{equation} \label{tensorvect}
\tr{X (Y\otimes Z)}=  \sum_{ijkl}\<i,k|X|j,l\>\<j |Y|i\>\<l|Z|k \> = \sum_{ijkl}\<i,j|X^{\c R}|k,l\>\<Y^{\dagger}|i,j\>\<k,l|Z^{\rm T}\>=\<Y^{\dagger}|X^{\c R}|Z^{\rm T}\> \; ,
\end{equation}
where $\Rr$ is the reshuffling operation defined in Eq.~(\ref{reshuf})
and ${\rm T}$ denotes the transposition in the standard basis. Now let us introduce the \ONBs
$\{ \ket{\sigma_m}\}_{m=0}^3$ of $\complex^4$ and
$\{ \ket{\gamma_p} \}_{p=0}^{n_B^2-1}$ of $\complex^{n_B^2}$
obtained from the vectorization of the matrices
$\hat{\sigma}_m$ and $\hat{\gamma}_p$  appearing in the
decomposition~(\ref{eq-Fano_decomposition}).
On account of Eq.~(\ref{tensorvect}), the matrix $M$ in this
decomposition coincides with the reshuffled density matrix $\rho^{\c
  R}$ in the vectorized bases $\{ \ket{\sigma_m}\}_{m=0}^3$ and
$\{ \ket{\overline{\gamma_p}} \}_{p=0}^{n_B^2-1}$:
\begin{equation}
M_{mp}= \tr (\rho\,\hat{\sigma}_m \otimes\hat{\gamma}_p ) = \<\sigma_m | \rho^{\c R}| \overline{\gamma_p} \> \; .
\label{mvect}
\end{equation}
The proof of  Proposition~\ref{prop99} uses the following lemma.

\begin{lemma}\label{lem100}
Let  $\rho$ be a state of $AB$ satisfying the assumptions of Proposition~\ref{prop99}.
Then the largest singular value of $\rho^{\c R}$ is equal to $1/\sqrt{2n_B}$.
\end{lemma}

\begin{proof}
We use the formal correspondence (via the Jamio{\l}kowski isomorphism,
see e.g. Ref.~\cite{Bengtsson}) of the states with maximally mixed
marginals with quantum channels (\ie , CPTP maps) that preserve the
maximally mixed state.
Let us introduce the $n_A \times n_B$ matrices
 $K_{\alpha}$ with matrix elements
\begin{equation}
\bra{i_A}K_{\alpha}\ket{k_B}\equiv\sqrt{p_{\alpha}}\bra{i_A, k _B} \Psi_{\alpha}\> \; ,
\end{equation}
where $p_{\alpha}$ and $\ket{\Psi_{\alpha}}$ are, respectively, the
eigenvalues and eigenvectors of $\rho$.
From Eqs.~(\ref{reshuf}) and~(\ref{vct}) it can be immediately derived that
\begin{equation}
\rho^{\c R}=\sum_{\alpha}K_{\alpha}\otimes \overline{K_{\alpha}}
\quad , \quad 
\big( \rho^{\c R}\big)^\dagger =\sum_{\alpha}K_{\alpha}^\dagger \otimes \overline{K_{\alpha}^\dagger}
\label{superoperator}
\end{equation}
and
\begin{equation}
{\tr}_A (\rho) = \sum_{\alpha}  \overline{K_{\alpha}^{\dagger} K_{\alpha}} =\frac{1}{n_B} \,\idty_B\quad , \quad
{\tr}_B ( \rho) = \sum_{\alpha}  K_{\alpha} K_{\alpha}^{\dagger} =\frac{1}{2}\,\idty_A \; ,
\end{equation}
where the bar denotes the complex conjugation in the standard basis, \ie,
$\bra{i_A} \overline{K_\alpha} \ket{k_B} = \bra{i_A} K_\alpha \ket{k_B}^\ast$.
The above relations show that the operators $n_B
K^{\dagger}_{\alpha}K_{\alpha}$ satisfy the completeness condition
$\sum_{\alpha}n_BK^{\dagger}_{\alpha}K_{\alpha}=\idty_B$. Therefore,
$\sqrt{n_B}K_{\alpha}$ can be interpreted as the Kraus operators of
some quantum operation. Also, by
Eq.~(\ref{superoperator}), $\rho^{\c R}$ is proportional to the superoperator form of
a quantum operation, as defined by Eq.~(\ref{eq-superoperator_form}). The same considerations apply to
$\left( \rho^{\c R}\right)^{\dagger}$, which is proportional to the
superoperator form of the quantum operation defined by the Kraus
operators $L_{\alpha}\equiv \sqrt{2}K_{\alpha}^{\dagger}$. Since the composition of two quantum operations is still a
quantum operation~\cite{Roga2008}, $2n_B \rho^{\c R}\left(
\rho^{\c R}\right)^{\dagger}$ is thus a CPTP map (in its superoperator
form), which moreover preserves the identity. From the quantum
analogue of the Frobenius-Perron theorem~\cite{Bruzda2009}, one
concludes that the leading eigenvalue of such a map is equal to one
(see also Theorem 1 of Ref.~\cite{Roga2013}). Therefore, the largest
singular value of $\rho^{\c R}$ is  equal to $1/\sqrt{2n_B}$.
\end{proof}

\vspace{1mm}

\Proofof{Proposition~\ref{prop99}}
Since the partial traces of $\rho$ are by assumption proportional to
$\idty$, we have $\tr \rho \,\hat{\sigma}_m\otimes \idty_B =\tr
\hat{\sigma}_m /2 = 0$ for
$m=1,...,3$, and similarly
$\tr  \rho \, \idty_A \otimes {\hat{\gamma}_p}  = {\tr \hat{\gamma}_p}/n_B= 0$ for $p=1,...,n_B^2-1$. Hence, in the first
row and first column of the matrix $M$,  only one entry is different
from $0$, namely $M_{00}=\tr  \rho\,\idty_A\otimes\idty_B  /\sqrt{2
  n_B} =1/\sqrt{2n_B}$, which is therefore a singular value of
$M$. Analogously, the matrix $\widetilde{M}$ has only zeros in its
first column. Thus $M$ and $\widetilde{M}$ have the same
singular values, excepted for the additional singular value
$\mu_1=1/\sqrt{2n_B}$ of  $M$. According to Eq.~(\ref{mvect}) and
Lemma~\ref{lem100}, $\rho^{\Rr}$ has the same singular values as $M$
and its largest singular value is
$\mu_1$. It thus follows from
Eq.~(\ref{eq-geo_disc_2_n_B_states}) that $D^G_\HS ( \rho)$ is the sum
of the two smallest  squared singular values of $\rho^{\Rr}$.
\finpro



\begin{thebibliography}{99}

\bibitem{Zurek2000} W. H. Zurek, Annalen der Physik {\bf 9}, 855 (2000).

\bibitem{Ollivier2001} H. Ollivier and W. H. Zurek, Phys. Rev. Lett. {\bf 88}, 017901 (2001).

\bibitem{Henderson2001} L. Henderson and V. Vedral, J. Phys. A: Math. Gen. {\bf 34}, 6899 (2001).

\bibitem{Girolami2011} D. Girolami and G. Adesso, Phys. Rev. A {\bf 83}, 052108 (2011).

\bibitem{Modi_review} K. Modi, A. Brodutch, H. Cable, T. Paterek, and V. Vedral,
Rev. Mod. Phys. {\bf 84}, 1655 (2012).

\bibitem{Huang2013} Y. Huang, New J. Phys. {\bf 16}, 033027 (2014).

\bibitem{Daki'c2010} B. Daki{\'c}, V. Vedral, and {\v C}. Brukner, Phys. Rev. Lett. {\bf 105}, 190502 (2010).

\bibitem{Aaronson2013}  B. Aaronson, R. Lo Franco, and G. Adesso, Phys. Rev. A {\bf 88}, 012120 (2013).

\bibitem{Spehner2013} D. Spehner and M. Orszag, New J. Phys. {\bf 15}, 103001 (2013).

\bibitem{Spehner2014} D. Spehner and M. Orszag, J. Phys. A: Math. Theor. {\bf 47}, 035302 (2014).

\bibitem{Roga2014} W. Roga, S. M. Giampaolo, and F. Illuminati, J. Phys A: Math. Theor. {\bf 47}, 365301 (2014).

\bibitem{Modi2010} K. Modi, T. Paterek, W. Son, V. Vedral, and M. Williamson,  Phys. Rev. Lett. {\bf 104}, 080501 (2010).

\bibitem{Giampaolo2013} S. M. Giampaolo, A. Streltsov, W. Roga, D. Bru{\ss}, and F. Illuminati, Phys. Rev. A {\bf 87}, 012313 (2013).

\bibitem{Ciccarello2014} F. Ciccarello, T. Tufarelli, and V. Giovannetti, New J. Phys. {\bf 16}, 013038 (2014).

\bibitem{Spehner_review} D. Spehner, J. Math. Phys. {\bf 55}, 075211 (2014).

\bibitem{Piani2014}   M. Piani, V. Narasimhachar, and J. Calsamiglia, New J. Phys. {\bf 16}, 113001 (2014)

\bibitem{Luo2013} L. Chang and S. Luo, Phys. Rev. A {\bf 87}, 062303 (2013).

\bibitem{Uhlmann1976} A. Uhlmann, Rep. Math. Phys. {\bf 9}, 273 (1976).

\bibitem{Bures1969} D. Bures, Trans. Am. Math. Soc. {\bf 135}, 199 (1969).

\bibitem{Luo2010} S. Luo and S. Fu, Phys. Rev. A {\bf 82}, 034302
  (2010).

\bibitem{Nakano2013} T. Nakano, M. Piani, and G. Adesso, Phys. Rev. A {\bf 88}, 012117 (2013).

\bibitem{Streltsov2011c} A. Streltsov, H. Kampermann, and D. Bru{\ss}, Phys. Rev. Lett. {\bf 107}, 170502 (2011).

\bibitem{Abad2012} T. Abad, V. Karimipour, and L. Memarzadeh, Phys. Rev. A {\bf 86}, 062316 (2012).

\bibitem{Gharibian2012} S. Gharibian, Phys. Rev. A {\bf 86}, 042106 (2012).

\bibitem{Wilde15} K.P. Seshadreesan and M.M. Wilde, Phys. Rev. A {\bf 92}, 042321 (2015)

\bibitem{Pirandola2011} S. Pirandola, C. Lupo, V. Giovannetti, S. Mancini, and S. L. Braunstein, New J. Phys. {\bf 13}, 113012 (2011).

\bibitem{Tan2008}   S.-H. Tan, B. I. Erkmen, V. Giovannetti, S. Guha,
  S. Lloyd, L. Maccone, S. Pirandola, and J.H. Shapiro,
  Phys. Rev. Lett. {\bf 101}, 253601 (2008)
  
\bibitem{Farace2014} A. Farace, A. De Pasquale, L. Rigovacca, and V. Giovanetti, New J. Phys. {\bf 16}, 073010 (2014).

\bibitem{Girolami2014} D. Girolami, A. M. Souza, V. Giovannetti, T. Tufarelli, J. G. Filgueiras, R. S. Sarthour, D. O. Soares-Pinto, I. S. Oliveira, and
G. Adesso, Phys. Rev. Lett. {\bf 112}, 210401 (2014).

\bibitem{Correa2013} L. A. Correa, J. P. Palao, D. Alonso, and G. Adesso, Sci. Rep. {\bf 4}, 3949 (2014).

\bibitem{Girolami2013} D. Girolami, T. Tufarelli, and G. Adesso, Phys. Rev. Lett. {\bf 110}, 240402 (2013).

\bibitem{Roga2015} W. Roga, D. Buono, and F. Illuminati, New J. Phys. {\bf 17}, 013031 (2015).

\bibitem{Streltsov12} A. Streltsov, G. Adesso, M. Piani, and
  D. Bru{\ss}, Phys. Rev. Lett. {\bf 109}, 050503 (2012).
 
\bibitem{Piani2012} M. Piani, Phys. Rev. A {\bf 86}, 034101 (2012).

\bibitem{Perez-Garcia2006} D. Perez-Garcia, M. M. Wolf, D. Petz, and M. B. Ruskai, J. Math. Phys. {\bf 47}, 083506 (2006).

\bibitem{Ozawa2000} M. Ozawa, Phys. Lett. A {\bf 268}, 158 (2000).

\bibitem{Bengtsson} I. Bengtsson and K. $\dot{\rm Z}$yczkowski, {\it
  Geometry of Quantum States: An Introduction to Quantum
  Entanglement}, Cambridge University Press, Cambridge (UK), 2006.

\bibitem{Bennett96} C. H. Bennett, D. P. DiVincenzo, J. A. Smolin, and W. K. Wootters, Phys. Rev. A {\bf 54}, 3824 (1996).

\bibitem{Datta2008} A. Datta, A. Shaji, and C. M. Caves,
  Phys. Rev. Lett. {\bf 100}, 050502 (2008).

\bibitem{Wigner1963} E. P. Wigner and M. M. Yanase, Proc. Natl. Acad. Sci. U.S.A. {\bf 49}, 910 (1963).

\bibitem{Luo2003} S. Luo, Phys. Rev. Lett. {\bf 91}, 180403 (2003).

\bibitem{Audenaert2008} K. Audenaert, L. Cai, and F. Hansen, Lett. Math. Phys. {\bf 85}, 135 (2008).

\bibitem{Nielsen} M. A. Nielsen and I. L. Chuang, {\it Quantum Computation and Information}, Cambridge University Press, Cambridge (UK), 2000.

\bibitem{Petz1996} D. Petz, Lin. Alg. and its Appl. {\bf 244}, 81 (1996).

\bibitem{Fuchs1999} C. A. Fuchs and J. van de Graaf, IEEE Trans. Inf. Theory {\bf 45}, 1216 (1999).

\bibitem{Audenaert2007} K. M. R. Audenaert, J. Calsamiglia,
  R. Munoz-Tapia, E. Bagan, L. Masanes, A. Acin, and F. Verstraete,  Phys. Rev. Lett. {\bf 98}, 160501 (2007).

\bibitem{Lieb1973} E. H. Lieb, Adv. Math. {\bf 11}, 267 (1973).

\bibitem{Holevo1972} A. S. Holevo, Theor. Math. Phys. {\bf 13}, 184 (1972).

\bibitem{Monras2011} A. Monras, G. Adesso, S. M. Giampaolo, G. Gualdi, G. B. Davies, and F. Illuminati, Phys. Rev. A {\bf 84}, 012301 (2011).

\bibitem{Giampaolo2007} S. M. Giampaolo and F. Illuminati, Phys. Rev. A {\bf 76}, 042301 (2007).

\bibitem{Wei03} T. C. Wei and P. M. Goldbart, Phys. Rev. A {\bf 68}, 042307 (2003).

\bibitem{Blasone08} M. Blasone, F. Dell'Anno, S. De Siena, and F. Illuminati, Phys. Rev. A {\bf 77}, 062304 (2008).

\bibitem{Streltsov10} A. Streltsov, H. Kampermann, and D. Bru\ss, New J. Phys. {\bf 12}, 123004 (2010).

\bibitem{Vidal00} G. Vidal, J. Mod. Opt. {\bf 47}, 355 (2000).

\bibitem{Helstrom1976} C. W. Helstrom, {\it Quantum Detection and Estimation Theory}, Academic Press, New York (US), 1976.

\bibitem{Paula2013} F. M. Paula, T. R. de Oliveira, and M. S. Sarandy, Phys. Rev. A {\bf 87}, 064101 (2013).

\bibitem{Marian15} P. Marian and T. A. Marian, J. Phys. A: Math. Theor. {\bf 48}, 115301 (2015).

\bibitem{Buono2015} D. Buono, W. Roga, and F. Illuminati, {\em "Quantum correlations of Gaussian states: characterization, quantification, and comparison by metrics, states, and operations"}, to appear (2015).

\bibitem{Vedral98} V. Vedral and M. B. Plenio, Phys. Rev. A {\bf 57}, 1619 (1998).

\bibitem{Luo12} S. Luo and S. Fu, Theor. Math. Phys. {\bf 171}, 870 (2012).

\bibitem{Barnum2002} H. Barnum and E. Knill, J. Math. Phys. {\bf 43}, 2097 (2002).

\bibitem{Plenio2005} M. B. Plenio, Phys. Rev. Lett. {\bf 95}, 090503 (2005).

\bibitem{Carlen} E. A. Carlen, {\it Trace inequalities and quantum entropy: An introductory  course}, in: {\it Entropy and the quantum},
R. Sims and D. Ueltschi (eds.), Contemp. Math. {\bf 529}, 73 (Amer. Math. Soc., Providence, RI, 2010).

\bibitem{Bathia} R. Bhatia, {\it Matrix Analysis}, Springer (1991).

\bibitem{Ogawa99} T. Ogawa and H. Nagaoka, IEEE Trans. Inf. Theory {\bf 45}, 2486 (1999).

\bibitem{Tyson09} J. Tyson, J. Math. Phys. {\bf 50}, 032106 (2009).

\bibitem{Bathia2000} R. Bhatia and F. Kittaneh, Lin. Alg. and its
  Appl. {\bf 318}, 109 (2000) 

\bibitem{Korbicz2012} J. Korbicz, P. Horodecki, and R. Horodecki, Phys. Rev. A {\bf 86}, 042319 (2012).

\bibitem{Rana2012} S. Rana and P. Parashar, Phys. Rev. A {\bf 85},
  024102 (2012)

\bibitem{Chen03} K. Chen and L.-A. Wu, Quant. Inf. Comp. {\bf 3}, 193 (2003).

\bibitem{HJ} R. Horn and C. R. Johnson, {\it Matrix analysis}, Cambridge University Press, 1985 (22nd printing 2009), page 449, Example 7.4.52.

\bibitem{Pirandola2011b} S. Pirandola, Phys. Rev. Lett. {\bf 106},
  090504 (2011)

\bibitem{GaussianMetrology2014} G. Adesso, Phys. Rev. A {\bf 90}, 022321 (2014).

\bibitem{Adesso2014} G. Adesso, V. D'Ambrosio, E. Nagali, M. Piani, and F. Sciarrino, Phys. Rev. Lett. {\bf 112}, 140501 (2014).

\bibitem{Verstraete2003} T.-C. Wei, K. Nemoto, P. M. Goldbart, P. G. Kwiat, W. J. Munro, and F. Verstraete,
Phys. Rev. A {\bf 67}, 022110 (2003).

\bibitem{Adesso2005} G. Adesso and F. Illuminati, Phys. Rev. A {\bf 72}, 032334 (2005).

\bibitem{Chruscinski2007} D. Chru{\'s}cinski and A. Kossakowski, Open Sys. Inf. Dyn. {\bf 14}, 25 (2007).

\bibitem{Breuer2005} H.-P. Breuer, J. Phys. A: Math. Gen. {\bf 38}, 9019 (2005).


\bibitem{Roga2008} W. Roga, M. Fannes, and K. {\.Z}yczkowski, J. Phys. A: Math. Theor. {\bf 41},  035305 (2008).

\bibitem{Bruzda2009} W. Bruzda, V. Cappellini, H.-J. Sommers, and K. {\.Z}yczkowski, Phys. Lett. A {\bf 373}, 320 (2009).

\bibitem{Roga2013} W. Roga, Z. Pucha{\l}a, {\L}. Rudnicki, and K. {\.Z}yczkowski, Phys. Rev. A {\bf 87}, 032308 (2013).

\end{thebibliography}
\end{document}